\pdfoutput=1
\RequirePackage{ifpdf}
\ifpdf 
\documentclass[pdftex]{sigma}
\else
\documentclass{sigma}
\fi

\numberwithin{equation}{section}

\newtheorem{Theorem}{Theorem}[section]
\newtheorem{Lemma}[Theorem]{Lemma}
\newtheorem{Proposition}[Theorem]{Proposition}
 { \theoremstyle{definition}
\newtheorem{Remark}[Theorem]{Remark} }

\newcommand{\rrtimes}{{\rtimes}}

\newcommand\zz{\mathbb{Z}}
\newcommand\cc{\mathbb{C}}
\newcommand\ii{\mathrm{i}}
\newcommand\bbbone{{\mathchoice \mathrm{1\mskip-4mu l} \mathrm{1\mskip-4mu l}
\mathrm{1\mskip-4.5mu l} \mathrm{1\mskip-5mu l}}}
\newcommand\one{\bbbone}
\newcommand\sph{\mathrm{sph}}
\newcommand\loc{\mathrm{loc}}
\newcommand\hol{\mathrm{hol}}
\newcommand\Hol{\mathrm{Hol}}
\newcommand\heis{\mathrm{heis}}
\newcommand\fl{\mathrm{f\/l}}
\newcommand\cf{\mathrm{cf}}
\newcommand\Cf{\mathrm{Cf}}
\newcommand\sch{\mathrm{sch}}
\newcommand\Sch{\mathrm{Sch}}
\newcommand\SSch{\mathrm{SSch}}
\newcommand\sgn{\mathrm{sgn}}
\newcommand\id{\mathrm{id}}
\newcommand\e{\mathrm{e}}
\newcommand\rmd{\mathrm{d}}

\newcommand\fg{{\mathfrak g}}
\newcommand\cH{\mathcal{H}}
\newcommand\cC{\mathcal{C}}
\newcommand\cV{\mathcal{V}}
\newcommand\cF{\mathcal{F}}
\newcommand\cS{\mathcal{S}}
\newcommand\cA{\mathcal{A}}
\newcommand\cK{\mathcal{K}}
\newcommand\cL{\mathcal{L}}
\newcommand\cY{\mathcal{Y}}
\newcommand\cU{\mathcal{U}}

\newcommand\p{ \partial}
\newcommand\W{\mathrm{W}}
\newcommand\gl{\mathrm{gl}}
\newcommand\GL{\mathrm{GL}}
\newcommand\so{\mathrm{so}}
\newcommand\SO{\mathrm{SO}}
\newcommand\vph[1]{\vphantom{(}#1\vphantom{)}}
\newcommand\dds{ \partial_s}
\newcommand\qt[1]{``#1''}
\newcommand\ee{\mathrm{e}}
\newcommand\ddr{\partial_r}
\newcommand\ddp{\partial_p}
\newcommand\pd{\partial}
\newcommand\ddw{\partial_w}
\newcommand\fbal[4]{\cF_{\alpha#1,\beta#2,\mu#3}^{{\mathrm{bal}}}(#4,\pd_{#4})}
\newcommand\fbalw{\fbal{}{}{}{w}}
\newcommand\bal{\mathrm{bal}}
\newcommand\fbalpre[4]{\cF_{#1,#2,#3}^{{\mathrm{bal}}}(#4,\pd_{#4})}
\newcommand\sbalpre[3]{\cS_{#1,#2}^{{\mathrm{bal}}}(#3,\pd_{#3})}
\newcommand\sbal[3]{\cS_{\alpha#1,\lambda#2}^{{\mathrm{bal}}}(#3,\pd_{#3})}
\newcommand\sbalw{\sbal{}{}{w}}

\newcommand\conbal[3]{\cF_{\theta#1,\alpha#2}^{{\mathrm{bal}}}(#3,\pd_{#3})}
\newcommand\conbalw{\conbal{}{}{w}}
\newcommand\hbalpre[2]{\cS_{#1}^{{\mathrm{bal}}}(#2,\pd_{#2})}
\newcommand\hbal[2]{\hbalpre{\lambda#1}{#2}}
\newcommand\hbalw{\hbal{}{w}}

\newcommand\ddu[1]{\partial_{u_{#1}}}
\newcommand\bbalpre[2]{\cF_{#1}^{\mathrm{bal}}(#2,\pd_{#2})}
\newcommand\bbal[2]{\bbalpre{\alpha#1}{#2}}
\newcommand\bbalw{\bbal{}{w}}

\begin{document}

\allowdisplaybreaks

\newcommand{\arXivNumber}{1505.02271}

\renewcommand{\PaperNumber}{108}

\FirstPageHeading

\ShortArticleName{From Conformal Group to Symmetries of Hypergeometric Type Equations}

\ArticleName{From Conformal Group to Symmetries\\ of Hypergeometric Type Equations}

\Author{Jan DEREZI\'NSKI~$^\dag$ and Przemys{\l}aw MAJEWSKI~$^{\dag\ddag}$}

\AuthorNameForHeading{J.~Derezi\'nski and P.~Majewski}

\Address{$^\dag$~Department of Mathematical Methods in Physics, Faculty of Physics, University of Warsaw,\\
\hphantom{$^\dag$}~Pasteura~5, 02-093~Warszawa, Poland}
\EmailD{\href{mailto:Jan.Derezinski@fuw.edu.pl}{Jan.Derezinski@fuw.edu.pl}}
\URLaddressD{\url{http://fuw.edu.pl/~derezins/}}

\Address{$^\ddag$~Bureau of Air Defence and Anti-missile Defence Systems, PIT-RADWAR~S.A.,\\
\hphantom{$^\ddag$}~Poligonowa~30, 04-025~Warszawa, Poland}
\EmailD{\href{mailto:Przemyslaw.Majewski@fuw.edu.pl}{Przemyslaw.Majewski@fuw.edu.pl}}
\URLaddressD{\url{http://fuw.edu.pl/~pmaj/}}

\ArticleDates{Received February 24, 2016, in f\/inal form October 20, 2016; Published online November 05, 2016}

\Abstract{We show that properties of hypergeometric type equations become transparent if they are derived from appropriate 2nd order partial dif\/ferential equations with constant coef\/f\/icients. In particular, we deduce the symmetries of the hypergeometric and Gegenbauer equation from conformal symmetries of the 4- and 3-dimensional Laplace equation. We also derive the symmetries of the conf\/luent and Hermite equation from the so-called Schr\"odinger symmetries of the heat equation in 2 and 1 dimension. Finally, we also describe how properties of the ${}_0F_1$ equation follow from the Helmholtz equation in 2 dimensions.}


\Classification{35J05; 33Cxx; 35B06}

{\small \tableofcontents}

\section{Introduction}\label{s1}

The paper is devoted to the properties of the {\em hypergeometric equation}
\begin{gather}
\big(z(1-z)\p_z^2+ (c-(a+b+1)z)\p_z-ab\big)F(z)=0,\label{hyp1}
\end{gather}
the {\em Gegenbauer equation}
\begin{gather} \big(\big(1-z^2\big)\p_z^2-({a}+{b}+1)z\p_z-{a}{b}\big)F(z)=0,\label{hyp2}
\end{gather}
the {\em confluent equation}
\begin{gather}
\big(z\p_z^2+(c-z)\p_z-a\big)F(z)=0,\label{hyp3}
\end{gather}
the {\em Hermite equation}
\begin{gather} \big(\p_z^2-2z\p_z-2{a}\big)F(z)=0,
\label{hyp4}
\end{gather}
and the {\em ${}_0F_1$ equation} (closely related to the {\em Bessel equation}, see, e.g.,~\cite{De})
\begin{gather}
\big(z\p_z^2+c\p_z-1\big)F(z)=0.\label{hyp5}
\end{gather}
Here, $z$ is a complex variable, $\partial_z$ is the dif\/ferentiation with respect to $z$, and $a$, $b$, $c$ are arbitrary complex parameters.

Special functions that solve these equations are typical representatives of {\em hypergeometric type functions} \cite{NU}. They often appear in applications \cite{Fl,MF}. In old times they were considered one of the central topics of mathematics, see, e.g.,~\cite{WW}. In our opinion, they indeed belong to the most natural objects in mathematics.

Properties of hypergeometric type functions look quite complicated. According to our observations, when these properties are discussed, most people react with boredom and/or irritation. We would like to convince the reader that in reality identities related to hypergeometric type equations are beautiful and can be derived in an elegant and transparent way.

We will show that in order to understand hypergeometric type equations it is helpful to start from certain 2nd order PDE's in several variables with constant coef\/f\/icients. If we start from rather obvious properties of these PDE's, reduce the number of variables and change the coordinates, we can observe how these properties become more complicated. At the end one obtains relatively complicated sets of identities for hypergeometric type equations.

\subsection{Hypergeometric type operators}

Equations (\ref{hyp1})--(\ref{hyp5}) are determined by an operator of the form
\begin{gather*}\cC=\cC(z,\partial_z):=\sigma(z)\p_z^2+\tau(z)\p_z+\eta.\end{gather*}
In our paper we will concentrate on the study of the operator $\cC$ itself, rather than on individual solutions $F$
of the equation \begin{gather*} \cC F=0.
\end{gather*}
Note, however, that properties of $F$'s can be to a large extent inferred from the properties of $\cC$ itself.

According to the terminology used in \cite{NU}, and then in \cite{De, DeWr}, equations (\ref{hyp1})--(\ref{hyp5}) belong to the class of {\em hypergeometric type equations}. This class is def\/ined by demanding that $\sigma$ is a~polynomial of at most 2nd order, $\tau$ is a polynomial of at most 1st order and $\eta$ is a number. More precisely, (\ref{hyp1})--(\ref{hyp5}) constitute standard forms of all nontrivial classes of hypergeometric type equations, as explained, e.g., in~\cite{De}.

Equations (\ref{hyp1})--(\ref{hyp5}) depend on a number of (complex) parameters. For instance, in the case of the hypergeometric equation these parameters are $a$, $b$, $c$. We will prefer to use dif\/ferent sets of parameters introduced in a systematic way in~\cite{De}, which are more convenient to express the symmetries of these equations. In~\cite{De} these new parameters were called the {\em Lie-algebraic parameters}. Indeed, as we will describe, they are eigenvalues of the ``Cartan algebra'' of appropriate Lie algebras. For instance, for the hypergeometric equation the Lie-algebraic parameters, denoted $\alpha$, $\beta$, $\mu$, are the dif\/ferences of the indices at the three singular points.

We will prefer not to use the operators $\cC$ directly. As explained in \cite{De}, we can write
\begin{gather*} \cC(z,\p_z)=\rho^{-1}(z)\partial_z\rho(z)\partial_z+\eta,
\end{gather*}
which def\/ines (up to a multiplicative factor) a certain function $\rho(z)$ called the {\em weight}. Fol\-lowing~\cite{De}, the operator
\begin{gather*} \cC^\bal(z,\p_z):=\rho(z)^{\frac12}\cC(z,\p_z)\rho(z)^{-\frac12},\end{gather*}
will be called the {\em balanced form of $\cC$}. The study of the balanced form is obviously equivalent to that of $\cC$, since both are related by a simple similarity transformation. The original operator $\cC$ will be sometimes called {\em the standard form of $\cC$}.

We will consider 3 classes of identities:
\begin{enumerate}\itemsep=0pt
\item[(1)] discrete symmetries,
\item[(2)] transmutation relations,
\item[(3)] factorizations.
\end{enumerate}
{\em Discrete symmetries} involve a transformation of the independent variable, together with a change of the parameters. The family of the discrete symmetries of the hypergeometric equation is especially famous. In the literature it is sometimes known under the name of the {\em Kummer's table}~\cite{AAR,Ku1,Ku2,LSV}.

The balanced form is especially convenient for a presentation of discrete symmetries, because some of them simply reduce to the change of sign of parameters.

{\em Transmutation relations} say that $\cC$ multiplied from the right by an appropriate 1st order operator equals $\cC$ for shifted parameters
multiplied from the left by a similar 1st order operator. In quantum physics the corresponding 1st order operators are often called {\em creation/annihilation operators}.

{\em Factorizations} say that $\cC$ can be written as a product of two 1st order operators, up to an additional term that does not contain the independent variable. It is easy to see that factorizations imply transmutation relations, as described, e.g., in~\cite{De}. Factorizations play an important role in quantum mechanics. They are often interpreted as the manifestation of {\em supersymmetry}~\cite{CKS}. In quantum mechanics discussion of these factorizations has a long history going back at least to~\cite{IH}.

\looseness=-1 Discrete symmetries, transmutation relations and factorizations are far from being trivial. Nevertheless, in our opinion they belong to the most {\em elementary} properties of hypergeometric type equations and functions. There exist many other properties, notably {\em addition formulas} and {\em integral representations}, which we view as more advanced. We do not consider them in our paper.

\subsection{Group-theoretical derivation}

We will see that all hypergeometric type equations can be obtained by separating the variables of 2nd order PDE's with constant coef\/f\/icients. We will always use the complex variable, to avoid discussing various signatures of these PDE's.

Every such a PDE has a Lie algebra and a Lie group of generalized symmetries. In this Lie algebra we can f\/ix a certain maximal commutative algebra, which we will call the ``{\em Cartan algebra}''. Operators whose adjoint action is diagonal in the ``Cartan algebra'' will be called ``{\em root operators}''. In the Lie group of generalized symmetries we will distinguish a discrete group, which we will call the {\em group of} ``{\em Weyl symmetries}''. This group will implement automorphisms of the Lie algebra leaving invariant the ``Cartan algebra''.

Note that in some cases the Lie algebra of generalized symmetries is semisimple, and then the names {\em Cartan algebra}, {\em root operators} and {\em Weyl symmetries} correspond to the standard names. In other cases the Lie algebra is non-semisimple, and then the names are less standard~-- this is the reason for the quotation marks that we use.

Parameters of hypergeometric type equation can be interpreted as the eigenvalues of elements of the ``Cartan algebra''. In particular, the number of parameters of a given class of equations equals the dimension of the corresponding ``Cartan algebra''. Each transmutation relation is related to a ``root operator''. Finally, each discrete symmetry of a hypergeometric type operator corresponds to a ``Weyl symmetry'' of the Lie algebra.

We can distinguish 3 kinds of PDE's of the complex variable with constant coef\/f\/icients:
\begin{enumerate}\itemsep=0pt \item[(1)] The {\em Laplace equation} on $\cc^n$
\begin{gather*}\Delta_{\cc^n}f=0,\end{gather*} whose Lie algebra of generalized symmetries is $\so(\cc^{n+2})$.
\item[(2)] The {\em heat equation} on $\cc^{n-2}\oplus\cc$,
\begin{gather*}\big(\Delta_{\cc^{n-2}}+\partial_t\big)f=0,\end{gather*}
 whose Lie algebra of generalized
 symmetries is $\sch(\cc^{n-2})$, the so-called {\em $($complex$)$ Schr\"o\-dinger Lie algebra}.
 \item[(3)] The {\em Helmholtz equation} on $\cc^{n-1}$,
\begin{gather*}\big(\Delta_{\cc^{n-1}}-1\big)f=0,\end{gather*}
 whose Lie algebra of symmetries is $\cc^{n-1}\rtimes \so(\cc^{n-1})$.
\end{enumerate}
Separating the variables in these equations usually leads to dif\/ferential equations with many variables. Only in a few cases it leads to ordinary dif\/ferential equations, which turn out to be of hypergeometric type. All these cases are described in the following table:

\begin{table}[ht]\centering
\caption{}\label{table}\vspace{1mm}
\begin{tabular}{ccccc}
\hline
PDE& $\begin{array}{@{}c@{}}\hbox{Lie}\\\hbox{algebra}\end{array}$
&$\begin{array}{@{}c@{}}\hbox{dimension of}\\\hbox{Cartan algebra}\end{array}$
&$\begin{array}{@{}c@{}}\hbox{discrete}\\\hbox{symmetries}\end{array}$&
equation \\
\hline
$\Delta_{\cc^4}$& $\so(\cc^6)$&3& cube & hypergeometric \tsep{2pt} \\
$\Delta_{\cc^3}$&$\so(\cc^5)$&2& square & Gegenbauer \\
$\Delta_{\cc^2}+\p_t$& $\sch(\cc^2)$&2&$\zz_2\times \zz_2$&
conf\/luent\\
$\Delta_{\cc}+\p_t$& $\sch(\cc^1)$&1&$\zz_4$&
Hermite \\
$\Delta_{\cc^2}-1$& $\cc^2\rtimes \so(\cc^2)$ &1& $\zz_2$& ${}_0F_1$\\
\hline
\end{tabular}
\end{table}

The Laplace equation on $\cc^n$, the heat equation on $\cc^{n-2}$ and the Helmholtz equation on $\cc^{n-1}$ together with their generalized symmetries can be elegantly derived by an appropriate reduction from the Laplace equation in $n+2$ dimensions
\begin{gather*}
\Delta_{\cc^{n+2}}K=0.
\end{gather*}
Thus, as follows from Table~\ref{table}, to derive symmetries of the hypergeometric and conf\/luent equations one should start from
\begin{gather}
\Delta_{\cc^6} K=0.\label{lapa}
\end{gather}
To derive the Gegenbauer, Hermite and ${}_0F_1$ equation together with all its symmetries it is enough to start with
\begin{gather}
\Delta_{\cc^5} K=0.\label{lapa1}
\end{gather}
It is easy to reduce the Laplace equation from 6 to 5 dimensions. Thus the Laplace equation in~6 dimensions is the ``mother'' of all hypergeometric type equations.

Let us describe these derivations in more detail.
\begin{itemize}\itemsep=0pt
\item We start from (\ref{lapa}), where the symmetries $\so(\cc^6)$ are obvious. By what we call the {\em conformal reduction}, we can reduce $\Delta_{\cc^6}$ to $\Delta_{\cc^4}$, and then further to the {\em hypergeometric operator}. Alternatively, one can reduce $\Delta_{\cc^6}$ to an appropriate Laplace--Beltrami operator, and then we obtain~(\ref{hyp1}) more directly.
\item We can repeat an analogous procedure one dimension lower. We start from (\ref{lapa1}), and at the end we obtain the {\em Gegenbauer
operator}.
\item One can reduce~(\ref{lapa}) to $\Delta_{\cc^2}+\p_t$, the heat operator in 2 dimensions, and then further to the {\em confluent operator}. Note that $\sch(\cc^2)$ is contained in $\so(\cc^6)$.
\item One can repeat the above steps one dimension lower, reducing~(\ref{lapa1}) to $\Delta_\cc+\p_t$, the heat operator in 1 dimension, and then further to the {\em Hermite operator}. Note that $\sch(\cc)$ is contained in $\so(\cc^5)$.
\item To obtain the {\em ${}_0F_1$ operator} one needs to separate variables in the Helmholtz opera\-tor~\mbox{$\Delta_{\cc^2}-1$}. Its symmetries
$\cc^2\rtimes\so(\cc^2)$ are contained in $\so(\cc^5)$ and one can start the derivation from~(\ref{lapa1}).
\end{itemize}

One can ask whether Table~\ref{table} can be enlarged, e.g., by considering $\Delta_{\cc^n}f=0$ with its conformal symmetries $\so(\cc^{n+2})$ for $n\geq5$. One can argue that the answer is negative and Table~\ref{table} is complete. Indeed, the Cartan algebra of $\so(n+2)$ has dimension $[n/2]$, and $n-[n/2]>1$ for $n\geq5$. Therefore, separation of variables in the Laplace equation in dimension $n\geq5$ leads to a~dif\/ferential equation in more than one variable.

\subsection{Organization of the paper}

The paper can be considered as a sequel to \cite{De}. Nevertheless, it is to a large degree self-contained and independent of~\cite{De}.

In Section~\ref{Hypergeometric type operators and their symmetries} we list the identities that we would like to derive/explain in our article. As~described in the introduction, these identities involve 5 classes of dif\/ferential operators (\ref{hyp1})--(\ref{hyp5}). All these operators are f\/irst transformed to the balanced form.

The versions of these identities for the standard form of equations (\ref{hyp1})--(\ref{hyp5}) can be found in \cite{De}. In order to reduce the length of the paper, in this paper we concentrate on the balanced form, which is more symmetric.

Sections~\ref{s2}, \ref{s3} and \ref{s4} provide basic def\/initions and concepts, mostly related to (complex) dif\/ferential geometry, Lie groups and Lie algebras. This material is very well known, especially in the real context. Unfortunately, the use of complex manifolds, natural in our context, has some disadvantages due to the rigidity of holomorphic functions and their multivaluedness. This is the reason for some annoying minor complications in these sections, such as {\em local} representations of groups.

In Section~\ref{Conformal invariance} we describe the action of the conformal group/Lie algebra in $n$ dimensions. We do this f\/irst for a general $n$. As a simple, but instructive exercise we consider the cases $n=1,2$.

In Section~\ref{s6} we consider the case $n=4$, which yields the hypergeometric operator. In Section~\ref{s7} we consider $n=3$, which leads to the Gegenbauer operator. These two sections are very parallel to one another. Both are direct applications of the formalism of Section~\ref{Conformal invariance}.

In Section~\ref{sec-sch} we consider the Schr\"odinger group $\Sch(\cc^{n-2})$ and its Lie algebra $\sch(\cc^{n-2})$. They describe generalized symmetries of the heat equation in $n-2$ dimensions. We f\/irst do this for a general dimension.

In Section~\ref{sec-con} we consider the case $n=4$, which yields the conf\/luent operator. In Section~\ref{sec-her} we consider $n=3$, which leads to the Hermite operator. Again, these two sections are quite parallel. They are applications of the formalism of Section~\ref{sec-sch}.

In the f\/inal Section~\ref{sec-bess} we consider the Helmholtz equation in 2 dimensions together with the af\/f\/ine Euclidean symmetries $\cc^2\rtimes \mathrm{O}(\cc^2)$ and $\cc^2\rtimes \so(\cc^2)$. This leads to the ${}_0F_1$ equation (or, equivalently, to the Bessel equation). We included this section for completeness, however its material is well-known and well documented in the literature.

Note that Sections \ref{s2}, \ref{s3}, \ref{s4}, \ref{Conformal invariance} and \ref{sec-sch} are quite general and abstract. On the other hand, Sections \ref{s6}, \ref{s7}, \ref{sec-con}, \ref{sec-her} and \ref{sec-bess} are more concrete and present applications of the general theory to the classes of hypergeometric type equations (\ref{hyp1}), (\ref{hyp2}), (\ref{hyp3}), (\ref{hyp4}) and (\ref{hyp5}), respectively.
To a large extent they can be read independently of the ``general'' part of the paper.

\subsection{Comparison with literature}

Properties of functions of hypergeometric type are described in numerous books, such as \cite{AAR,EMOT1,EMOT2,EMOT3,Ho,Kuz,MOS,NIST,R,WW}. In particular, the properties presented in Section~\ref{Hypergeometric type operators and their symmetries} (transmutation relations, discrete symmetries and factorizations) are known in one form or another. A similar presentation can be found in~\cite{De}. (Unlike in this paper, the presentation of~\cite{De} involves the standard form of hypergeometric type equations and not their balanced form).

Lie algebras associated with the Bessel and Hermite functions can be found in papers by Weisner~\cite{We1,We2}.

The idea of studying hypergeometric type equations with help of Lie algebras was developed further by Miller. His early book~\cite{M1} considers mostly small Lie algebras/Lie groups, typically $\mathrm{sl}(2,\cc)$/$\mathrm{SL}(2,\cc)$ and its contractions, and applies them to obtain various identities about hypergeometric type functions. These Lie algebras have 1-dimensional ``Cartan algebras'' and a~single pair of roots. This kind of analysis is able to explain only a~single pair of transmutation relations, whereas to explain bigger families of transmutation relations one needs larger Lie algebras.

A Lie algebra strictly larger than $\mathrm{sl}(2,\cc)$ is $\so(4,\cc)$. There exists a large literature on the relation of the hypergeometric equation with $\so(4,\cc)$ and its real forms, see, e.g.,~\cite{KM}. This Lie algebra is however still too small to account for all symmetries of the hypergeometric equation~-- its Cartan algebra is only 2-dimensional, whereas the equation has three parameters.

An explanation of symmetries of the Gegenbauer equation in terms of $\so(5)$ and of the hypergeometric equation in terms of $\so(6)\simeq \mathrm{sl}(4)$ was f\/irst given by Miller, see~\cite{M4}, and especially~\cite{M5}.

Miller and Kalnins wrote a series of papers where they studied the symmetry approach to separation of variables for various 2nd order partial dif\/ferential
equations, such as the Laplace and wave equation, see, e.g.,~\cite{KM1}. A large part of this research is summed up in the book by Miller~\cite{M3}. As an important consequence of this study, one obtains detailed information about symmetries of hypergeometric type equations.

The main tool that we use to describe properties of hypergeometric type functions is the theory of {\em generalized symmetries} of 2nd order linear PDE's. This theory is described in another book by Miller \cite{M2}, and further developed in \cite{M3}.

The fact that conformal transformations of the Euclidean space are generalized symmetries of the Laplace equation was apparently known already to Kelvin.
Its explanation in terms of the null quadric f\/irst appeared in \cite{Boc}. Null quadric as a tool to study conformal symmetries of the Laplace equation is the basic tool of \cite{KM1,KMR}.

The conformal invariance of the Laplace equation generalizes to arbitrary pseudo-Riemannian manifolds. In fact, the Laplace--Beltrami operator plus an appropriate multiple of the scalar curvature, sometimes called the {\em Yamabe Laplacian}, is invariant in a generalized sense with respect to conformal maps. This can be found for instance in~\cite{Or, Tay}.

The group of generalized symmetries of the heat equation was known already to Lie~\cite{L}. It was rediscovered (in the essentially equivalent context of the free Schr\"odinger equation) by Schr\"odinger~\cite{Sch}. It was then studied, e.g., in~\cite{Ha,Ni}.
Elementary notions from dif\/ferential geometry used in our paper are well known. One of standard references in this subject is~\cite{KN1,KN2}.

A topic that is extensively treated in the literature on the relation of special functions to group theory, such as \cite{M1, Ol,V, VK,Wa}, is derivation of various addition formulas. Addition formulas say that a certain special function can be written as a sum, often inf\/inite, of some related functions. These identities can be typically interpreted in terms of a certain representation of an appropriate Lie group. These identities are very interesting, however we do not discuss them. The only elements of Lie groups that we consider are very special~-- they are the ``Weyl symmetries''. They yield discrete symmetries of hypergeometric type equations, such as Kummer's table. We leave out addition formulas, because their theory is considerably more complicated than what we consider in our paper.

The relationship of Kummer's table with the group of symmetries of a~cube (which is the Weyl group of $\so(\cc^6)$) was discussed in~\cite{LSV}. A recent paper, where symmetries of the hypergeometric equation play an important role is~\cite{Ko}. (We learned the term ``transmutation relations'' from this paper).

The use of transmutation relations as a tool to derive recurrence relations for hypergeometric type functions is well known and can be found, e.g., in the book by Nikiforov--Uvarov~\cite{NU}, in the books by Miller~\cite{M1} or in older works such as \cite{Tr,We1,We2}.

There exists various generalizations of hypergeometric type functions. Let us mention the class of $\mathcal{A}$-hypergeometric functions, which provides a~natural generalization of the usual hypergeometric function to many-variable situations \cite{Be,Bod}. Saito \cite{Sa} considers generalized symmetries in the framework of $\mathcal{A}$-hypergeometric functions. Note, however, that the results of Saito are incomplete in the case of the classic hypergeometric equation. He admits this: ``When $p=2$, the symmetry Lie algebra is much larger than $ {\mathfrak g}_{2}$'', and he quotes the paper by Miller~\cite{M5}. Similarly, the (surprisingly large) Lie algebras of symmetries of the Gegenbauer and Hermite equations cannot be easily seen from a (seemingly very general) analysis of Saito.

There are a number of topics related to the hypergeometric type equation that we do not touch. Let us mention the question whether hypergeometric functions can be expressed in terms of algebraic functions. This topic, in the context of $\mathcal{A}$-hypergeometric functions was considered, e.g., in the interesting papers~\cite{Be,Bod}.

In our paper we stick to a rather limited class of equations. We do not have the ambition to go for generalizations. This limited class has a surprisingly rich structure, which seems to be lost when we consider their generalizations.

Many, perhaps most identities and ideas described in our paper can be found in one form or another in the literature, especially in the works by Miller, also by Miller and Kalnins, as we discussed above. Nevertheless, we believe that our work raises important points that are not explicitly described in the literature. We argue that symmetry properties of all hypergeometric type equations become almost obvious if we add a~certain number of variables obtaining the Laplace equation. We describe this idea in a unif\/ied framework, identifying the relationship of theory of hypergeometric type equations with such elements of group theory as roots, Cartan algebras and Weyl groups. These ideas are summed in Table~1, which to our knowledge has not appeared in the literature, except for the paper \cite{De} written by one of us.

We use various (minor but helpful) ideas to make our presentation as short and transparent as possible: e.g., the {\em balanced form} of hypergeometric type equations, {\em Lie algebraic parameters} and {\em split coordinates} in $\cc^n$. In our derivations the symmetries are completely obvious at the starting point, then at each step they become more and more complicated.

The derivation of generalized symmetries of the Laplacian, given after Theorem~\ref{wer}, is pro\-bably partly original. It leads to an interesting geometric object, which we call~$\Delta^\diamond$. It satisf\/ies identities~(\ref{deq1a}) and (\ref{deq3a}), which seem quite important in the context of conformal invariance of the Laplace equation. These identities are elementary and quite simple, however we have never seen them in the literature. They can be used to derive factorizations of hypergeometric type equations, relating them to Casimir operators of certain distinguished subalgebras, another point that is probably original.

\section{Hypergeometric type operators and their symmetries}\label{Hypergeometric type operators and their symmetries}

In this section we describe the families of identities that we would like to interpret in a group-theoretical fashion in this article. As mentioned above, all of them involve the balanced form of the operators (\ref{hyp1})--(\ref{hyp5}).

\subsection{Hypergeometric operator}\label{subs-1}

In the hypergeometric equation (\ref{hyp1}) we prefer to replace the parameters $a$, $b$, $c$ with
\begin{gather*}
\alpha:=c-1,\qquad \beta: =a+b-c,\qquad \mu:=b-a. 
\end{gather*}
We obtain the {\em $($standard$)$ hypergeometric operator}
\begin{gather}
\cF_{\alpha,\beta ,\mu} (w,\p_w)
 =w(1-w)\p_w^2+\big((1+\alpha)(1-w)-(1+\beta )w\big)\p_w\nonumber\\
 \hphantom{\cF_{\alpha,\beta ,\mu} (w,\p_w) =}{} +\tfrac14\mu^2-\tfrac14(\alpha+\beta +1)^2.\label{hy1-tra}
\end{gather}
Instead of (\ref{hy1-tra}) we prefer to consider the {\em balanced hypergeometric operator}
\begin{gather}\nonumber
\cF_{\alpha,\beta,\mu}^\bal(w,\partial_w):= w^{\frac{\alpha}{2}}(1-w)^{\frac{\beta}{2}}\cF_{\alpha,\beta,\mu}(w,\partial_w){(1-w)^{-\frac{\beta}{2}}w^{-\frac{\alpha}{2}}}\\
\hphantom{\cF_{\alpha,\beta,\mu}^\bal(w,\partial_w)}{} = \partial_ww(1-w)\partial_w-\frac{\alpha^2}{4w}-\frac{\beta^2}{4(1-w)}+\frac{\mu^2-1}{4}.
\label{hyp1c}
\end{gather}

\emph{Discrete symmetries.} $\fbal{}{}{}{w}$ does not change if we f\/lip the signs of $\alpha$, $\beta$, $\mu$. Besides, the following operators coincide with $\fbal{}{}{}{w}$:
\begin{alignat*}{3}
&w = z\colon \qquad && \fbal{}{}{}{z},&\\
&w = 1-z\colon \qquad &&\fbalpre{\beta}{\alpha}{\mu}{z},&\\
&w = \frac{1}{z}\colon \qquad &&z^{\frac{1}{2}}(-z)\fbalpre{\mu}{\beta}{\alpha}{z}z^{-\frac{1}{2}},&\\
&w = 1-\frac{1}{z}\colon \qquad &&z^{\frac{1}{2}}(-z)\fbalpre{\mu}{\alpha}{\beta}{z} z^{-\frac{1}{2}},&\\
&w = \frac{1}{1-z}\colon \qquad &&(1-z)^{\frac{1}{2}}(z-1)\fbalpre{\beta}{\mu}{\alpha}{z} (1-z)^{-\frac{1}{2}},&\\
&w = \frac{z}{z-1}\colon \qquad && (1-z)^{\frac{1}{2}}(z-1)\fbalpre{\alpha}{\mu}{\beta}{z} (1-z)^{-\frac{1}{2}}.&
\end{alignat*}

Transmutation relations:
\begin{gather*}
\sqrt{w(1-w)} \left(\ddw-\frac{\alpha}{2w}+\frac{\beta}{2(1-w)}\right) \fbalw\\
\qquad{} =\fbal{+1}{+1}{}{w}\sqrt{w(1-w)}\left(\ddw-\frac{\alpha}{2w}+\frac{\beta}{2(1-w)}\right),
\\[1mm]
\sqrt{w(1-w)}\left(\ddw+\frac{\alpha}{2w}-\frac{\beta}{2(1-w)}\right)\fbalw\\
\qquad{} =\fbal{-1}{-1}{}{w}\sqrt{w(1-w)}\left(\ddw+\frac{\alpha}{2w}-\frac{\beta}{2(1-w)}\right),
\\[1mm]
\sqrt{w(1-w)}\left(\ddw-\frac{\alpha}{2w}-\frac{\beta}{2(1-w)}\right)\fbalw\\
\qquad{} =\fbal{+1}{-1}{}{w}\sqrt{w(1-w)}\left(\ddw-\frac{\alpha}{2w}-\frac{\beta}{2(1-w)}\right),
\\[1mm]
\sqrt{w(1-w)}\left(\ddw+\frac{\alpha}{2w}+\frac{\beta}{2(1-w)}\right)\fbalw\\
\qquad{} =\fbal{-1}{+1}{}{w}\sqrt{w(1-w)}\left(\ddw+\frac{\alpha}{2w}+\frac{\beta}{2(1-w)}\right),
\\[1mm]
\sqrt{w}\left(2(1-w)\ddw-\frac{\alpha}{\vph{w}}-\mu-3\right)\fbalw\\
\qquad{} =\fbal{+1}{}{+1}{w} \sqrt{w}\left(2(1-w)\ddw-\frac{\alpha}{\vph{w}}-\mu-1\right),
\\[1mm]
\sqrt{w}\left(2(1-w)\ddw+\frac{\alpha}{\vph{w}}+\mu-3\right)\fbalw\\
\qquad{} =\fbal{-1}{}{-1}{w}\sqrt{w}\left(2(1-w)\ddw+\frac{\alpha}{\vph{w}}+\mu-1\right),
\\
\sqrt{w}\left(2(1-w)\ddw-\frac{\alpha}{\vph{w}}+\mu-3\right)\fbalw\\
\qquad{} =\fbal{+1}{}{-1}{w}\sqrt{w}\left(2(1-w)\ddw-\frac{\alpha}{\vph{w}}+\mu-1\right),
\\[1mm]
\sqrt{w}\left(2(1-w)\ddw+\frac{\alpha}{\vph{w}}-\mu-3\right)\fbalw\\
\qquad{} =\fbal{-1}{}{+1}{w}\sqrt{w}\left(2(1-w)\ddw+\frac{\alpha}{\vph{w}}-\mu-1\right),
\\
\sqrt{1-w}\left(-2w\ddw-\frac{\beta}{1-w}-\mu-3\right)\fbalw\\
\qquad{} =\fbal{}{+1}{+1}{w} \sqrt{1-w}\left(-2w\ddw-\frac{\beta}{1-w}-\mu-1\right),
\\[1mm]
\sqrt{1-w}\left(-2w\ddw+\frac{\beta}{1-w}+\mu-3\right)\fbalw\\
\qquad{} =\fbal{}{-1}{-1}{w}\sqrt{1-w}\left(-2w\ddw+\frac{\beta}{1-w}+\mu-1\right),
\\[1mm]
\sqrt{1-w}\left(-2w\ddw-\frac{\beta}{1-w}+\mu-3\right)\fbalw\\
\qquad{} =\fbal{}{+1}{-1}{w} \sqrt{1-w}\left(-2w\ddw-\frac{\beta}{1-w}+\mu-1\right),
\\[1mm]
\sqrt{1-w}\left(-2w\ddw+\frac{\beta}{1-w}-\mu-3\right)\fbalw\\
\qquad{} =\fbal{}{-1}{+1}{w}\sqrt{1-w}\left(-2w\ddw+\frac{\beta}{1-w}-\mu-1\right).
\end{gather*}

Factorizations:
\begin{gather*}
\fbalw \\
 =\sqrt{w(1-w)}\left(\ddw-\frac{\alpha-1}{2w}+\frac{\beta-1}{2(1-w)}\right)\sqrt{w(1-w)}\left(\ddw-\frac{\alpha}{2w}+\frac{\beta}{2(1-w)}\right) \\
\hphantom{=}{} -\frac{1}{4}\left(\left(\beta+\alpha\right)\left(\beta+\alpha-2\right)+\mu^2-1\right) \\
 =\sqrt{w(1-w)}\left(\ddw+\frac{\alpha+1}{2w}+\frac{\beta-1}{2(1-w)}\right)\sqrt{w(1-w)}\left(\ddw-\frac{\alpha}{2w}-\frac{\beta}{2(1-w)}\right)\\
\hphantom{=}{} -\frac{1}{4}\left(\left(\beta-\alpha\right)\left(\beta-\alpha-2\right)+\mu^2-1\right)\\
 =\sqrt{w(1-w)}\left(\ddw-\frac{\alpha-1}{2w}-\frac{\beta+1}{2(1-w)}\right)\sqrt{w(1-w)}\left(\ddw-\frac{\alpha}{2w}-\frac{\beta}{2(1-w)}\right) \\
\hphantom{=}{}-\frac{1}{4}\left(\left(\beta-\alpha\right)\left(\beta-\alpha+2\right)+\mu^2-1\right)\\
 =\sqrt{w(1-w)}\left(\ddw+\frac{\alpha+1}{2w}-\frac{\beta+1}{2(1-w)}\right)\sqrt{w(1-w)}\left(\ddw-\frac{\alpha}{2w}+\frac{\beta}{2(1-w)}\right)\\
\hphantom{=}{}-\frac{1}{4}\left(\left(\beta+\alpha\right)\left(\beta+\alpha+2\right)+\mu^2-1\right),
\\[1mm]
 -(1-w)\fbalw \\
=-\sqrt{w}\left(-\frac{1}{2}+(1-w)\ddw-\frac{\alpha-1}{2w}-\frac{\mu-1}{2}\right)\sqrt{w}\left(-\frac{1}{2}+(1-w)\ddw+\frac{\alpha}{2w}+\frac{\mu}{2}\right)\\
\hphantom{=}{}-\frac{1}{4}\left(\left(\mu+\alpha\right)\left(\mu+\alpha-2\right)+\beta^2-1\right) \\
 =-\sqrt{w}\left(-\frac{1}{2}+(1-w)\ddw+\frac{\alpha+1}{2w}-\frac{\mu-1}{2}\right)\sqrt{w}\left(-\frac{1}{2}+(1-w)\ddw
-\frac{\alpha}{2w}+\frac{\mu}{2}\right)\\
\hphantom{=}{} -\frac{1}{4}\left(\left(\mu-\alpha\right)\left(\mu-\alpha-2\right)+\beta^2-1\right)\\
 =-\sqrt{w}\left(-\frac{1}{2}+(1-w)\ddw-\frac{\alpha-1}{2w}+\frac{\mu+1}{2}\right)\sqrt{w}\left(-\frac{1}{2}+(1-w)\ddw
 +\frac{\alpha}{2w}-\frac{\mu}{2}\right)\\
\hphantom{=}{}-\frac{1}{4}\left(\left(\mu-\alpha\right)\left(\mu-\alpha+2\right)+\beta^2-1\right)\\
 =-\sqrt{w}\left(-\frac{1}{2}+(1-w)\ddw+\frac{\alpha+1}{2w}+\frac{\mu+1}{2}\right)\sqrt{w}\left(-\frac{1}{2}+(1-w)\ddw
-\frac{\alpha}{2w}-\frac{\mu}{2}\right)\\
\hphantom{=}{} -\frac{1}{4}\left(\left(\mu+\alpha\right)\left(\mu+\alpha+2\right)+\beta^2-1\right),
\\[1mm]
 -w\fbalw\\
 =-\sqrt{1-w}\left(-\frac{1}{2}-w\ddw-\frac{\beta-1}{2(1-w)}-\frac{\mu-1}{2}\right)\sqrt{1-w}\left(-\frac{1}{2}-w\ddw
 +\frac{\beta}{2(1-w)}+\frac{\mu}{2}\right)\\
\hphantom{=}{} -\frac{1}{4}\left(\left(\mu+\beta\right)\left(\mu+\beta-2\right)+\alpha^2-1\right) \\
 =-\sqrt{1-w}\left(-\frac{1}{2}-w\ddw+\frac{\beta+1}{2(1-w)}-\frac{\mu-1}{2}\right)\sqrt{1-w}\left(-\frac{1}{2}-w\ddw
-\frac{\beta}{2(1-w)}+\frac{\mu}{2}\right)\\
\hphantom{=}{} -\frac{1}{4}\left(\left(\mu-\beta\right)\left(\mu-\beta-2\right)+\alpha^2-1\right)\\
 =-\sqrt{1-w}\left(-\frac{1}{2}-w\ddw-\frac{\beta-1}{2(1-w)}+\frac{\mu+1}{2}\right)\sqrt{1-w}\left(-\frac{1}{2}-w\ddw
 +\frac{\beta}{2(1-w)}-\frac{\mu}{2}\right)\\
\hphantom{=}{}-\frac{1}{4}\left(\left(\mu-\beta\right)\left(\mu-\beta+2\right)+\alpha^2-1\right) \\
 =-\sqrt{1-w}\left(-\frac{1}{2}-w\ddw+\frac{\beta+1}{2(1-w)}+\frac{\mu+1}{2}\right)\sqrt{1-w}\left(-\frac{1}{2}-w\ddw
 -\frac{\beta}{2(1-w)}-\frac{\mu}{2}\right)\\
\hphantom{=}{}-\frac{1}{4}\left(\left(\mu+\beta\right)\left(\mu+\beta+2\right)+\alpha^2-1\right).
\end{gather*}
It is striking how symmetric the above formulas look like. The main goal of our paper is to explain why this is so.

\subsection{Gegenbauer operator}\label{subs-2}

In the hypergeometric equation (\ref{hyp1}) let us move the singular points to $-1$, $1$ and demand that it is ref\/lection invariant. Then we can eliminate one of the parameters, say~$c$. We obtain the Gegenbauer equation~(\ref{hyp2}). We introduce new parameters
\begin{gather*}
\alpha :=\frac{{a}+{b}-1}{2},\qquad \lambda: =\frac{{b}-{a}}{2}.
\end{gather*}
We obtain the {\em $($standard$)$ Gegenbauer operator}
\begin{gather}
\cS_{\alpha ,\lambda }(w,\p_w):= \big(1-w^2\big)\p_w^2-2(1+\alpha )w\p_w+\lambda ^2-\left(\alpha +\frac{1}{2}\right)^2.\label{stan2}
\end{gather}
The {\em balanced Gegenbauer operator} is
\begin{gather}
\notag \cS_{\alpha,\lambda}^\bal(w,\partial_w):= \big(w^2-1\big)^{\frac{\alpha}{2}}\cS_{\alpha,\lambda}(w,\p_w)\big(w^2-1\big)^{-\frac{\alpha}{2}}\\
\hphantom{\cS_{\alpha,\lambda}^\bal(w,\partial_w)}{} = \partial_w\big(1-w^2\big)\partial_w-\frac{\alpha^2}{1-w^2}+\lambda^2-\frac{1}{4}.\label{bal2}
\end{gather}

\emph{Discrete symmetries.} $\sbal{}{}{w}$ does not change if we f\/lip the signs of $\alpha$, $\lambda$. Besides, the following operators coincide with $\sbal{}{}{w}$:
\begin{alignat*}{3}
& w = z\colon \qquad && \sbal{}{}{z},& \\
& w = \frac{z}{\sqrt{z^2-1}}\colon \qquad && \big(z^2-1\big)^{\frac{1}{4}}\big(z^2-1\big) \sbalpre{\lambda}{\alpha}{z}\big(z^2-1\big)^{-\frac{1}{4}}.&
\end{alignat*}

\emph{Transmutation relations:}
\begin{gather*}
\sqrt{1-w^2} \left(-\frac{5}{2}-w\ddw-\frac{\alpha}{1-w^2}-\lambda\right) \sbalw\\
\qquad{} =\sbal{+1}{+1}{w} \sqrt{1-w^2}\left(-\frac{1}{2}-w\ddw-\frac{\alpha}{1-w^2}-\lambda\right),
\\[1mm]
\sqrt{1-w^2} \left(-\frac{5}{2}-w\ddw+\frac{\alpha}{1-w^2}+\lambda\right) \sbalw\\
 \qquad{} =\sbal{-1}{-1}{w} \sqrt{1-w^2}\left(-\frac{1}{2}-w\ddw+\frac{\alpha}{1-w^2}+\lambda\right),
\\[1mm]
\sqrt{1-w^2} \left(-\frac{5}{2}-w\ddw-\frac{\alpha}{1-w^2}+\lambda\right) \sbalw\\
 \qquad{} =\sbal{+1}{-1}{w} \sqrt{1-w^2}\left(-\frac{1}{2}-w\ddw-\frac{\alpha}{1-w^2}+\lambda\right),
\\[1mm]
\sqrt{1-w^2} \left(-\frac{5}{2}-w\ddw+\frac{\alpha}{1-w^2}-\lambda\right) \sbalw\\
 \qquad{} =\sbal{-1}{+1}{w} \sqrt{1-w^2}\left(-\frac{1}{2}-w\ddw+\frac{\alpha}{1-w^2}-\lambda\right),
\\[1mm]
w\left(-\frac{5}{2}+\frac{1-w^2}{w}\ddw-\lambda\right) \sbalw
 =\sbal{}{+1}{w} w\left(-\frac{1}{2}+\frac{1-w^2}{w}\ddw-\lambda\right),
\\[1mm]
w\left(-\frac{5}{2}+\frac{1-w^2}{w}\ddw+\lambda\right) \sbalw
 =\sbal{}{-1}{w} w\left(-\frac{1}{2}+\frac{1-w^2}{w}\ddw+\lambda\right),
\\[1mm]
\sqrt{1-w^2}\left(\ddw+\frac{w}{1-w^2} \alpha \right) \sbalw
 =\sbal{+1}{}{w} \sqrt{1-w^2}\left(\ddw+\frac{w}{1-w^2} \alpha \right),
\\[1mm]
\sqrt{1-w^2}\left(\ddw-\frac{w}{1-w^2} \alpha \right) \sbalw
 =\sbal{+1}{}{w} \sqrt{1-w^2}\left(\ddw-\frac{w}{1-w^2} \alpha \right).
\end{gather*}

\emph{Factorizations:}
\begin{gather*}
- \big(1-w^2\big) \sbalw \\
 \qquad{} = -w \left(\frac{1}{2}+\frac{1-w^2}{w}\ddw-\lambda\right) w \left(-\frac{1}{2}+\frac{1-w^2}{w}\ddw+\lambda\right)-\lambda^2+\lambda+\alpha^2-\frac{1}{4}\\
 \qquad{} = -w \left(\frac{1}{2}+\frac{1-w^2}{w}\ddw+\lambda\right) w \left(-\frac{1}{2}+\frac{1-w^2}{w}\ddw-\lambda\right)-\lambda^2-\lambda+\alpha^2-\frac{1}{4},\\[1mm]
 \sbalw \\
 \qquad{} = \sqrt{1-w^2} \left(\ddw+\frac{w}{1-w^2} \left(\alpha-1\right) \right) \sqrt{1-w^2} \left(\ddw-\frac{w}{1-w^2} \alpha \right)\\
 \qquad\quad{} -\alpha^2+\alpha+\lambda^2-\frac{1}{4}\\
 \qquad{} = \sqrt{1-w^2} \left(\ddw-\frac{w}{1-w^2} \left(\alpha+1\right) \right) \sqrt{1-w^2} \left(\ddw+\frac{w}{1-w^2} \alpha \right) -\alpha^2-\alpha+\lambda^2-\frac{1}{4},
\\[1mm]
-w^2 \sbalw \\
\qquad{} =-\sqrt{1-w^2}\left(\frac{1}{2}-w\ddw-\frac{\alpha-1}{1-w^2}-\lambda\right)\sqrt{1-w^2}\left(-\frac{1}{2}-w\ddw+\frac{\alpha}{1-w^2}+\lambda\right)\\
 \qquad \quad{}-\left(\lambda+\alpha\right)\left(\lambda+\alpha-2\right)-\frac{3}{4}\\
\qquad{} =-\sqrt{1-w^2}\left(\frac{1}{2}-w\ddw+\frac{\alpha+1}{1-w^2}-\lambda\right)\sqrt{1-w^2}\left(-\frac{1}{2}-w\ddw-\frac{\alpha}{1-w^2}+\lambda\right)\\
 \qquad\quad{} -\left(\lambda-\alpha\right)\left(\lambda-\alpha-2\right)-\frac{3}{4}\\
 \qquad{}=-\sqrt{1-w^2}\left(\frac{1}{2}-w\ddw+\frac{\alpha+1}{1-w^2}+\lambda\right)\sqrt{1-w^2}\left(-\frac{1}{2}-w\ddw-\frac{\alpha}{1-w^2}-\lambda\right)\\
 \qquad\quad{} -\left(\lambda+\alpha\right)\left(\lambda+\alpha+2\right)-\frac{3}{4}\\
 \qquad{}=-\sqrt{1-w^2}\left(\frac{1}{2}-w\ddw-\frac{\alpha-1}{1-w^2}+\lambda\right)\sqrt{1-w^2}\left(-\frac{1}{2}-w\ddw+\frac{\alpha}{1-w^2}-\lambda\right)\\
 \qquad\quad{} -\left(\lambda-\alpha\right)\left(\lambda-\alpha+2\right)-\frac{3}{4}.
\end{gather*}

\subsection{Conf\/luent operator}\label{subs-3}

In the conf\/luent equation (\ref{hyp3}) we introduce new parameters
\begin{gather*}\alpha:=c-1,\qquad \theta: =2a-c.\end{gather*}
The {\em $($standard$)$ confluent operator} is
\begin{gather}
\cF_{\theta ,\alpha}(w,\p_w):=w\p_w^2+(1+\alpha-w)\p_w-\frac{1}{2}(1+\theta +\alpha).\label{stan3}
\end{gather}
The {\em balanced confluent operator} is
\begin{gather}
 \cF_{\theta,\alpha}^\bal(w,\partial_w):= w^{\frac{\alpha}{2}}\e^{-\frac{w}{2}}\cF_{\theta,\alpha}(w,\p_w){\e^{\frac{w}{2}}w^{-\frac{\alpha}{2}}}
=\partial_ww\partial_w-\frac{w}{4}-\frac{\theta}{2}-\frac{\alpha^2}{4w}.\label{bal3}
\end{gather}

\emph{Discrete symmetries.} $\cF_{\theta,\alpha}^{\mathrm{bal}}(w,\ddw)$ does not change if we f\/lip the sign of $\alpha$. Besides, the following operators coincide with $\cF_{\theta,\alpha}^{\mathrm{bal}}(w,\ddw)$:
\begin{gather*}
w = z\colon \quad \cF_{\theta,\alpha}^{\mathrm{bal}}(z,\partial_z ),\qquad w = -z\colon \quad \cF_{-\theta,\alpha}^{\mathrm{bal}}(z,\partial_z ).
\end{gather*}

\emph{Transmutation relations:}
\begin{gather*}
\frac{1}{\sqrt{w}}\left(w\ddw+\frac{\alpha}{2}+\frac{w}{2}\right) \conbalw
 =\conbal{+1}{-1}{w} \frac{1}{\sqrt{w}}\left(w\ddw+\frac{\alpha}{2}+\frac{w}{2}\right),\\[2mm]
\frac{1}{\sqrt{w}}\left(w\ddw-\frac{\alpha}{2}+\frac{w}{2}\right) \conbalw
 =\conbal{+1}{+1}{w} \frac{1}{\sqrt{w}}\left(w\ddw-\frac{\alpha}{2}+\frac{w}{2}\right),\\[2mm]
\frac{1}{\sqrt{w}}\left(w\ddw+\frac{\alpha}{2}-\frac{w}{2}\right) \conbalw
 =\conbal{-1}{-1}{w} \frac{1}{\sqrt{w}}\left(w\ddw+\frac{\alpha}{2}-\frac{w}{2}\right),\\[2mm]
\frac{1}{\sqrt{w}}\left(w\ddw-\frac{\alpha}{2}-\frac{w}{2}\right) \conbalw
 =\conbal{-1}{+1}{w} \frac{1}{\sqrt{w}}\left(w\ddw-\frac{\alpha}{2}-\frac{w}{2}\right),\\[2mm]
\left(-w\ddw-\frac{\theta}{2}-\frac{w}{2}-\frac{3}{2}\right) \conbalw
 =\conbal{+2}{}{w} \left(-w\ddw-\frac{\theta}{2}-\frac{w}{2}-\frac{1}{2}\right),\\[2mm]
\left(w\ddw-\frac{\theta}{2}-\frac{w}{2}+\frac{3}{2}\right) \conbalw
 =\conbal{-2}{}{w} \left(w\ddw-\frac{\theta}{2}-\frac{w}{2}+\frac{1}{2}\right).
\end{gather*}

\emph{Factorizations:}
\begin{gather*}
-w \conbalw
 =\left(w\ddw-\frac{\theta+1}{2}-\frac{w}{2}\right) \left(-w\ddw-\frac{\theta+1}{2}-\frac{w}{2}\right)-\frac{1}{4}\left(\theta+1\right)^2+\frac{1}{4}\alpha^2\\
\hphantom{-w \conbalw}{} =\left(-w\ddw-\frac{\theta-1}{2}-\frac{w}{2}\right) \left(w\ddw-\frac{\theta-1}{2}-\frac{w}{2}\right)-\frac{1}{4}\left(\theta-1\right)^2+\frac{1}{4}\alpha^2,\\[2mm]
 \conbalw
 =\frac{1}{\sqrt{w}}\left(w\ddw-\frac{\alpha-1}{2}-\frac{w}{2}\right)\frac{1}{\sqrt{w}}\left(w\ddw+\frac{\alpha}{2}+\frac{w}{2}\right)
 -\frac{1}{2}\left(\theta-\alpha+1\right)\\
\hphantom{\conbalw}{} =\frac{1}{\sqrt{w}}\left(w\ddw+\frac{\alpha+1}{2}+\frac{w}{2}\right) \frac{1}{\sqrt{w}}\left(w\ddw-\frac{\alpha}{2}-\frac{w}{2}\right)-\frac{1}{2}\left(\theta-\alpha-1\right)\\
\hphantom{\conbalw}{} =\frac{1}{\sqrt{w}}\left(w\ddw+\frac{\alpha+1}{2}-\frac{w}{2}\right) \frac{1}{\sqrt{w}}\left(w\ddw-\frac{\alpha}{2}+\frac{w}{2}\right)-\frac{1}{2}\left(\theta+\alpha+1\right)\\
\hphantom{\conbalw}{} =\frac{1}{\sqrt{w}}\left(w\ddw-\frac{\alpha-1}{2}+\frac{w}{2}\right)\frac{1}{\sqrt{w}}\left(w\ddw+\frac{\alpha}{2}-\frac{w}{2}\right)-\frac{1}{2}\left(\theta+\alpha-1\right).
\end{gather*}

\subsection{Hermite operator} \label{subs-4}

In the Hermite equation (\ref{hyp4}) we prefer to use the parameter
\begin{gather*}\lambda={a}-\tfrac12.\end{gather*} The {\em $($standard$)$ Hermite operator} is
\begin{gather}
\cS_\lambda (w,\p_w) := \p_w^2-2w\p_w-2\lambda -1.\label{stan4}
\end{gather}
The {\em balanced Hermite operator} is
\begin{gather} \cS_\lambda ^\bal(w,\p_w):=\e^{-\frac{w^2}{2}}\cS_{\lambda}(w,\p_w)\e^{\frac{w^2}{2}} =\partial_w^2-w^2-2\lambda.\label{bal4}
\end{gather}

\emph{Discrete symmetries.} The following operators coincide with $\hbalw$:
\begin{alignat*}{5}
& w = z \colon \quad && \cS_{\lambda}^{\mathrm{bal}}(z, \partial_z),\quad\qquad && w = \ii z\colon \quad&& -\cS_{-\lambda}^{\mathrm{bal}}(z, \partial_z),&\\
& w =-z\colon \quad && \cS_{\lambda}^{\mathrm{bal}}(z, \partial_z),\quad \qquad && w = -\ii z\colon \quad && -\cS_{-\lambda}^{\mathrm{bal}}(z, \partial_z).&
\end{alignat*}

\emph{Transmutation relations:}
\begin{gather*}
(\ddw+w)\hbalw=\hbal{+1}{w}(\ddw+w),\\[1mm]
(\ddw-w )\hbalw=\hbal{-1}{w}(\ddw-w),\\[1mm]
\big({-}w\ddw-\lambda-w^2-\tfrac{5}{2}\big)\hbalw =\hbal{+2}{w}\big({-}w\ddw-\lambda-w^2-\tfrac{1}{2}\big),\\[1mm]
\big(w\ddw-\lambda-w^2+\tfrac{5}{2}\big)\hbalw =\hbal{-2}{w}\big(w\ddw-\lambda-w^2+\tfrac{1}{2}\big).
\end{gather*}

\emph{Factorizations:}
\begin{gather*}
-w^2\hbalw =\big(w\ddw-\lambda-\tfrac{3}{2}-w^2\big) \big({-}w\ddw-\lambda-\tfrac{1}{2}-w^2\big)- (\lambda+1 )^2+\tfrac{1}{4}\\
\hphantom{-w^2\hbalw}{}
=\big({-}w\ddw-\lambda+\tfrac{3}{2}-w^2\big) \big(w\ddw-\lambda+\tfrac{1}{2}-w^2\big)-(\lambda-1)^2+\tfrac{1}{4},\\[1mm]
\hbalw = (\ddw-w ) (\ddw+w )-2\lambda-1 = (\ddw+w ) (\ddw-w )-2\lambda+1.
\end{gather*}

\subsection[${}_0F_1$ operator]{$\boldsymbol{{}_0F_1}$ operator}\label{subs-5}

In the ${}_0F_1$ equation (\ref{hyp5}) we prefer to use the parameter
\begin{gather*} \alpha:=c-1.
\end{gather*}
The {\em $($standard$)$ ${}_0F_1$ operator} is
\begin{gather*}
\cF_\alpha (w,\p_w):=w\p_w^2+(\alpha +1)\p_w-1.
\end{gather*}
The {\em balanced ${}_0F_1$ operator} is
\begin{gather*}
 \cF_{\alpha}^\bal(w,\partial_w) :=w^{\frac{\alpha}{2}}\cF_{\alpha}(w,\p_w)w^{-\frac{\alpha}{2}} =\partial_ww\partial_w-1-\frac{\alpha^2}{4w}.
\end{gather*}

\emph{Discrete symmetries.} $\cF_\alpha (w,\p_w)$ does not change if we f\/lip the sign of $\alpha$.

\emph{Transmutation relations:}
\begin{gather*}
\frac{1}{\sqrt{w}}\left(w \ddw-\frac{\alpha}{2}\right)\bbalw=\bbal{+1}{w}\frac{1}{\sqrt{w}}\left(w\ddw-\frac{\alpha}{2}\right),\\[1mm]
\frac{1}{\sqrt{w}}\left(w \ddw+\frac{\alpha}{2}\right)\bbalw=\bbal{-1}{w}\frac{1}{\sqrt{w}}\left(w\ddw+\frac{\alpha}{2}\right).
\end{gather*}

\emph{Factorizations:}
\begin{gather*}
\bbalw = \frac{1}{\sqrt{w}}\left(w\ddw-\frac{\alpha-1}{2}\right)\frac{1}{\sqrt{w}}\left(w\ddw+\frac{\alpha}{2}\right)-1\\
\hphantom{\bbalw }{} =\frac{1}{\sqrt{w}}\left(w\ddw+\frac{\alpha+1}{2}\right)\frac{1}{\sqrt{w}}\left(w\ddw-\frac{\alpha}{2}\right)-1.
\end{gather*}

\section{Basic complex geometry}\label{s2}

In this section we describe basic notation for complex geometry.
Throughout the section, $\Omega$,~$\Omega_1$,~$\Omega_2$ are open subsets of $\cc^n$ or, more generally, complex manifolds. We will write $\cc^\times$ for the multiplicative group $\cc\backslash\{0\}$.

We will write $\cA(\Omega)$ for the set of holomorphic functions on $\Omega$. $y=(y^1,\dots,y^n)$ will denote generic coordinates on $\Omega$. We will write $\cA^\times(\Omega)$ for the set of nowhere vanishing holomorphic functions on~$\Omega$.
\subsection{Vector f\/ields}

Let $\hol(\Omega)$ denote the Lie algebra of holomorphic vector f\/ields on $\Omega$. Every $A\in\hol(\Omega)$ can be identif\/ied with the dif\/ferential operator{\samepage
\begin{gather*}Af(y)=\sum_i A^i(y)\partial_{y^i}f(y),\qquad f\in\cA(\Omega),\end{gather*}
where $A^i\in\cA(\Omega)$, $i=1,\dots,n$.}

We will denote by $\cA\rtimes\hol(\Omega)$ the Lie algebra of 1st order dif\/ferential operators on~$\Omega$ with holomorphic coef\/f\/icents. Such operators can be written as
\begin{gather*}(A+M)f(y):=\sum_i A^i(y)\partial_{y^i}f(y)+M(y)f(y),\end{gather*}
where $A\in\hol(\Omega)$ and $M\in\cA(\Omega)$.

Let $\fg$ be a Lie subalgebra of $\hol(\Omega)$. A linear function $\fg\ni A\mapsto M_A\in\cA(\Omega)$ satisfying
\begin{gather*} A_{1} M_{A_2}-A_{2}M_{A_1}=M_{[A_1,A_2]}\end{gather*}
will be called a {\em cocycle for $\fg$}. Every cocycle together with $\eta\in\cc$ determines a~homomorphism
\begin{gather*}\fg\ni A\mapsto A+\eta M_A\in\cA\rtimes\hol(\Omega).\end{gather*}

\subsection{Point transformations}\label{Point transformations}

The set of biholomorphic maps $\Omega_1\to\Omega_2$ will be denoted $\Hol(\Omega_1,\Omega_2)$. We set $\Hol(\Omega):=\Hol(\Omega,\Omega)$.

Let ${\alpha}\in\Hol(\Omega_1,\Omega_2)$. The {\em transport} of functions, vector f\/ields, etc. by the map $\alpha$ will be also denoted by $\alpha$. More precisely, for $f\in\cA(\Omega_1)$ we def\/ine ${\alpha}f\in\cA(\Omega_2)$ by
\begin{gather*}({\alpha}f)(y):=f\big({\alpha}^{-1}(y)\big).\end{gather*}
For $A\in\hol(\Omega_1)$, ${\alpha}(A)\in\hol(\Omega_2)$ is def\/ined as
\begin{gather*}{\alpha}(A):=\alpha A \alpha^{-1}.\end{gather*}

If $m\in\cA^\times(\Omega_2)$, then we have a map $m{\alpha}\colon \cA(\Omega_1)\to\cA(\Omega_2)$ given by
\begin{gather*}(m{\alpha} f)(y):=m(y)f\big({\alpha}^{-1}(y)\big).\end{gather*}
$\cA^\times\rtimes\Hol(\Omega_1,\Omega_2)$ will denote the set of transformations $\cA(\Omega_1)\to\cA(\Omega_2)$ of this form. Clearly, $\cA^\times\rtimes\Hol(\Omega )$ is a group.

Suppose that $G$ is a subgroup of $\Hol(\Omega)$. A family $G\ni {\alpha}\mapsto m_{\alpha}\in\cA^\times(\Omega)$ satisfying
\begin{gather*}m_{{\alpha}_2}(y)m_{{\alpha}_1}\big({\alpha}_2^{-1}(y)\big)=m_{{\alpha}_2{\alpha}_1}(y),\qquad {\alpha}_1,{\alpha}_2\in G,\qquad y\in\Omega,\end{gather*}
will be called a {\em cocycle for $G$}. Every cocycle together with $\eta\in\zz$ determines a homomorphism
\begin{gather}
G\ni {\alpha}\mapsto m_{\alpha}^\eta {\alpha}\in\cA^\times\rtimes\Hol(\Omega).\label{biva}\end{gather}

\subsection{Local cocycles}\label{Local cocycles}

Unfortunately, the above def\/inition of a cocycle on a group is too rigid for our purposes. Below we introduce a weaker version of this concept, which we will be better adapted to our goals.

As before, we assume that $G$ is a subgroup of $\Hol(\Omega)$. Besides, we f\/ix $\Omega_0$ open in $\Omega$. For $\alpha\in G$ we will write
\begin{gather} \Omega_0^\alpha:=\Omega_0\cap\alpha(\Omega_0).\label{skip}\end{gather}
Furthermore, we suppose that to every $\alpha\in G$ we associate $m_\alpha\in\cA^\times\big(\Omega_0^\alpha\big)$ satisfying{\samepage
\begin{gather*}
m_{\alpha_2}(y)m_{\alpha_1}\big(\alpha_2^{-1}(y)\big)=m_{\alpha_2\alpha_1}(y),
\qquad \alpha_1,\alpha_2\in G,\qquad y\in\Omega_0\cap\alpha_2(\Omega_0)\cap
\alpha_2\circ\alpha_1(\Omega_0).\end{gather*}
Then $G\mapsto m_\alpha$ will be called a {\em local cocycle for $G$ based on~$\Omega_0$}.}

Let $p\in\cA^\times(\Omega_0)$. Then
\begin{gather} m_\alpha(y):=\frac{p(y)}{p\big(\alpha^{-1}(y)\big)},\qquad y\in \Omega_0^\alpha\label{coc}\end{gather}
is a (trivial) example of a local cocycle based on~$\Omega_0$. Note that if $p$ cannot be extended to a~holomorphic function on the whole $\Omega$, then~(\ref{coc}) cannot be extended to a true cocycle.

Let $\eta\in\zz$. For any $\alpha\in G$ we can def\/ine the map
\begin{gather}
m_\alpha^\eta\alpha\in\cA^\times\rtimes \Hol\big(\Omega_0^{\alpha^{-1}},\Omega_0^\alpha\big).\label{biva1}\end{gather}
For $\alpha_1,\alpha_2\in G$ and $\eta\in\zz$ we have the following weak form of the chain rule:
\begin{gather*}
\big(m_{\alpha_2}^\eta\alpha_2\circ m_{\alpha_1}^\eta\alpha_1\big)(y)=
m_{\alpha_2{\circ}\alpha_1}^\eta\alpha_2{\circ}\alpha_1(y),\qquad
y\in\Omega_0\cap\alpha_2(\Omega_0)\cap \alpha_2{\circ}\alpha_1(\Omega_0).
\end{gather*}
It will be convenient have a special notation for such a collection of maps~(\ref{biva1}): We will write that
\begin{gather*} G\ni\alpha\mapsto m_{\alpha}^\eta
 {\alpha}\mathop{\in}_\loc\cA^\times\rtimes\Hol(\Omega_0)\end{gather*}
is a {\em local representation of $G$}.

\subsection{Half-integer powers of a cocycle}

For non-integer exponents the power function is unfortunately multivalued. Because of that, strictly speaking, $\eta\not\in\zz$ should not be allowed in~(\ref{biva}). However, we will be forced to consider situations when $\eta$ is a half integer. This can be handled by the following formalism.

The non-identity element of the group $\zz_2$ acts on $\cc^\times$ by $\cc^\times\ni a\mapsto-a\in\cc^\times$. This def\/ines $\cc^\times/\zz_2$, which is the space of pairs of non-zero complex numbers dif\/fering by a~sign.

Let $\eta\in\frac12+\zz$. Then for any $a\in\cc^\times$, the power $a^\eta$ can be interpreted as an element in~$\cc^\times/\zz_2$.

Let us restrict our attention to $\Omega$ that are simply connected. We then def\/ine
 \begin{gather}
 \cA^\times(\Omega)/\zz_2:= \big\{(f,-f)\colon f\in\cA^\times(\Omega)\big\},\label{ropo}\end{gather}
 If $f\in\cA^\times(\Omega)$, then $f^\eta$ is well def\/ined as an element of $\cA^\times(\Omega)/\zz_2$.

\begin{Remark} If $\Omega$ is not simply connected, then on the left hand side of~(\ref{ropo}) instead of $\Omega$ we need to put the double cover of $\Omega$. Then $f^\eta$ is still well def\/ined. However we will not use this construction.
\end{Remark}

Let us go back to the setup of Section~\ref{Point transformations}. We can then def\/ine $m_\alpha^\eta\in\cA^\times(\Omega)/\zz_2$. Therefore, (\ref{biva})
can be interpreted as a group of transformations of $\cA^\times(\Omega)/\zz_2$.

A similar remark applies to Section~\ref{Local cocycles}.

\subsection{Generalized symmetries}
Let $\cC$ be a linear dif\/ferential operator on a complex manifold $\Omega$. Let $\alpha\in\Hol(\Omega)$. We say that it is a {\em symmetry} of $\cC$ if\/f
\begin{gather*}\alpha\cC=\cC \alpha.\end{gather*}

Let $m^\sharp,m^\flat\in\cA^\times(\Omega)$. Def\/ine a pair of transformations in $\cA^\times\rtimes\Hol(\Omega)$:
\begin{gather*} \alpha^\sharp :=m^\sharp \alpha, \qquad \alpha^\flat :=m^\flat \alpha.\end{gather*}
We say that a pair $(\alpha^\sharp,\alpha^\flat)$ is a {\em generalized symmetry} of $\cC$ if
\begin{gather*}\alpha^\flat\cC=\cC \alpha^\sharp.\end{gather*}
Clearly, the kernel of $\cC$ is invariant wrt the action of $\alpha^\sharp$:
\begin{gather*}\cC f=0\qquad\text{implies}\quad\cC \alpha^\sharp f=0.\end{gather*}
Generalized symmetries of $\cC$ form a group.

Let $A\in\hol(\Omega)$. We say that it is an {\em infinitesimal symmetry} of $\cC$ if\/f
\begin{gather*}A\cC=\cC A.\end{gather*}

Let $M^\sharp,M^\flat\in\cA(\Omega)$. One can also consider a pair of operators in $\cA\rtimes\hol(\Omega)$
\begin{gather*} A^\sharp := A+M^\sharp,\qquad A^\flat:=A+M^\flat.\end{gather*}
We say that a pair $(A^\sharp,A^\flat)$ is a {\em generalized infinitesimal symmetry} of $\cC$ if
\begin{gather*}A^\flat\cC=\cC A^\sharp.\end{gather*}
Clearly, the kernel of $\cC$ is invariant wrt the action of $A^\sharp$:
\begin{gather*}\cC f=0\qquad \text{implies}\quad \cC A^\sharp f=0.\end{gather*}
Inf\/initesimal generalized symmetries of $\cC$ form a Lie algebra.

\section{Line bundles}\label{s3}

\subsection{Scaling}\label{bund}

A holomorphic bundle $\pi\colon \cV\to\cY$ is called a {\em line bundle} if its f\/ibers are modelled on $\cc^\times$. $\cV$ is equipped with {\em scaling}, a homomorphism $\cc^\times\ni s\mapsto\lambda_s\in\Hol(\cV)$ preserving the f\/ibers, that is, satisfying $\pi\lambda_s=\pi$. The vector f\/ield obtained by dif\/ferentiating $\lambda_s$ is called the {\em vertical vector field} and denoted $V$:
\begin{gather*}\frac{\rmd}{\rmd s}\lambda_s\Big|_{s=1}=:V.\end{gather*}
For $v\in\cV$, we will often simply write $sv$ instead of $\lambda_s v$. We will also write $s=\frac{\lambda_s(v)}{v}$.

Let $\cY_0\subset \cY$ be open. A {\em section based on} $\cY_0$ is a~holomorphic map $\gamma\colon \cY_0\to\cV$ such that $\pi\circ\gamma=\id$. Every section based on $\cY_0$ determines a~trivialization of $\pi^{-1}(\cY_0)$
\begin{gather*}\cY_0\times\cc^\times\ni(y,s)\mapsto s\gamma(y)\in\pi^{-1}(\cY_0).\end{gather*}

\subsection{Vector f\/ields on a line bundle}

Let $\hol^{\cc^\times}(\cV)$ denote the Lie algebra of scaling invariant vector f\/ields, that is,
\begin{gather*}\hol^{\cc^\times}(\cV):=\big\{B\in\hol(\cV)\colon \lambda_s B=B\lambda_s,\; s\in\cc^\times\big\}.\end{gather*}
Let $B\in\hol^{\cc^\times}(\cV)$. Then $B$ determines a unique element of $\hol(\cY)$, which will be denoted $B^\diamond$.

Let $\gamma$ be a section based on $\cY_0$. $B^\diamond$ can be transported by $\gamma$ onto $\gamma(\cY_0)$. Thus we obtain two vector f\/ields on $\gamma(\cY_0)$: $ B\big|_{\gamma(\cY_0)} $ and $\gamma\big( B^\diamond\big)$.
\begin{Proposition}\label{coc1}
For any $v\in \gamma(\cY_0)$, $ B(v)-\gamma( B^\diamond)(v)$ is parallel to $V(v)$. Therefore, there exists $M_B^\gamma\in\cA(\cY_0)$ such that
\begin{gather} B\big(\gamma(y)\big)=M_B^\gamma(y) V \big(\gamma(y)\big)+\gamma\big( B^\diamond\big)(y)
,\qquad y\in\cY_0.\label{prio7}\end{gather}
Moreover,
\begin{gather*}\hol^{\cc^\times}(\cV)\ni B\mapsto M_B^\gamma\end{gather*} is a cocycle.
Hence, for any $\eta\in\cc$,
\begin{gather}\hol^{\cc^\times}(\cV)\ni B\mapsto B^{\gamma,\eta}:=B^\diamond+ \eta M_B^\gamma\in\cA\rtimes\hol(\cY_0)\label{nota1}\end{gather}
is a representation of the Lie algebra of scaling invariant vector fields.
\end{Proposition}

\begin{proof} Every $B\in\hol(\cV_0)$ can be written uniquely as
\begin{gather} B=\tilde M_B^\gamma V+B^\gamma,\label{prio4}
\end{gather}
where $\tilde M_B^\gamma\in\cA(\cV_0)$ and for any $s\in\cc^\times$ the vector f\/ield $\cY_0\ni y\mapsto B^\gamma(s\gamma(y))$ is tangent to the section
$s\gamma(\cY_0)$. Assume now that $B\in\hol^{\cc^\times}(\cV_0)$. This means $[V,B]=0$, which is equivalent to
\begin{gather} \big(V \tilde M_B^\gamma\big) V+[V,B^\gamma]=0.\label{prio}\end{gather}
We also have
\begin{gather}
\gamma(B^\diamond)(y)=B^\gamma\big(\gamma(y)\big),\label{prio1}\end{gather}
and we set
\begin{gather} M_B^\gamma(y):=\tilde M_B^\gamma\big(\gamma(y)\big),\qquad y\in\cY_0.\label{prio2}\end{gather}
Restricting (\ref{prio4}) to $\gamma(\cY_0)$, using (\ref{prio1}) and setting (\ref{prio2}), we obtain~(\ref{prio7}).

Now let $B_1,B_2\in\hol(\cV_0)$. Replacing $B$ in (\ref{prio4}) with $[B_1,B_2]$, we can write
\begin{gather}
[B_1,B_2]=\tilde M_{[B_1,B_2]}^\gamma V+[B_1,B_2]^\gamma.\label{prio11}\end{gather}
If in addition $B_1,B_2\in\hol^{\cc^\times}(\cV_0)$, then using (\ref{prio}) with $B$ replaced with $B_1$ and $B_2$, we obtain
\begin{gather}
[B_1,B_2]=\big(B_1^\gamma\tilde M_{B_2}^\gamma-B_2^\gamma\tilde M_{B_1}^\gamma\big) V
+[B_1^\gamma,B_2^\gamma].\label{prio6}\end{gather}
Comparing (\ref{prio6}) with (\ref{prio11}), we obtain
\begin{gather}
B_1^\gamma\tilde M_{B_2}^\gamma-B_2^\gamma\tilde M_{B_1}^\gamma=\tilde M_{[B_1,B_2]}^\gamma.\label{prio3}\end{gather}
Restricting (\ref{prio3}) to the section $\gamma(\cY_0)$, using~(\ref{prio1}) and setting (\ref{prio2}), we obtain the cocycle relation
\begin{gather*}
B_1^\diamond M_{B_2}^\gamma-B_2^\diamond M_{B_1}^\gamma= M_{[B_1,B_2]}^\gamma.\tag*{\qed}
\end{gather*}
\renewcommand{\qed}{}
\end{proof}

\subsection{Point transformations of a line bundle}

Let $\Hol^{\cc^\times}(\cV)$ denote the group of scaling invariant biholomorphic maps of $\cV$, that is
\begin{gather*}\Hol^{\cc^\times}(\cV):=\big\{\alpha\in\Hol(\cV)\colon \alpha\lambda_s=\lambda_s\alpha,\; s\in\cc^\times\big\}.\end{gather*}
Let $\alpha\in\Hol^{\cc^\times}(\cV)$. Then $\alpha$ determines a unique element of $\Hol(\cY)$, which will be denoted by~${\alpha}^\diamond$.

Let $\gamma$ be a section over $\cY_0$. As in (\ref{skip}), we set $\cY_0^{\alpha^\diamond}:={\alpha^\diamond}(\cY_0)\cap\cY_0$. We def\/ine
$m_\alpha^\gamma\in\cA\big(\cY_0^{\alpha^\diamond}\big)$ by
\begin{gather*}m_\alpha^\gamma(y):=\frac{\gamma(y)} {\alpha\circ\gamma\circ(\alpha^\diamond)^{-1}(y)},\qquad y\in\cY_0^{\alpha^\diamond}.\end{gather*}

\begin{Proposition}\label{coc2}
\begin{gather*}\Hol^{\cc^\times}(\cV)\ni\alpha \mapsto m_\alpha^\gamma\end{gather*}
is a cocycle. Hence for any $\eta\in\zz$
\begin{gather} \Hol^{\cc^\times}(\cV)\ni \alpha\mapsto \alpha^{\gamma,\eta}:=
(m_\alpha^\gamma)^\eta\alpha^\diamond\mathop{\in}_\loc
\cA^\times\rtimes\Hol(\cY_0)\label{nota2}\end{gather}
is a local representation.
\end{Proposition}

\subsection{Homogeneous functions of integer degree}\label{s1a}

As before, $\cY_0\subset\cY$ is open. We set $\cV_0:=\pi^{-1}(\cY_0)$. For $\eta\in\zz$, let $\Lambda^\eta\big(\cV_0\big)$ denote the space of holomorphic functions on $\cV_0$ homogeneous of degree $\eta$, that is, functions $k\in\cA\big(\cV_0)$ satisfying
 \begin{gather} k(sv)=s^\eta
 k(v),\qquad v\in\cV_0,\qquad s\in\cc^\times.\label{seciu}\end{gather}

Clearly, (\ref{seciu}) implies
\begin{gather} Vk=\eta k.\label{deg}\end{gather}

Let $\gamma$ be a section based on $\cY_0$. We then have an obvious map $\psi^{\gamma,\eta}\colon \Lambda^\eta(\cV_0)\to\cA(\cY_0)$: for $k\in \Lambda^\eta(\cV_0)$ we set
\begin{gather}\big(\psi^{\gamma,\eta} k\big)(y):=k\big( \gamma(y)\big),\qquad y\in\cY_0.\label{psi}\end{gather}

$\psi^{\gamma,\eta}$ is bijective and we can introduce its inverse, denoted $\phi^{\gamma,\eta}$, def\/ined for any $f\in\cA(\cY_0)$ by
\begin{gather} \big(\phi^{\gamma,\eta}f\big)\big(s\gamma(y)\big)=s^\eta f(y),\qquad s\in\cc^\times,
\qquad y\in\cY_0.\label{phi}\end{gather}

\begin{Proposition}\label{coc7} With the notation of \eqref{nota1} and \eqref{nota2},
\begin{gather}
B^{\gamma,\eta}:=\psi^{\gamma,\eta} B\phi^{\gamma,\eta}\in\cA\rtimes\hol(\cY_0),\nonumber\\ 
{\alpha}^{\gamma,\eta}:=\psi^{\gamma,\eta} {\alpha}\phi^{\gamma,\eta}\in\cA^\times\rtimes\Hol\big(\cY_1,\cY_2\big).\label{poiu2}
\end{gather}
\end{Proposition}

\subsection{Homogeneous functions of non-integer degree}\label{s1b}

One can try to generalize the above construction to $\eta\in\cc\backslash\zz$. In this case, there is a problem with the def\/inition of functions homogeneous of degree $\eta$, because the power function is multivalued on $\cc^\times$. Therefore, we cannot use $\cV_0:=\pi^{-1}(\cY_0)$. Instead, let us we assume that
$\cV_0\subset\cV$ is open, connected, $\pi(\cV_0)=\cY_0$ and $\pi^{-1}(y)\cap\cV_0$ is simply connected for any $y\in\cY_0$. We say that $k\in\Lambda^\eta(\cV_0)$ if $k\in\cA\big(\cV_0)$ and
\begin{gather} k(sv)=s^\eta k(v),\qquad v,sv\in\cV_0,\qquad s\in\cc^\times.\label{seciu1}\end{gather}
Note that (\ref{seciu1}) is unambiguous, because, for any $y\in\cY_0$, on $\pi^{-1}(y)\cap \cV_0$ we have a~unique continuation of holomorphic functions. (\ref{deg}) still holds.

Let $\gamma$ be a section based on $\cY_0$ whose image is contained in $\cV_0$. $\psi^{\gamma,\eta}$ is still bijective and we can introduce its inverse, denoted $\phi^{\gamma,\eta}$, def\/ined for any $f\in\cA(\cY)$ by
\begin{gather*} \big(\phi^{\gamma,\eta}f\big)\big(s\gamma(y)\big)=s^\eta f(y),\qquad s\in\cc^\times,
\qquad y\in\cY_0,\qquad s\gamma(y)\in\cV_0.
\end{gather*}
With this def\/inition, (\ref{poiu2}) is still true.

\section{Complex Euclidean spaces}\label{s4}

\subsection{Linear transformations}

Let us f\/irst consider the vector space $\cc^n$ without the Euclidean structure.
The af\/f\/ine general linear Lie algebra $\cc^n\rtimes \gl(\cc^n)$ can be identif\/ied with the subalgebra of $\hol(\cc^n)$ spanned by
\begin{gather*} \partial_{y^j},\quad j=1,\dots,n,\qquad
 y^i\partial_{y^j},\quad i,j=1,\dots,n. \end{gather*}
 Similarly, the af\/f\/ine general linear group $\cc^n\rtimes \GL(\cc^n)$ is a subgroup of $\Hol(\cc^n)$.

We will have a special notation for {\em the generator of dilations}
\begin{gather*}
D_{\cc^n}:=\sum_{i=1}^ny^i\p_{y^i}.\end{gather*}
Obviously,
\begin{alignat*}{3}
& D_{\cc^n}B =BD_{\cc^n},\qquad && B\in\gl(\cc^n),&\\
& D_{\cc^n}\alpha=\alpha D_{\cc^n},\qquad && \alpha\in\GL(\cc^n).&
\end{alignat*}

Let $\sigma=(\sigma_1,\dots,\sigma_n)$ be a permutation of $\{1,\dots,n\}$. Then
\begin{gather*}\sigma\big(y^1,\dots, y^n\big):=\big(y^{\sigma_1^{-1}},\dots, y^{\sigma_n^{-1}}\big)\end{gather*}
def\/ines an element of $\GL(\cc^n)$. On the level of point transformations it acts as
\begin{gather*}(\sigma f)\big(y^{1},\dots,y^n\big):=
f\big(y^{\sigma_1},\dots,y^{\sigma_n}\big).\end{gather*}

\subsection{Bilinear scalar product}
Suppose that
\begin{gather*}\cc^n\ni y,x\mapsto \langle y|x\rangle=\sum_{i,j}g_{i,j}y^ix^{ j}\end{gather*}
is a nondegenerate symmetric bilinear form on~$\cc^n$ called the {\em scalar product}. Clearly, if we know the {\em square} of each vector \begin{gather*}\langle y|y\rangle=\sum_{i,j}g_{i,j}y^iy^{j},\end{gather*}
we have the complete information about the scalar product.

$[g^{ij}]$ will denote the inverse of $[g_{ij}]$.
The {\em orthogonal Lie algebra of $\cc^n$}, understood as a Lie subalgebra of $\hol(\cc^n)$, is def\/ined as
\begin{gather*}\so(\cc^n):=\big\{B\in \gl(\cc^n)\colon B\langle y|y\rangle =0\big\}.\end{gather*}
For $i,j=1,\dots,n$, def\/ine
\begin{gather*}B_{i,j}:=\sum_k(g_{j,k}y_i\p_{y_k}-g_{i,k}y_k\p_{y_j}).\end{gather*}
 $\{B_{i,j}\colon i<j\}$ is a basis of $\so(\cc^n)$. Clearly, $B_{i,j}=-B_{j,i}$, in particular $B_{i,i}=0$.

Likewise, recall that the {\em orthogonal} and the {\em special orthogonal group of $\cc^n$} is def\/ined as
\begin{gather*}
\mathrm{O}(\cc^n):=\big\{{\alpha}\in \GL(\cc^n)\colon \langle {\alpha}y|{\alpha}x\rangle=\langle y|x\rangle,\; y,x\in \cc^n\big\},\\
\SO(\cc^n):=\big\{\alpha\in\mathrm{O}(\cc^n)\colon \det\alpha=1\big\}.
\end{gather*}

 We def\/ine
\begin{gather*}\text{the Laplacian}\quad\Delta_{\cc^n}:=\sum_{i,j=1}^ng^{i,j}\partial_{y^i}\partial_{y^j}, \qquad \text{and} \\
\text{the Casimir operator}\quad\cC_{\cc^n}:=\frac{1}{2}\sum_{i,j=1}^ng^{i,k}g^{j,l}B_{i,j}B_{k,l}.
\end{gather*}
Clearly,
\begin{alignat*}{5}
& \Delta_{\cc^n}B=B\Delta_{\cc^n},\quad && B\in\cc^n\rtimes\so(\cc^n), \quad \qquad && \cC_{\cc^n}B=B\cC_{\cc^n},\quad && B\in\so(\cc^n), & \\
&\Delta_{\cc^n}\alpha=\alpha\Delta_{\cc^n},\quad && \alpha \in\cc^n\rtimes{\rm O}(\cc^n), \quad \qquad && \cC_{\cc^n}\alpha=\alpha\cC_{\cc^n},\quad && \alpha \in{\rm O}(\cc^n).&
\end{alignat*}
Note the identity
\begin{gather}
\Delta_{\cc^n}=\frac{1}{\langle y|y\rangle}\big(D_{\cc^n}^2+(n-2)D_{\cc^n}+\cC_{\cc^n}\big).\label{polar}\end{gather}
We will denote by $\cS^{n-1}(R)$ the {\em $($complex$)$ sphere in $\cc^n$ of squared radius~$R$}, that is
\begin{gather*}\cS^{n-1}(R):=\big\{y\in\cc^n\colon \langle y|y\rangle=R\big\}.\end{gather*}
We also introduce the {\em null quadric}
\begin{gather} \cV^{n-1}:=\cS^{n-1}(0)\backslash\{0\}.\label{null}\end{gather}

\subsection{Split coordinates}

The coordinates that we describe in this subsection are particularly convenient for the analysis of $\so(\cc^n)$ and $\mathrm{O}(\cc^n)$. Let $n=2m$ if $n$ is even and $n=2m+1$ if $n$ is odd. Set
\begin{gather*}
I_n:=\begin{cases}\{-1,1,\dots,-m,m\},&\text{for even} \ n,\\
\{0,-1,1,\dots,-m,m\},&\text{for odd} \ n.
\end{cases}\end{gather*}
The coordinates in $\cc^n$ will be labelled by $I_n$, so that the square of $y=[y^i]_{i\in I_n}$ is given by
\begin{gather*}\langle y|y\rangle =\sum\limits_{i\in I_n}y_{i}y_{-i}
=\begin{cases}
\displaystyle \sum\limits_{i=1}^m2y_{-i}y_i,&\text{for even} \ n,\\
\displaystyle y_0^2+\sum\limits_{i=1}^m2y_{-i}y_i&\text{for odd} \ n.
\end{cases}\end{gather*}
Clearly, $g_{i,j}=g^{i,j}=\delta_{i,-j}$.

For $n=2m$, $\so(\cc^n)$ has a basis consisting of
\begin{alignat}{3}
& N_i:=B_{i-i}=y_{-i}\p_{y_{-i}}-y_{i}\p_{y_{i}},\qquad && j=1,\dots,m,& \label{posa1}\\
& B_{ij}:=y_{-i}\p_{y_{j}}-y_{-j}\p_{y_{i}},\qquad && 1\leq|i|<|j|\leq m.& \label{posa2}
\end{alignat}
For $n=2m+1$ we have to add
\begin{gather}
B_{0j}=y_{0}\p_{y_{j}}-y_{-j}\p_{y_{0}},\qquad |j|=1,\dots,m.\label{posa3}
\end{gather}

The subalgebra of $\so(\cc^n)$ spanned by (\ref{posa1}) is its {\em Cartan algebra}. (\ref{posa2}), and in the odd case also (\ref{posa3}), are its {\em root operators}:
\begin{gather*}
[N_k,B_{i,j}]=(\sgn(i)\delta_{k,|i|}+\sgn(j)\delta_{k,|j|})B_{i,j},\qquad [N_k,B_{0,j}]=\sgn(j)\delta_{k,|j|}B_{i,j}.
\end{gather*}

We have
\begin{gather*}
\Delta_{\cc^n}=\sum_{i\in I_n}\partial_{y_{i}}\partial_{y_{-i}},\qquad \cC_{\cc^n}=\frac{1}{2}\sum_{i,j\in I_n}B_{i,j}B_{{-i,-j}}.
\end{gather*}

\subsection{Weyl symmetries}\label{Weyl symmetries}

In our applications of the group invariance we will restrict ourselves only to the so-called ``Weyl symmetries''. It will be convenient to consider ``Weyl symmetries'' contained not only in~$\SO(\cc^n)$, but in the whole $\mathrm{O}(\cc^n)$.

Permutations of $I_{2m}$ that preserve its decomposition into pairs $\{-1,1\},\dots,\{-m,m\}$ correspond to a subgroup of $\mathrm{O}(2m)$ that will be denoted $\W(\cc^{2m})$. It is isomorphic to ${\zz}_2^m\rtimes S_m$. It is generated by two kinds of transformations: $\tau_j$, $j=1,\dots,m$, which swap the elements of the $j$th pair, and permutations from~$S_m$, which permute the pairs. If $\sigma=(\sigma_1,\dots,\sigma_m)\in S_m$, then
\begin{gather*}
\sigma f(y_{-1},y_{1},\dots,y_{-m},y_m):=f(y_{-\sigma_1},y_{\sigma_1},\dots,y_{-\sigma_m},y_{\sigma_m}).
\end{gather*}
For $j=1,\dots,n$,
\begin{gather*}
\tau_j f(y_{-1},y_{1},\dots,y_{-j},y_j,\dots):=f(y_{-1},y_{1},\dots,y_{j},y_{-j},\dots).
\end{gather*}

We have{\samepage
 \begin{alignat*}{3}
& \sigma B_{i,j}\sigma^{-1}=B_{\sigma_i,\sigma_j},\qquad&& \tau_k B_{i,j}\tau_k^{-1}=(-1)^{\delta_{|i|,k}+\delta_{|j|,k}}B_{i,j},&\\
& \sigma N_j\sigma^{-1}=N_{\sigma_j},\qquad&& \tau_k N_{j}\tau_k^{-1} =(-1)^{\delta_{j,k}}N_j.&
\end{alignat*}}

Using $\cc^{2m+1}=\cc\oplus\cc^{2m}$, we embed $\W(\cc^{2m})$ in $\mathrm{O}(\cc^{2m+1})$. We also introduce a transformation $\tau_0\in \mathrm{O}(\cc^{2m+1})$ given by
\begin{gather*}
\tau_0f(y_0,y_{-1},y_1,\dots,y_{-m},y_m) :=f(-y_0,y_{-1},y_1,\dots,y_{-m},y_m).
\end{gather*}
Clearly, $\tau_0$ commutes with $\W(\cc^{2m})$. The group $\W(\cc^{2m+1})$ is def\/ined as the group generated by~$\W(\cc^{2m})$ and~$\tau_0$, and is isomorphic to $\zz_2\times\zz_2^m\rtimes S_m$. We have for $i,j=1,\dots,m$
\begin{gather*}
\tau_0 B_{0,j}\tau_0^{-1}=-B_{0,j},\qquad \tau_0 B_{i,j}\tau_0^{-1}=B_{i,j},\qquad \tau_0N_j\tau_0^{-1}=N_j.
\end{gather*}

In both even and odd cases $\W(\cc^n)$ acts as a group of automorphisms
of $\so(\cc^n)$
 leaving invariant the Cartan algebra. To compute the determinant of
 elements of $\W(\cc^n)$ it suf\/f\/ices to remember that
$\det\sigma=1$ for $\sigma\in S_m$ and $\det\tau_j=-1$.

\section{Conformal invariance}
\label{Conformal invariance}

The main subject of this section is the description of generalized (inf\/initesimal) symmetries of the Laplace equation
\begin{gather} \Delta_{\cc^n}f=0.\label{lapl}\end{gather}
We will see in particular that the Lie algebra of generalized symmetries is $\so(\cc^{n+2})$. We will see that it is convenient to start the description of these symmetries from the space $\cc^{n+2}$, which we will call the {\em extended space}. The space $\cc^n$ will be embedded inside $\cc^{n+2}$ as a section of the null quadric. We will see how the Laplacian $\Delta_{\cc^{n+2}}$ reduces to the Laplacian~$\Delta_{\cc^n}$.

\subsection{Conformal invariance of Riemannian manifolds}

Suppose that a (complex) manifold $\Omega$ is equipped with a nondegenerate holomorphic covariant 2-tensor f\/ield $g$, called the {\em $($complex$)$ metric tensor}. We will say that $(\Omega,g)$ is a~{\em $($complex$)$ Riemannian space}.

Thus if $A,B\in\hol(\Omega)$, then we have a holomorphic function
\begin{gather*}\Omega\ni y\mapsto g(A,B)(y)=g_{i,j}(y)A^i(y)B^{ j}(y)\end{gather*}
called the {\em scalar product}.

Let ${\alpha}\in\Hol(\Omega)$. We can transport $g$ by ${\alpha}$:
\begin{gather*}{\alpha}(g)(A,B):=g\big({\alpha}^{-1}(A),{\alpha}^{-1}(B)\big).\end{gather*}
We say that ${\alpha}$ is {\em conformal} if there exists $m_{\alpha}\in\cA^\times(\Omega)$ such that
\begin{gather*}{\alpha} (g)=m_{\alpha} g.\end{gather*}
Let $\Cf(\Omega)$ denote the group of conformal maps on $(\Omega,g)$. One can check that{\samepage
\begin{gather*}\Cf(\Omega)\times\Omega\ni({\alpha},y)\mapsto
m_{\alpha}(y)\in\cA^\times(\Omega)\end{gather*} is a cocycle.}

Let $C\in\hol(\Omega)$. The {\em Lie derivative} of $g$ in the direction of $C$ is denoted $Cg$ and def\/ined by
\begin{gather*}(Cg)(A,B):=C (g(A,B) )-g([C,A],B)-g(A,[C,B]).\end{gather*}
We say that $C$ is {\em infinitesimally conformal} if there exists $M_C\in\cA(\Omega)$ such that
\begin{gather*}Cg=M_Cg.\end{gather*}
Let $\cf(\Omega)$ denote the Lie algebra of inf\/initesimally conformal f\/ields. One can check that \begin{gather*}
\cf(\Omega)\times\Omega\ni(C,y)\mapsto M_C(y)\in\cA(\Omega)\end{gather*}
 is a cocycle.

We say that a manifold $\Omega$ has a {\em conformal structure}, if it is covered by a family of open sets~$\Omega_i$ equipped with bilinear scalar products~$g_i$ such that on $\Omega_i\cap\Omega_j$ we have
\begin{gather*}g_i=\rho_{i,j} g_j\end{gather*}
for some $\rho_{i,j}\in\cA^\times(\Omega_i\cap\Omega_j)$. Clearly, a~Riemannian structure on $\Omega$ is not necessary to def\/ine $\Cf(\Omega)$ and $\cf(\Omega)$~-- we need only a conformal structure on~$\Omega$.

\subsection{Null quadric}

Consider the extended space, that is, the complex Euclidean space $\cc^{n+2}$. The central role will be played by the representations
\begin{gather}
\so\big(\cc^{n+2}\big)\ni B \mapsto B\in \hol\big(\cc^{n+2}\big),\label{popo1}\\
\mathrm{O}\big(\cc^{n+2}\big)\ni {\alpha} \mapsto {\alpha}\in \Hol\big(\cc^{n+2}\big),\label{popo2}
\end{gather}
and the symmetry
\begin{alignat}{3}
&B\Delta_{\cc^{n+2}} =\Delta_{\cc^{n+2}}B,\qquad&& B\in\so\big(\cc^{n+2}\big),&\label{syme1a}\\
&{\alpha}\Delta_{\cc^{n+2}}=\Delta_{\cc^{n+2}}{\alpha},\qquad&& {\alpha}\in \mathrm{O}\big(\cc^{n+2}\big).&\label{syme2a}
\end{alignat}

As in (\ref{null}), we introduce
\begin{gather*}\cV:=\big\{z\in\cc^{n+2} \colon z\neq0,\; \langle z|z\rangle=0\big\}\end{gather*}
called the {\em null quadric}.

Multiplication by $s\in\cc^\times$ preserves $\cV$. Therefore, we can def\/ine the {\em projective quadric}
\begin{gather*}\cY:=\cV/\cc^\times=\big\{\cc^\times z\colon z\in\cV\big\}.\end{gather*}
It is an $n$-dimensional complex manifold. Let $\pi\colon \cV\to\cY$ denote the natural projection. Clearly, $\cV$ is a complex line bundle over~$\cY$. As in Section~\ref{bund}, the multiplication by $s$ will be often denoted by $\cc^\times\ni s\mapsto \lambda_s\in\Hol(\cV)$ and the corresponding vertical vector f\/ield by $V\in\hol(\cV)$.

We can restrict (\ref{popo1}) and (\ref{popo2}) to $\cV$ and note that they are scaling invariant. Thus, we have natural embeddings
\begin{gather}
\so\big(\cc^{n+2}\big)\ni B\mapsto B\in\hol^{\cc^\times }(\cV),\label{pyo1}\\
\mathrm{O}\big(\cc^{n+2}\big) \ni {\alpha} \mapsto {\alpha} \in\Hol^{\cc^\times }(\cV).\label{pyo2}
\end{gather}
(Recall that $\hol^{\cc^\times }(\cV)$, resp.\ $\Hol^{\cc^\times }(\cV)$ denote the scaling invariant holomorphic vector f\/ields, resp.\ bijections). Therefore, (\ref{pyo1}) and (\ref{pyo2}) induce their actions on~$\cY$:
\begin{gather}
\so\big(\cc^{n+2}\big)\ni B\mapsto B^\diamond \in\hol(\cY),\label{so0}\\
\mathrm{O}\big(\cc^{n+2}\big)\ni {\alpha} \mapsto {\alpha}^\diamond \in\Hol(\cY).\label{SO0}
\end{gather}

\subsection{Conformal invariance of projective quadric}

Let $g$ denote the restriction of the metric tensor on $\cc^{n+2}$ to $\cV$. Note that the null space of $g$ is 1-dimensional and is spanned by the vertical f\/ield~$V$. In particular,
\begin{gather} g(V,A)=0,\qquad A\in\hol(\cV).\label{one1}\end{gather}
Moreover, the scaling scales the metric tensor:
\begin{gather}
\lambda_{s}(g) =s^2 g,\qquad s\in\cc^\times,\label{one2}\\
Vg =2g.\label{one3}
\end{gather}
Using (\ref{one1}), we can extend (\ref{one2}) to multiplication by nonconstant functions:
\begin{Proposition}\label{gfgh} Let $\cU$ be open in $\cY$, $m\in\cA^\times(\cU)$. Define $\lambda_m\in\Hol\big(\pi^{-1}(\cU)\big)$ by
\begin{gather*}\lambda_m(z):=m(\pi(z))z,\qquad z\in\pi^{-1}(\cU).\end{gather*}
Then $\lambda_{m}g=m^2g$. In particular, for any section $\gamma$, the restriction $\lambda_m\colon \gamma(\cU)\to m\circ\pi\gamma(\cU)$ is conformal.
\end{Proposition}

Let $\cU$ be open in $\cY$. Let $\gamma$ be a section over $\cU$. The tensor $g$ restricted to $\gamma(\cU)$ is non\-de\-ge\-nerate. We can transport it by $\gamma^{-1}$ onto $\cU$. This way we endow $\cU$ with a metric tensor.

For $i=1,2$, let $\cU_i$ be two open subsets of $\cY$ equipped with sections $\gamma_i$. Let $g_i$ be the corresponding complex Riemannian tensors. Obviously, there exists $\rho_{2,1}\in\cA^\times(\cU_1\cap\cU_2)$ such that
\begin{gather*}\gamma_2(y)=\rho_{2,1}(y)\gamma_1(y),\qquad y\in\cU_1\cap\cU_2.\end{gather*}
Therefore, by Proposition~\ref{gfgh},
\begin{gather*}g_2=\big(\gamma_2^{-1}\gamma_1\big) g_1=\rho_{2,1}^2 g_1.\end{gather*}

Cover $\cY$ with open subsets $\cU_i$, $i=1,\dots,N$, equipped with sections $\gamma_i$. Let $g_i$ be the corresponding Riemannian tensors on~$\cU_i$. Then on $\cU_i\cap\cU_j$ we have
\begin{gather*}g_j=\rho_{j,i}^2 g_i,\qquad \rho_{j,i}\in\cA^\times(\cU_i\cap\cU_j).\end{gather*}
This way we endow $\cY$ with a conformal structure. It is easy to see that it does not depend on the choice of the covering and sections.

\begin{Proposition} \eqref{so0} is infinitesimally conformal and \eqref{SO0} is conformal.\end{Proposition}
\begin{proof} Let $\gamma$ be a section over $\cU\subset\cY$.

Let $B\in \so(\cc^{n+2})$. By Proposition~\ref{coc1}, there exists $M_B^\gamma\in\cA(\cU)$ such that
\begin{gather*}\gamma(B^\diamond)= \big(B -M_B^\gamma V\big)\big|_{\gamma(\cU)}.\end{gather*}
Now $Bg=0$ and $Vg=2g$. Hence
\begin{gather*}\gamma(B^\diamond) g= -2M_B^\gamma g.\end{gather*}
Therefore, $\gamma(B^\diamond)$ is inf\/initesimally conformal on~$\gamma(\cU)$. Hence $B^\diamond$ is inf\/initesimally conformal on~$\cU$.

Let ${\alpha}\in \mathrm{O}(\cc^{n+2})$. Clearly, $\gamma_{\alpha}:={\alpha}\circ\gamma\circ({\alpha}^\diamond)^{-1}$ is a~section based on ${\alpha}^\diamond(\cU)$. Therefore, there exists $m\in\cA^\times\big({\alpha}^\diamond(\cU)\cap \cU\big)$ such that
\begin{gather}
\gamma(y)=m(y)\gamma_{\alpha}(y),\qquad y\in{\alpha}^\diamond(\cU)\cap \cU.
\label{qwas}
\end{gather}
Def\/ine now $\beta\colon \pi^{-1}\big(\cU\cap ({\alpha}^\diamond)^{-1}(\cU)\big)\to \pi^{-1}\big({\alpha}^\diamond(\cU)\cap \cU\big)$
 \begin{gather*}\beta(z):=m\big(\pi(z)\big)\alpha(z).\end{gather*}
Substituting $z=\gamma(y)$ and using (\ref{qwas}) we obtain
\begin{gather*}\beta(\gamma(y))=m(y)\alpha\big(\gamma(y)\big)=\gamma(\alpha^\diamond(y)).\end{gather*}
Thus $\beta$ maps $\gamma\big(\cU\cap ({\alpha}^\diamond)^{-1}(\cU)\big)$ onto $\gamma\big({\alpha}^\diamond(\cU)\cap \cU\big)$. Let $\beta^\gamma$ denote this restriction of $\beta$. By Proposition~\ref{gfgh}, $\beta^\gamma$ is conformal. Clearly, it satisf\/ies the identity
\begin{gather*}\beta^\gamma=\gamma\alpha^\diamond\gamma^{-1}.\end{gather*}
Hence $\alpha^\diamond$ is conformal.
\end{proof}

\subsection{Conformal invariance of complex Euclidean space}

Fix a vector $e\in\cV$, and set
\begin{gather*} \cV_{e}:=\{z\in\cV\colon \langle z|e\rangle\neq0\}.
\end{gather*}
Clearly, $\cV_{e}$ is invariant with respect to the action of $\cc^\times$ and $\cY_{e}:=\cV_{e}/\cc^\times$ is an open dense subset of~$\cY$. We have a~natural section of the line bundle $\cV\to\cY$ based on $\cY_e$:
\begin{gather}
\cY_e\ni \cc^\times z\mapsto\frac{z}{\langle z|e\rangle}\in\cV_e.\label{secto}\end{gather}

Next choose a vector $d\in\cV$ such that $\langle e|d\rangle=1$. Clearly, $\{e,d\}^\perp$ is $n$-dimensional. It will be convenient further on to choose coordinates $(z_i)_{i\in I_n}$ in $\{e,d\}^\perp$. Each $z\in\cc^{n+2}$ can be written as
 \begin{gather*}
z=\left[\begin{matrix}(z_i)_{i\in I_n}\\z_{-m-1}\\z_{m+1}\end{matrix}\right]=(z_i)_{i\in I_n}+z_{-m-1} e+z_{m+1}d,
\end{gather*}
where $(z_i)_{i\in I_n}\in \{e,d\}^\perp\simeq\cc^n$, $(z_{-m-1},z_{m+1})\in\cc^2$.

Using $y=(y_i)_{i\in I_n}$ as generic variables in $\cc^n$, and noting that $z_{m+1}=\langle z|e\rangle$, we see that~$\cY_{e}$ can be identif\/ied with~$\cc^n$ through the map
 \begin{gather*}\cc^n\ni y\mapsto \cc^\times\left[\begin{matrix}y\\-\frac{\langle y|y\rangle }{2}\\ 1\end{matrix}\right]\in\cY_{e}.
\end{gather*}
With this identif\/ication, the section (\ref{secto}) can be written as
 \begin{gather}\cc^n\ni y\mapsto
\left[\begin{matrix}y\\-\frac{\langle y|y\rangle }{2}\\ 1\end{matrix}\right]\in\cV_{e}.\label{secto2}\end{gather}
\begin{Remark}
The above discussion shows that $\cY_{e}$ has a natural structure of the {\em affine $n$-dimen\-sional Euclidean space}. The choice of~$d\in\cV_e$ (which does not inf\/luence the def\/inition of $\cY_e$) determines the origin of coordinates in $\cY_e$.
\end{Remark}

The stabilizer of $\left[\begin{matrix}0\\1\\0\end{matrix}\right]=e\in\cV$ inside $\mathrm{O}(\cc^{n+2})$ is isomorphic to $\cc^n\rtimes \mathrm{O}(\cc^n)$, and is given by
\begin{gather*}\left[\begin{matrix}
\one &0&w\\
- w^\#&1&-\frac12\langle w|w\rangle\\
0&0&1\end{matrix}\right]\left[\begin{matrix}
\beta&0&0\\
0&1&0\\
0&0&1\end{matrix}\right],\qquad \beta\in \mathrm{O}\big(\cc^n\big),\qquad w\in\cc^n.\end{gather*}
The \looseness=-1 stabilizer of $\cc^\times e\in\cY$ inside $\mathrm{O}(\cc^{n+2})$ is isomorphic to $\cc^n\rtimes\big( \mathrm{O}(\cc^n)\times \mathrm{O}(\cc^2)\big)$ and is given by
\begin{gather*}\left[\begin{matrix}
\one&0&w\\
- w^\#&1&-\frac12\langle w|w\rangle\\
0&0&1\end{matrix}\right]\left[\begin{matrix}
\beta&0&0\\
0&s&0\\
0&0&s^{-1}\end{matrix}\right],\qquad \beta\in \mathrm{O}\big(\cc^n\big),\qquad w\in\cc^n,\qquad s\in\cc^\times.\end{gather*}

\subsection{Laplacian on bundle of homogeneous functions}

Let $\cV_0$ be an open subset of $\cV$ and $\eta\in\cc$. We def\/ine $\Lambda^\eta(\cV_0)$ to be the set of holomorphic functions on $\cV_0$ homogeneous of degree $\eta$. (See Sections~\ref{s1a} and~\ref{s1b} for a discussion.)

Clearly, $B\in \so(\cc^{n+2})$ preserves $\Lambda^\eta(\cV_0)$. We will denote by $B^{\diamond,\eta}$ the restriction of $B\in\hol(\cV)$ to $\Lambda^\eta(\cV_0)$.

Clearly, ${\alpha}\in \mathrm{O}(\cc^{n+2})$ maps $\Lambda^\eta(\cV_0)$ onto $\Lambda^\eta\big({\alpha}(\cV_0)\big)$. We will denote by ${\alpha}^{\diamond,\eta}$ the restriction of ${\alpha}$ to $\Lambda^\eta(\cV_0)$. Thus we have representations
\begin{gather*}
\so\big(\cc^{n+2}\big) \ni B \mapsto B^{\diamond,\eta},\qquad 
\mathrm{O}\big(\cc^{n+2}\big) \ni {\alpha} \mapsto {\alpha}^{\diamond,\eta}.
\end{gather*}

We f\/ind the following theorem curious because it allows in some situations to restrict a {\em second order} dif\/ferential operator to a~submanifold.

\begin{Theorem} \label{wer}Let $\Omega\subset\cc^{n+2}$ be open. Let $ K\in\cA(\Omega)$ be homogeneous of degree~$\frac{2-n}{2}$ such that
\begin{gather*} K\big|_{ cV\cap\Omega}=0.\end{gather*}
Then
\begin{gather*}\Delta_{\cc^{n+2}} K\big|_{\cV\cap\Omega}=0.\end{gather*}
\end{Theorem}

\begin{proof}
 We will give two proofs. Each of the proofs will provide us with a formula, which will be useful later on.

{\bf Method I.} We use the decomposition $\cc^{n+2}=\cc^n\oplus\cc^2$. As usual, we will denote by $\langle z|z\rangle $ the square of $z\in\cc^{n+2}$, by~$D_{\cc^{n+2}}$ the generator of dilations, by $\cC_{\cc^{n+2}}$ the Casimir operator and by $\Delta_{\cc^{n+2}}$ the Laplacian on~$\cc^{n+2}$. We will also need the corresponding objects on $\cc^n$: $\langle z|z\rangle_{\cc^n}$, $D_{\cc^n}$, $\cC_{\cc^n}$, and $\Delta_{\cc^n}$. We will write
\begin{alignat*}{3}
& \langle z|z\rangle_{m+1}:=2z_{-m-1}z_{m+1},\qquad&&  \Delta_{m+1}:=2\p_{z_{-m-1}}\p_{z_{m+1}},&\\
& D_{m+1}:=z_{-m-1}\p_{z_{-m-1}}+z_{m+1}\p_{z_{m+1}},\qquad && N_{m+1}:=z_{-m-1}\p_{z_{-m-1}}-z_{m+1}\p_{z_{m+1}}.&
\end{alignat*}
Thus we have
\begin{gather*}
\langle z|z\rangle =\langle z|z\rangle_{\cc^n}+\langle z|z\rangle_{m+1},\qquad
\Delta_{\cc^{n+2}} =\Delta_{\cc^n}+\Delta_{m+1},\qquad D_{\cc^{n+2}} =D_{\cc^n}+D_{m+1}.
\end{gather*}

The following identity is a consequence of (\ref{polar}):
\begin{gather}\nonumber
\langle z|z\rangle_{\cc^n}\Delta_{\cc^{n+2}} =
\langle z|z\rangle_{\cc^n}\Delta_{\cc^n}+\big(\langle
z|z\rangle-\langle z|z\rangle_{m+1}\big)\Delta_{m+1}\\ \nonumber
\hphantom{\langle z|z\rangle_{\cc^n}\Delta_{\cc^{n+2}}}{}
=\cC_{\cc^n}+\Bigl(D_{\cc^n}-\frac{2-n}{2}\Bigr)^2-
\left(\frac{2-n}{2}\right)^2 +\langle z|z\rangle \Delta_{m+1}-
D_{m+1}^2+N_{m+1}^2\\ \nonumber
\hphantom{\langle z|z\rangle_{\cc^n}\Delta_{\cc^{n+2}}}{}
=\langle z|z\rangle \Delta_{m+1}+
\left(D_{\cc^n}-\frac{2-n}{2}-D_{m+1}\right) \left(D_{\cc^{n+2}}-\frac{2-n}{2}\right)\\
\hphantom{\langle z|z\rangle_{\cc^n}\Delta_{\cc^{n+2}}=}{}
-\left(\frac{2-n}{2}\right)^2+\cC_{\cc^n}+N_{m+1}^2.\label{deq}
\end{gather}
$\bigl(\frac{2-n}{2}\bigr)^2$ is a scalar. $\cC_{\cc^n}$ and $N_{m+1}^2$ are polynomials in elements of $\so(\cc^{n+2})$. $\cV$ is tangent to $\so(\cc^{n+2})$. Therefore, all operators in the last line of (\ref{deq}) can be restricted to $\cV$. The operator $D_{\cc^{n+2}}-\frac{2-n}{2}$ vanishes on functions in $\Lambda^{\frac{2-n}{2}}(\Omega)$. The operator $\langle z|z\rangle\Delta_{m+1} $ is zero when restricted to $\cV$.

{\bf Method II.} We write $\cc^{n+2}=\cc^{n+1}\oplus\cc$ with the distinguished variable denoted by~$t$. We assume that the square of $z\in\cc^{n+2}$ is given by
\begin{gather*}\langle z|z\rangle=\langle z|z\rangle_{\cc^{n+1}}+t^2.\end{gather*}
We will use various operators on $\cc^{n+1}$: $D_{\cc^{n+1}}$, $\cC_{\cc^{n+1}}$, and $\Delta_{\cc^{n+1}}$. We have
\begin{gather*}
D_{\cc^{n+2}} = D_{\cc^{n+1}}+t\partial_t,\qquad \Delta_{\cc^{n+2}}= \Delta_{\cc^{n+1}}+\partial_t^2.
\end{gather*}
We have the following identity
\begin{gather}\notag
\langle z|z\rangle_{\cc^{n+1}}\Delta_{\cc^{n+2}}
=\langle z|z\rangle_{\cc^{n+1}}\Delta_{\cc^{n+1}}+\langle z|z\rangle\partial_{t}^2-t^2\partial_{t}^2\\ \nonumber
\hphantom{\langle z|z\rangle_{\cc^{n+1}}\Delta_{\cc^{n+2}}}{}
=\cC_{\cc^{n+1}}+\left(D_{\cc^{n+1}}+\frac{n-1}{2}\right)^2-\left(\frac{n-1}{2}\right)^2 +\langle z|z\rangle\partial_t^2\\
\hphantom{\langle z|z\rangle_{\cc^{n+1}}\Delta_{\cc^{n+2}}=}{}
- \left(t\partial_{t}-\frac12\right)^2+\left(\frac12\right)^2\nonumber\\
\hphantom{\langle z|z\rangle_{\cc^{n+1}}\Delta_{\cc^{n+2}}}{}
=\langle z|z\rangle \partial_{t}^2+\left(D_{\cc^{ n+1}}+\frac{n}{2}-t\p_{t}\right)\left(D_{\cc^{n+2}}-\frac{2-n}{2}\right) \nonumber\\
\hphantom{\langle z|z\rangle_{\cc^{n+1}}\Delta_{\cc^{n+2}}=}{}
-\left(\frac{n-2}{2}\right)\frac{n}{2}+\cC_{\cc^{n+1}}.\label{deq3}
 \end{gather}
Then we argue similarly as in Method~I.
\end{proof}

Using Theorem \ref{wer} we can restrict the Laplacian to functions in $\Lambda^\eta(\cV_0)$ for $\eta=\frac{2-n}{2}$. More precisely, we introduce the following def\/inition.

Let $k\in\Lambda^{\frac{2-n}{2}}(\cV_0)$. Let $\Omega$ be any open subset of $\cc^{n+2}$ such that $\cV_0=\Omega\cap\cV$, let $ K\in\cA(\Omega)$ be homogeneous of degree $\frac{2-n}{2}$ and
\begin{gather*}k= K\big|_{\cV_0}.\end{gather*}
(We can always f\/ind such $\Omega$ and~$K$.) Note that $\Delta_{\cc^{n+2}} K$ is homogeneous of degree $\frac{-2-n}{2}$. We set
\begin{gather*}\Delta_{\cc^{n+2}}^\diamond k:=\Delta_{\cc^{n+2}} K\big|_{\cV_0}.\end{gather*}

By Theorem \ref{wer}, the above def\/inition does not depend on the choice of $\Omega$ and $ K$ and def\/ines a~map
\begin{gather*}
\Delta_{\cc^{n+2}}^\diamond\colon \ \Lambda^{\frac{2-n}{2}}(\cV_0)\to \Lambda^{\frac{-2-n}{2}}(\cV_0).
\end{gather*}

\begin{Remark} Let us explain the notation $\Delta_{\cc^{n+2}}^\diamond$ for the {\em reduced Laplacian}. We do not put the degree of homogeneity $\eta=\frac{2-n}{2}$ as a superscript, because it is f\/ixed by Theorem~\ref{wer}, unlike in the case of the representations of $\so(\cc^{n+2})$ and ${\rm O}(\cc^{n+2})$. The subscript $\cc^{n+2}$ is a little confusing, because $\Delta_{\cc^{n+2}}^\diamond$ acts on functions of only $n+1$ variables, and after f\/ixing a section on functions of $n$ variables. However, the initial operator is clearly $\Delta_{\cc^{n+2}}$. Finally, the diamond $\diamond$ is a symbol that we have already used in the context of line bundles.
\end{Remark}

Restricting (\ref{syme1a}) and (\ref{syme2a}) to $\Lambda^{\frac{2-n}{2}}(\cV_0)$ we obtain
\begin{alignat*}{3}
&B^{\diamond,\frac{-2-n}{2}}\Delta_{\cc^{n+2}}^\diamond =\Delta_{\cc^{n+2}}^\diamond B^{\diamond,\frac{2-n}{2}},\qquad&& B\in \so\big(\cc^{n+2}\big),&\\ 
&{\alpha}^{\diamond,\frac{-2-n}{2}}\Delta_{\cc^{n+2}}^\diamond =\Delta_{\cc^{n+2}}^\diamond {\alpha}^{\diamond,\frac{2-n}{2}},\qquad&& {\alpha}\in \mathrm{O}\big(\cc^{n+2}\big).&
\end{alignat*}

The following proposition is the consequence of the proof of Theorem~\ref{wer}.
\begin{Proposition}\quad
\begin{enumerate}\itemsep=0pt
\item[$(1)$] In the notation of Method~I of the proof of Theorem~{\rm \ref{wer}}, we have
\begin{gather}
\langle z|z\rangle_{\cc^n}\Delta_{\cc^{n+2}}^\diamond= \left(\frac{n-2}{2}\right)^2+\cC_{\cc^n}+N_{m+1}^2.\label{deq1a}\end{gather}
\item[$(2)$]
In the notation of Method II of the proof of Theorem~{\rm \ref{wer}}, we have
\begin{gather}
\langle z|z\rangle_{\cc^{n+1}}\Delta_{\cc^{n+2}}^\diamond= -\left(\frac{n-2}{2}\right)\frac{n}{2}+\cC_{\cc^{n+1}}.\label{deq3a}\end{gather}
\end{enumerate}
\end{Proposition}

\begin{proof}
(\ref{deq1a}) follows from (\ref{deq}). (\ref{deq3a}) follows from (\ref{deq3}).
\end{proof}

\subsection{Conformal invariance of Laplacian for a general section}

The operator $\Delta_{\cc^{n+2}}^\diamond$ is quite abstract. In this subsection we describe how to make it more explicit.

Consider an open set $\cY_0$ in $\cY$. Let $\cV_0:=\pi^{-1}(\cY_0)$ and $\eta\in\frac12\zz$. Choose a section $\gamma$ of the line bundle $\cV\to\cY$ based on $\cY_0$. As in (\ref{phi}) and (\ref{psi}), we can introduce $\psi^{\gamma,\eta }\colon \Lambda^\eta(\cV_0)\to\cA(\cY_0)$ and its left inverse $\phi^{\gamma ,\eta}$. We set
\begin{alignat*}{3}
&B^{\gamma ,\eta} :=\psi^{\gamma,\eta} B^{\diamond,\eta}\phi^{\gamma
 ,\eta} = \psi^{\gamma,\eta}B\phi^{\gamma ,\eta},\qquad&& B \in\so\big(\cc^{n+2}\big),&\\
&\alpha^{\gamma ,\eta}:=\psi^{\gamma,\eta} \alpha^{\diamond,\eta}\phi^{\gamma,\eta}=\psi^{\gamma,\eta} \alpha\phi^{\gamma
 ,\eta},\qquad&& \alpha \in\mathrm{O}\big(\cc^{n+2}\big).&
\end{alignat*}

As explained in Propositions~\ref{coc1}, \ref{coc2} and~\ref{coc7}, we obtain a representation and a local representation
\begin{gather}
\so\big(\cc^{n+2}\big)\ni B \mapsto B^{\gamma ,\eta}\in \cA\rtimes \hol(\cY_0),\label{pos1c}\\
\mathrm{O}\big(\cc^{n+2}\big) \ni {\alpha}\mapsto {\alpha}^{\gamma ,\eta} \underset{\loc}{\in} \cA^\times\rtimes\Hol(\cY_0).\label{pos2c}
\end{gather}
We also def\/ine
\begin{gather*}
\Delta_{\cc^{n+2}}^\gamma:=\psi^{\gamma,\frac{-2-n}{2}}\Delta_{\cc^{n+2}}^\diamond \phi^{\gamma,\frac{2-n}{2}} =\psi^{\gamma,\frac{-2-n}{2}}\Delta_{\cc^{n+2}}
\phi^{\gamma,\frac{2-n}{2}} .\end{gather*}
We have the identities{\samepage
\begin{alignat}{3}
&B^{\gamma ,\frac{-2-n}{2}}\Delta_{\cc^{n+2}}^\gamma =\Delta_{\cc^{n+2}}^\gamma B^{\gamma ,\frac{2-n}{2}},\qquad&& B \in \so\big(\cc^{n+2}\big),&\label{posd1}\\
&{\alpha}^{\gamma ,\frac{-2-n}{2}}\Delta_{\cc^{n+2}}^\gamma =\Delta_{\cc^{n+2}}^\gamma {\alpha}^{\gamma ,\frac{2-n}{2}},\qquad&& {\alpha} \in \mathrm{O}\big(\cc^{n+2}\big).&\label{posd2}
\end{alignat}}

Thus we have shown that (inf\/initesimal) conformal transformations of the $n$-dimensional manifold $\cY_0$ lead to generalized (inf\/initesimal) symmetries of $\Delta_{\cc^{n+2}}^\gamma$. Even if (in a somewhat dif\/ferent form) this is a known fact, it seems that our derivation is new and of interest. In particular, it shows that a large class of second order $n$-dimensional operators together with their generalized symmetries directly come from the $(n+2)$-dimensional Laplacian with its true symmetries.

\subsection[Conformal invariance of Laplacian on $\cc^n$]{Conformal invariance of Laplacian on $\boldsymbol{\cc^n}$}\label{subsec-conf}

Let us describe more closely the above construction in the case of the section (\ref{secto}). In this case, instead of $\gamma$ we will write ``$\fl$'', for {\em flat}. We identify of $\cY_e$ with $\cc^n$. We can restrict~(\ref{so0}) to an action of $\so(\cc^{n+2})$ on $\cY_e$, and (\ref{SO0}) to a local action of $\mathrm{O}(\cc^{n+2})$ on $\cY_{e}$. Using~(\ref{nota1}) and~(\ref{nota2}), we obtain
\begin{gather*}
\so\big(\cc^{n+2}\big) \ni B \mapsto B^{\fl }\in\hol\big(\cc^n\big),\qquad 
\mathrm{O}\big(\cc^{n+2}\big)\ni {\alpha}\mapsto {\alpha}^{\fl }\underset{\loc}{\in}\Hol\big(\cc^n\big).
\end{gather*}

We introduce $\psi^{\fl,\eta }\colon \Lambda^\eta(\cV_e)\to\cA(\cc^n)$ and its left inverse $\phi^{\fl ,\eta}$. (\ref{pos1c}) and (\ref{pos2c}) can be rewritten as
\begin{gather*}
\so\big(\cc^{n+2}\big) \ni B \mapsto B^{\fl,\eta}\in \cA\rtimes \hol\big(\cc^n\big),\qquad 
\mathrm{O}\big(\cc^{n+2}\big) \ni {\alpha} \mapsto {\alpha}^{\fl,\eta} \underset{\loc}{\in}\cA^\times\rtimes\Hol\big(\cc^n\big).
\end{gather*}

The $(n+2)$-dimensional Laplacian reduced to the f\/lat section is just the usual $n$-dimensional Laplacian:
\begin{gather*}\Delta_{\cc^{n+2}}^\fl=\Delta_{\cc^n}.\end{gather*}
The symmetries (\ref{posd1}) and (\ref{posd2}) become the generalized symmetries of the usual Laplacian:
\begin{alignat*}{3}
&B^{\fl ,\frac{-2-n}{2}}\Delta_{\cc^n} =\Delta_{\cc^n} B^{\fl ,\frac{2-n}{2}},\qquad&& B \in \so\big(\cc^{n+2}\big),&\\
&{\alpha}^{\fl ,\frac{-2-n}{2}}\Delta_{\cc^n}=\Delta_{\cc^n} {\alpha}^{\fl ,\frac{2-n}{2}},\qquad&&{\alpha} \in \mathrm{O}\big(\cc^{n+2}\big).&
\end{alignat*}
Thus $\cc^{n+2}$ serves to describe in a simple way conformal symmetries of $\cc^n$. When used in this fashion, the space $\cc^{n+2}$ will be sometimes called
the {\em extended space}.

Below we sum up information about conformal symmetries on the level of the extended space~$\cc^{n+2}$ and the space~$\cc^n$. We will use the split coordinates, that is, $ z\in\cc^{n+2}$ and $y\in \cc^n$ have the square
\begin{gather*}
 \langle z|z\rangle =\sum_{j\in I_{n+2}}z_{-j}z_j,\qquad \langle y|y\rangle =\sum_{j\in I_n}y_{-j}y_j.\end{gather*}
As a rule, if a given operator does not depend on $\eta$, we will omit $\eta$.

 {\bf Cartan algebra of $\so(\cc^{n+2})$.}
Cartan operators of $\so(\cc^n)$, $i=1,\dots,m$:
\begin{gather*}
N_i =z_{-i}\p_{z_{-i}}-z_{i}\p_{z_i},\qquad N_i^{\fl}=y_{-i}\p_{y_{-i}}-y_{i}\p_{y_i}.
\end{gather*}
Generator of dilations:
\begin{gather*}
N_{m+1}=z_{-m-1}\p_{z_{-m-1}}-z_{m+1}\p_{z_{m+1}},\qquad
N_{m+1}^{\fl ,\eta}=\sum\limits_{i\in I_n}y_i\p_{y_i}-\eta = D_{\cc^n}-\eta.
\end{gather*}

{\bf Root operators.} Roots of $\so(\cc^n)$, $|i|<|j|$, $i,j\in I_n$:
\begin{gather*}
B_{i,j}=z_{-i}\p_{z_j}-z_{-j}\p_{z_i},\qquad B_{i,j}^{\fl }=y_{-i}\p_{y_j}-y_{-j}\p_{y_i}.
\end{gather*}
Generators of translations, $j\in I_n$:
\begin{gather*}
B_{-m-1,j}=z_{m+1}\p_{z_j}-z_{-j}\p_{z_{-m-1}},\qquad B_{-m-1,j}^{\fl } =\p_{y_j}.
\end{gather*}
Generators of special conformal transformations, $j\in I_n$:
\begin{gather*}
B_{m+1,j} =z_{-m-1}\p_{z_j}-z_{-j}\p_{z_{m+1}},\qquad B_{m+1,j}^{\fl ,\eta} =-\frac{1}{2} \langle y|y\rangle \p_{y_j}
+y_{-j}\sum\limits_{i\in I_n}y_i\p_{y_i}-\eta y_{-j}.
\end{gather*}

{\bf Weyl symmetries.} We will write $K$ for a function on $\cc^{n+2}$ and $f$ for a function on $\cc^n$.
Ref\/lection:
\begin{gather*}
\tau_0 K(z_0,\dots) =K(-z_0,\dots),\\
\tau_0^{\fl }f(y_0,\dots)=f(-y_0,\dots).
\end{gather*}
Flips, $j=1,\dots,m$:
\begin{gather*}
\tau_j K(\dots,z_{-j},z_j,\dots,z_{-m-1},z_{m+1})=K(\dots,z_{j},z_{-j},\dots,z_{-m-1},z_{m+1}),\\
\tau_j^{\fl}f(\dots,y_{-j},y_j,\dots)=f(\dots, y_{j},y_{-j},\dots).
\end{gather*}
Inversion:
\begin{gather*}
\tau_{m+1} K(\dots,z_{-m-1},z_{m+1})=K(\dots,z_{m+1},z_{-m-1}),\\
\tau_{m+1}^{\fl,\eta} f(y)=\left(-\frac{\langle y|y\rangle }{2}\right)^\eta f\left(-\frac{2y}{\langle y|y\rangle }\right).
\end{gather*}
Permutations, $\sigma\in S_m$:
\begin{gather*}
\sigma K(\dots,z_{-j},z_j,\dots,z_{-m-1},z_{m+1})=K(\dots,z_{-\sigma_j},z_{\sigma_j},\dots,z_{-m-1},z_{m+1}),\\
\sigma^{\fl }f(\dots,y_{-j},y_j,\dots)=f(\dots, y_{-\sigma_j},y_{\sigma_j},\dots).
\end{gather*}
Special conformal transformations, $j=1,\dots,m$:
\begin{gather*}
\sigma_{(j,m+1)} K(z_{-1},z_1,\dots,z_{-j},z_j,\dots,z_{-m-1},z_{m+1})\\
\qquad{} =K(z_{-1},z_1,\dots,z_{-m-1},z_{m+1},\dots,z_{-j},z_j),\\[1mm]
\sigma_{(j,m+1)}^{\fl ,\eta} f(y_{-1},y_{1},\dots,y_{-j},y_j,\dots)=
 y_j^\eta f\left(\frac{y_{-1}}{y_j},\frac{y_{1}}{y_j},\dots,-\frac{\langle y|y\rangle}{2y_j},\frac{1}{y_j},\dots\right).
\end{gather*}

{\bf Laplacian:}
\begin{gather*}
\Delta_{\cc^{n{+}2}}=\sum\limits_{i\in I_{n+2}}\p_{z_i}\p_{z_{-i}},\qquad
\Delta_{\cc^{n+2}}^\fl =\sum\limits_{i\in I_n}\p_{y_i}\p_{y_{-i}} = \Delta_{\cc^n}.
\end{gather*}

{\bf Computations.} Let us describe how to derive these formulas in an easy way. Consider $\cc^{n+1}\times\cc^\times$ (def\/ined by $z_{m+1}\neq0$), which is an open dense subset of $\cc^{n+2}$. Clearly, $\cV_e$ is contained in $\cc^{n+1}\times\cc^\times$.

We will write $\Lambda^\eta(\cc^{n+1}\times\cc^\times)$ for the space of functions homogeneous of degree $\eta$ on $\cc^{n+1}\times\cc^\times$.

Instead of using the maps $\phi^{\fl ,\eta}$ and $\psi^{\fl ,\eta}$, as in (\ref{phi}) and (\ref{psi}), we will prefer $\Phi^{\fl ,\eta}\colon \cA(\cc^n)\to\Lambda^\eta(\cc^{n+1}\times\cc^\times)$ and $\Psi^{\fl ,\eta}\colon \Lambda^\eta(\cc^{n+1}\times\cc^\times)\to \cA(\cc^n)$ def\/ined below.

For $ K\in \Lambda^\eta\bigl(\cc^{n+1}\times\cc^\times\bigr)$, we def\/ine $\Psi^{\fl,\eta} K\in\cA(\cc^n)$ by
\begin{gather*}\big(\Psi^{\fl,\eta} K\big)(y)= K\left(y,-\frac{\langle y|y\rangle }{2},1\right), \qquad y\in \cc^n.\end{gather*}

Let $f\in\cA(\cc^n)$. Then there exists a unique function in $\Lambda^\eta\bigl(\cc^{n+1}\times\cc^\times\bigr)$ that extends $f$ and does not depend on $z_{-m-1}$. It is given by
\begin{gather*}\big(\Phi^{\fl ,\eta} f\big) (\dots,z_m,z_{-m-1},z_{m+1}) :=z_{m+1}^\eta f\left(\dots,\frac{z_m}{z_{m+1}}\right).\end{gather*}

$\Psi^{\fl,\eta}$ is a left inverse of $\Phi^{\fl ,\eta}$:
\begin{gather*}
\Psi^{\fl,\eta}\Phi^{\fl ,\eta}=\id.
\end{gather*}
Clearly,
\begin{gather*}
\Phi^{\fl ,\eta}f\big|_{\cV_e} =\phi^{\fl ,\eta}f,\qquad \Psi^{\fl,\eta}K =\psi^{\fl,\eta}\big(K\big|_{\cV_e}\big).
\end{gather*}
Moreover, functions in $\Lambda^\eta(\cc^{n+1}\times\cc^\times)$ restricted to $\cV_e$ are in $\Lambda^\eta(\cV_e)$. Therefore,
\begin{alignat*}{3}
& B^{\fl ,\eta} =\Psi^{\fl,\eta} B\Phi^{\fl ,\eta},\qquad && B \in \so\big(\cc^{n+2}\big),&\\
& {\alpha}^{\fl ,\eta} =\Psi^{\fl,\eta} {\alpha}\Phi^{\fl ,\eta},\qquad && {\alpha} \in \mathrm{O}\big(\cc^{n+2}\big).&
\end{alignat*}
Note also that
\begin{gather*}\Delta_{\cc^{n+2}}^\fl=\Psi^{\fl,\eta} \Delta_{\cc^{n+2}}\Phi^{\fl ,\eta}=\Delta_{\cc^n}.\end{gather*}

In practice, the above idea can be implemented by the following change of coordinates on~$\cc^{n+2}$:
\begin{gather*}
y_i:= \frac{z_i}{z_{m+1}},\quad i\in I_{n},\qquad R:= \sum\limits_{i\in I_{n+2}}z_{i}z_{-i},\qquad p:= z_{m+1}.
\end{gather*}
The inverse transformation is
\begin{gather*}
z_i= py_i,\quad i\in I_{n},\qquad z_{-m-1}= \frac{1}{2}\left(\frac{R}{p}-p\sum\limits_{i\in I_{n}}y_iy_{-i}\right),\qquad z_{m+1}=p.
\end{gather*}
The derivatives are equal to
\begin{gather*}
\p_{z_i}= z_{m+1}^{-1}\p_{y_i}+2z_{-i}\p_R,\quad i\in I_{n},\qquad \p_{z_{-m-1}}= 2z_{m+1}\p_R,\\
\p_{z_{m+1}}= \p_p-z_{m+1}^{-2}\sum\limits_{i\in I_{n}}z_i\p_{y_i}+2 z_{-m-1}\p_R.
\end{gather*}

Note that these coordinates are def\/ined on $\cc^{n+1}\times\cc^\times$. $\cV_e$ is given by the condition $R=0$. The section (\ref{secto}) (see also~(\ref{secto2})) is given by $p=1$.

For a function $y\mapsto f(y)$ we have
\begin{gather*}\big(\Phi^{\fl ,\eta} f\big)(y,R,p)=p^\eta f(y).\end{gather*}
For a function $(y,R,p)\mapsto K(y,R,p)$ we have
\begin{gather*}\big(\Psi^{\fl,\eta} K\big)(y)= K(y,0,1).\end{gather*}
Note also that on $\Lambda^\eta(\cc^{n+1}\times\cc^\times)$ we have
\begin{gather*} p\p_p=\eta. \end{gather*}

\subsection[Dimension $n=1$]{Dimension $\boldsymbol{n=1}$}

Let us illustrate the constructions of this section by describing the projective quadric in the lowest dimensions, where everything is very explicit. We start with dimension $n=1$.

The 1-dimensional projective quadric is isomorphic to the Riemann sphere or, what is the same, the 1-dimensional projective complex space:
\begin{gather*}\cY^1\simeq\cc\cup\{\infty\}=P^1\cc.\end{gather*}

Indeed, consider $\cc^3$ with the scalar product
\begin{gather*}\langle z|z\rangle=z_0^2+2z_{-1}z_{+1}.\end{gather*}
We can cover $\cY^1$ with two maps:
\begin{gather*}
\cc\ni s\mapsto \phi_+(s) =\left(s,1,-\frac12s^2\right)\cc^\times\in\cY^1,\qquad \cc\ni s\mapsto \phi_-(s) =\left(s,-\frac12s^2,1\right)\cc^\times\in\cY^1.
\end{gather*}
The transition map is
\begin{gather*}\phi_+^{-1}\phi_-(s)=-\frac{2}{s}.\end{gather*}
The Lie algebra $\so(\cc^3)$ is spanned by
\begin{gather*}B_{0,1},\quad B_{0,-1},\quad N_1,\end{gather*}
with the commutation relations
\begin{gather*}
[B_{0,1},B_{0,-1}] =N_1,\qquad [N_1,B_{0,1}] =B_{0,1},\qquad [N_1,B_{0,-1}] =-B_{0,-1}.
\end{gather*}
The Casimir operator is
\begin{gather}\label{casimir1}
\cC=2B_{0,1}B_{0,-1}-N_1^2-N_1 =2B_{0,-1}B_{0,1}-N_1^2+N_1.
\end{gather}

\subsection[Dimension $n=2$]{Dimension $\boldsymbol{n=2}$}

The 2-dimensional projective quadric is isomorphic to the product of two Riemann spheres:
\begin{gather*}\cY^2\simeq P^1\cc\times P^1\cc .\end{gather*}
Indeed, consider $\cc^4$ with the scalar product
\begin{gather*}\langle z|z\rangle=2z_{-1}z_{+1}+2z_{-2}z_{+2}.\end{gather*}
We can cover $\cY^2$ with four maps:
\begin{gather*}
\cc\times\cc\ni (t,s)\mapsto \phi_{+1}(t,s)=(-ts,1,t,s)\cc^\times\in\cY^2,\\
\cc\times\cc\ni (t,s)\mapsto \phi_{-1}(t,s)=(1,-ts,s,t)\cc^\times\in\cY^2,\\
\cc\times\cc\ni (t,s)\mapsto \phi_{+2}(t,s)=(-s,-t,1,ts)\cc^\times\in\cY^2,\\
\cc\times\cc\ni (t,s)\mapsto \phi_{-2}(t,s)=(-t,-s,-ts,1)\cc^\times\in\cY^2.
\end{gather*}
Here are the transition maps:
\begin{gather*}
\begin{split}
& \phi_{-1}^{-1}\phi_{+1}(t,s)=\phi_{-2}^{-1}\phi_{+2}(t,s) =\big({-}t^{-1},-s^{-1}\big),\\
& \phi_{-2}^{-1}\phi_{-1}(t,s)=\phi_{+2}^{-1}\phi_{+1}(t,s) =\big({-}t^{-1},s\big),\\
& \phi_{-2}^{-1}\phi_{+1}(t,s)=\phi_{+2}^{-1}\phi_{-1}(t,s) =\big(t,-s^{-1}\big).
\end{split}
\end{gather*}
The Lie algebra $\so(\cc^4)$ is spanned by
\begin{gather*}N_1,\quad N_2,\quad B_{1,2},\quad B_{1,-2},\quad B_{-1,2},\quad B_{-1,-2}.\end{gather*}
Its Casimir operator is
\begin{gather*}\cC=2B_{1,2}B_{-1,-2}+2B_{1,-2}B_{-1,2}-N_1^2-N_2^2+2N_1.\end{gather*}

As is well known, $\so(\cc^4)$ decomposes into the direct sum $\so_+(\cc^3)\oplus\so_-(\cc^3)$ of two commuting Lie algebras isomorphic to $\so(\cc^3)$ spanned by
\begin{gather*}
B_{1,2},\quad B_{-1,-2},\quad N_1+N_2\qquad\text{and}\qquad B_{1,-2},\quad B_{-1,2},\quad N_1-N_2
\end{gather*}
with the commutation relations
\begin{alignat*}{3}
& [B_{1,2},B_{-1,-2}]=N_1+N_2,\qquad && [B_{1,-2},B_{-1,2}]=N_1-N_2,&\\
& [N_1+N_2,B_{1,2}]=2B_{1,2}, \qquad && [N_1-N_2,B_{1,-2}]=2B_{1,-2},&\\
& [N_1+N_2,B_{-1,-2}]=-2B_{-1,-2}, \qquad && [N_1-N_2,B_{-1,2}]=-2B_{-1,2}.&
\end{alignat*}
The corresponding Casimir operators are
\begin{gather*}
\cC_{+}=2B_{1,2}B_{-1,-2}-\frac12(N_1+N_2)^2-N_1-N_2 =2B_{-1,-2}B_{1,2}-\frac12(N_1+N_2)^2+N_1+N_2,\\
\cC_{-}=2B_{1,-2}B_{-1,2}-\frac12(N_1-N_2)^2-N_1+N_2=2B_{-1,2}B_{1,-2}-\frac12(N_1-N_2)^2+N_1-N_2.
\end{gather*}
Thus
\begin{gather*}\cC=\cC_{+}+\cC_{-}.\end{gather*}

In the enveloping algebra of $\so(\cc^4)$ the operators $\cC_{+}$ and $\cC_{-}$ are distinct. They satisfy $\alpha(\cC_{-})=\cC_{+}$ for $\alpha\in\mathrm{O}(\cc^4)\backslash\SO(\cc^4)$, for instance for $\tau_i$, $i=1,2$.

However, inside the associative algebra of dif\/ferential operators on $\cc^4$ we have the identity
\begin{gather*}B_{1,2}B_{-1,-2}-B_{1,-2}B_{-1,2}=N_1N_2+N_2,\end{gather*}
which implies
\begin{gather*}\cC_{+}=\cC_{-}\end{gather*}
inside this algebra. Therefore, represented in the algebra of dif\/ferential operators we have
\begin{gather}
\cC =4B_{1,2}B_{-1,-2}-(N_1+N_2)^2-2N_1-2N_2 \nonumber\\
\hphantom{\cC}{} =4B_{-1,-2}B_{1,2}-(N_1+N_2)^2+2N_1+2N_2\nonumber\\
\hphantom{\cC}{} =4B_{1,-2}B_{-1,2}-(N_1-N_2)^2-2N_1+2N_2 \nonumber\\
\hphantom{\cC}{}=4B_{-1,2}B_{1,-2}-(N_1-N_2)^2+2N_1-2N_2.\label{casimir2}
\end{gather}

\section[$\so(\cc^6)$ and the hypergeometric equation]{$\boldsymbol{\so(\cc^6)}$ and the hypergeometric equation}\label{s6}

In this section we derive the hypergeometric operator and its $\so(\cc^6)$ symmetries. We will consider the following levels:
\begin{enumerate}\itemsep=0pt
\item[(1)] extended space $\cc^6$ and the Laplacian $\Delta_{\cc^6}$,
\item[(2)] reduction to the so-called spherical section and the corresponding Laplace--Beltrami operator,
\item[(3)] depending on the choice of coordinates, separation of variables leads to the balanced or standard hypergeometric operator.
\end{enumerate}

Alternatively, one can use a dif\/ferent derivation:
\begin{enumerate}\itemsep=0pt
\item[$(2)'$] reduction to $\cc^4$ and $\Delta_{\cc^4}$ with help of the f\/lat section,
\item[$(3)'$] with appropriate coordinates, separation of variables leads to the balanced or standard hypergeometric operator.
\end{enumerate}

A separate subsection is devoted to factorizations of the hypergeometric operator. We will see that they are closely related to $\so(\cc^4)$ subalgebras of $\so(\cc^6)$ and their Casimir operators.

\subsection[Extended space $\cc^6$]{Extended space $\boldsymbol{\cc^6}$}

We consider $\cc^6$ with the coordinates
\begin{gather}
z_{-1},\quad z_1,\quad z_{-2},\quad z_2,\quad z_{-3}, \quad z_3\label{sq0}\end{gather}
and the scalar product given by
\begin{gather}
\langle z|z\rangle=2z_{-1}z_1+2z_{-2}z_2+2z_{-3}z_3.\label{sq1}\end{gather}

 {\bf Lie algebra $\so(\cc^6)$.} Cartan algebra:
\begin{gather*}N_i=z_{-i}\partial_{z_{-i}}-z_{i}\partial_{z_{i}},\qquad i=1,2,3.\end{gather*}
Root operators:
\begin{gather*}B_{i,j} =z_{-i}\p_{{j}}-z_{-j}\p_{{i}},\qquad 1\leq|i|<|j|\leq 3.\end{gather*}

{\bf Weyl symmetries.} Transpositions:
\begin{gather*}
\sigma_{(12)}K(z_{-1},z_1,z_{-2},z_2,z_{-3},z_3)=K(z_{-2},z_2,z_{-1},z_1,z_{-3},z_3),\\
\sigma_{(13)}K(z_{-1},z_1,z_{-2},z_2,z_{-3},z_3)=K(z_{-3},z_3,z_{-2},z_2,z_{-1},z_1),\\
\sigma_{(23)}K(z_{-1},z_1,z_{-2},z_2,z_{-3},z_3)=K(z_{-1},z_1,z_{-3},z_3,z_{-2},z_2).
\end{gather*}
Cycles:
\begin{gather*}
\sigma_{(123)}K(z_{-1},z_1,z_{-2},z_2,z_{-3},z_3)=K(z_{-3},z_3,z_{-1},z_1,z_{-2},z_2),\\
\sigma_{(132)}K(z_{-1},z_1,z_{-2},z_2,z_{-3},z_3)=K(z_{-2},z_2,z_{-3},z_3,z_{-1},z_1).
\end{gather*}
Flips:
\begin{gather*}
\tau_1 K(z_{-1},z_1,z_{-2},z_2,z_{-3},z_3)=K(z_{1},z_{-1},z_{-2},z_2,z_{-3},z_3),\\
\tau_2 K(z_{-1},z_1,z_{-2},z_2,z_{-3},z_3)=K(z_{-1},z_{1},z_{2},z_{-2},z_{-3},z_3),\\
\tau_3 K(z_{-1},z_1,z_{-2},z_2,z_{-3},z_3)=K(z_{-1},z_{1},z_{-2},z_2,z_{3},z_{-3}).
\end{gather*}

{\bf Laplacian:}
\begin{gather}
\Delta_{\cc^6}=2\partial_{z_{-1}}\partial_{z_{1}}+2\partial_{z_{-2}}\partial_{z_{2}}+2\partial_{z_{-3}}\partial_{z_{3}}.\label{sq2}\end{gather}

{\bf Symmetries:}
\begin{subequations}
\begin{alignat}{3}
& N_i\Delta_{\cc^6}=\Delta_{\cc^6} N_i,\qquad && 1\leq i\leq 3,&\label{sq31}\\
& B_{i,j}\Delta_{\cc^6}=\Delta_{\cc^6} B_{i,j},\qquad && 1\leq|i|<|j|\leq3, & \\
& \sigma\Delta_{\cc^6} =\Delta_{\cc^6}\sigma,\qquad && \sigma\in S_3, & \\
&\tau_j\Delta_{\cc^6} =\Delta_{\cc^6}\tau_j,\qquad && 1\leq j\leq3.& \label{sq34}
\end{alignat}
\end{subequations}

\subsection{Spherical section}

In this subsection we consider the section of the quadric
\begin{gather*}\cV^5:=\big\{z\in\cc^6\colon 2z_{-1}z_1+2z_{-2}z_2+2z_{-3}z_3=0\big\}\end{gather*}
 given by equations
\begin{gather*}4 =2\left(z_{-1}z_1+z_{-2}z_2\right)=-2z_3z_{-3}.\end{gather*}
We will call it the {\em spherical section}, because it coincides with $\cS^3(4)\times\cS^1(-4)$. The superscript used for this section will be ``$\sph$'' for spherical.

We will see that this section is well suited to obtain the hypergeometric equation, both in the balanced and standard form, because its conformal factor is trivially equal to~$1$.

As a preparation for a discussion of this section, let us choose the coordinates
\begin{alignat*}{5}
& r =\sqrt{2\left(z_{-1}z_1+z_{-2}z_2\right)}, \qquad && p =\sqrt{2z_3z_{-3}}, \qquad && w =\frac{z_{-1}z_1}{z_{-1}z_1+z_{-2}z_{2}},& \\
& u_1 =\sqrt\frac{z_{-1}}{z_{1}},\qquad && u_2 =\sqrt\frac{z_{-2}}{z_{2}},\qquad && u_3 =\sqrt\frac{z_{-3}}{z_{3}}.&
\end{alignat*}
The null quadric in these coordinates is given by $r^2+p^2=0$. The generator of dilations is
\begin{gather*}D_{\cc^6}=r\ddr+p\ddp.\end{gather*}
The spherical section is given by the condition $r^2=4$.

Let us now describe in detail various objects in the spherical section.

{\bf Lie algebra $\so(\cc^6)$.} Cartan operators:
\begin{gather*}
N_{1}^\sph =u_1\ddu{1},\qquad N_{2}^\sph =u_2\ddu{2},\qquad N_{3}^\sph =u_3\ddu{3}.
\end{gather*}
Roots:
\begin{subequations}\label{coor}
\begin{gather}
B_{2,1}^\sph=u_1u_2\sqrt{w(1-w)} \left(\ddw-\frac{N^\sph_1}{2w}+\frac{N^\sph_2}{2(1-w)}\right) ,\\[1mm]
B_{-2,-1}^\sph=\frac{1}{u_1 u_2}\sqrt{w(1-w)} \left(\ddw+\frac{N^\sph_1}{2w}-\frac{N^\sph_2}{2(1-w)}\right) ,\\[1mm]
B_{2,-1}^\sph=\frac{u_2}{u_1}\sqrt{w(1-w)} \left(\ddw+\frac{N^\sph_1}{2w}+\frac{N^\sph_2}{2(1-w)}\right) ,\\[1mm]
B_{-2,1}^\sph=\frac{u_1}{u_2}\sqrt{w(1-w)} \left(\ddw-\frac{N^\sph_1}{2w}-\frac{N^\sph_2}{2(1-w)}\right) ,\\[1mm]
B_{3,1}^{\sph,{\eta}}=\frac{\ii}{2}u_1u_3\sqrt{w} \left(\eta+2(1-w)\ddw-\frac{N^\sph_1}{w}- N^\sph_3\right) ,\\[1mm]
B_{-3,-1}^{\sph,{\eta}}=\frac{\ii}{2}\frac{1}{u_1u_3}\sqrt{w} \left(\eta+2(1-w)\ddw+\frac{N^\sph_1}{w}+ N^\sph_3\right) ,\\[1mm]
B_{3,-1}^{\sph,{\eta}}=\frac{\ii}{2}\frac{u_3}{u_1}\sqrt{w} \left(\eta+2(1-w)\ddw+\frac{N^\sph_1}{w}- N^\sph_3\right) ,\\[1mm]
B_{-3,1}^{\sph,{\eta}}=\frac{\ii}{2}\frac{u_1}{u_3}\sqrt{w} \left(\eta+2(1-w)\ddw-\frac{N^\sph_1}{w}+ N^\sph_3\right) ,\\[1mm]
B_{3,2}^{\sph,{\eta}}=\frac{\ii}{2}u_2u_3\sqrt{1-w} \left(\eta-2w\ddw-\frac{N^\sph_2}{1-w}- N^\sph_3\right) ,\\[1mm]
B_{-3,-2}^{\sph,{\eta}}=\frac{\ii}{2}\frac{1}{u_2u_3}\sqrt{1-w} \left(\eta-2w\ddw+\frac{N^\sph_2}{1-w}+ N^\sph_3\right) ,\\[1mm]
B_{3,-2}^{\sph,{\eta}}=\frac{\ii}{2}\frac{u_3}{u_2}\sqrt{1-w} \left(\eta-2w\ddw+\frac{N^\sph_2}{1-w}- N^\sph_3\right) ,\\[1mm]
B_{-3,2}^{\sph,{\eta}}=\frac{\ii}{2}\frac{u_2}{u_3}\sqrt{1-w} \left(\eta-2w\ddw-\frac{N^\sph_2}{1-w}+ N^\sph_3\right) .
\end{gather}
\end{subequations}

{\bf Weyl symmetries.} Transpositions:
\begin{gather*}
\sigma_{12}^{\sph,{\eta}} f(w,u_1,u_2,u_3) = f (1-w,u_2,u_1,u_3 ),\\
\sigma_{13}^{\sph,{\eta}} f(w,u_1,u_2,u_3) = \left(\ii\sqrt{w}\right)^\eta f\left(\frac{1}{w},u_3,u_2,u_1\right),\\
\sigma_{23}^{\sph,{\eta}} f(w,u_1,u_2,u_3) = \left(\ii\sqrt{1-w}\right)^\eta f\left(\frac{w}{w-1},u_1,u_3,u_2\right).
\end{gather*}
Cycles:
\begin{gather*}
\sigma_{312}^{\sph,{\eta}} f(w,u_1,u_2,u_3) = \left(\ii\sqrt{1-w}\right)^\eta f\left(\frac{1}{1-w},u_3,u_1,u_2\right),\\
\sigma_{231}^{\sph,{\eta}} f(w,u_1,u_2,u_3)= \left(\ii\sqrt{w}\right)^\eta f\left(1-\frac{1}{w},u_2,u_3,u_1\right).
\end{gather*}
Flips:
\begin{gather*}
\tau_{1}^{\sph,{\eta}} f(w,u_1,u_2,u_3)= f\left(w,\frac{1}{u_1},u_2,u_3\right),\\
\tau_{2}^{\sph,{\eta}} f(w,u_1,u_2,u_3)= f\left(w,u_1,\frac{1}{u_2},u_3\right),\\
\tau_{3}^{\sph,{\eta}} f(w,u_1,u_2,u_3) = f\left(w,u_1,u_2,\frac{1}{u_3}\right).
\end{gather*}

 The Laplacian in coordinates (\ref{coor}) is
\begin{gather*}
\Delta_{\cc^6} =\frac{4}{r^2} \biggl(\frac{1}{4}\left((r\ddr)^2+2(r\ddr)+\frac{r^2}{p^2}(p\ddp)^2\right)\\
\hphantom{\Delta_{\cc^6} =\frac{4}{r^2} \biggl(}{} +\ddw w(1-w)\ddw-\frac{(u_1\ddu{1})^2}{4w}-\frac{(u_1\ddu{2})^2}{4(1-w)}-\frac{r^2}{p^2}\frac{(u_3\ddu{3})^2}{4}\biggr).
\end{gather*}
Using
\begin{gather*}
(r\ddr)^2+2r\ddr+\tfrac{r^2}{p^2}(p\ddp)^2 =\left(r\ddr+\frac{r^2}{p^2}(p\ddp)+1\right) (r\ddr+p\ddp+1 )-1\\
\hphantom{(r\ddr)^2+2r\ddr+\tfrac{r^2}{p^2}(p\ddp)^2 =}{} -\left(\frac{r^2}{p^2}+1\right) (p\ddp ) (r\ddr+1 ),
\end{gather*}
 $r^2+p^2=0$ and $r\ddr+p\ddp=-1$, we obtain
\begin{gather*}
\Delta_{\cc^6}^\diamond = \frac{4}{r^2}\left(\ddw w(1-w)\ddw-\frac{(u_1\ddu{1})^2}{4w}-\frac{(u_2\ddu{2})^2}{4(1-w)}+\frac{(u_3\ddu{3})^2}{4}-\frac14\right).
\end{gather*}

To convert $\Delta_{\cc^6}^\diamond$ into the {\it reduced Laplacian} $\Delta_{\cc^6}^\sph$, we simply remove the prefactor $\frac{4}{r^2}$, obtaining the Laplace--Beltrami operator on $\cS^3(4)\times\cS^1(-4)$:
\begin{gather}
\Delta_{\cc^6}^\sph =\ddw w(1-w)\ddw- \frac{\big(N_1^\sph\big)^2}{4w}-\frac{\big(N_2^\sph\big)^2}{4(1-w)}+\frac{\big(N_3^\sph\big)^2}{4}-\frac{1}{4}.\label{sq9}
\end{gather}

{\bf Generalized symmetries:}
\begin{alignat}{3}
& N_i^{\sph}\Delta_{\cc^6}^\sph =\Delta_{\cc^6}^\sph N_i^{\sph},\qquad &&1 \leq i\leq 3, & \label{yq1}\\
& B_{i,j}^{\sph,-3}\Delta_{\cc^6}^\sph =\Delta_{\cc^6}^\sph B_{i,j}^{\sph,-1},\qquad&& 1 \leq|i|<|j|\leq3, &\label{yq2}\\
& \sigma^{\sph,-3}\Delta_{\cc^6}^\sph =\Delta_{\cc^6}^\sph \sigma^{\sph,-1},\qquad && \sigma\in S_3,&\label{yq3}\\
& \tau_j^{\sph}\Delta_{\cc^6}^\sph =\Delta_{\cc^6}^\sph \tau_j^{\sph},\qquad&& 1\leq j\leq3.&\label{yq4}
\end{alignat}

\subsection{Balanced hypergeometric operator}

Using the spherical section we make an ansatz
\begin{gather}
f(w,u_1,u_2,u_3)=u_1^\alpha u_2^\beta u_3^\mu F(w).\label{form0}\end{gather}
Clearly,
\begin{gather*}
N_1^\sph f=\alpha f,\qquad N_2^\sph f=\beta f,\qquad N_3^\sph f=\mu f.
\end{gather*}
Therefore, on functions of the form (\ref{form0}), $\Delta_{\cc^6}^\sph$, that is~(\ref{sq9}), coincides with the balanced hypergeometric operator~(\ref{hyp1c}). The generalized symmetries for the roots (\ref{yq2}), for the permutations~(\ref{yq3}) and for the f\/lips~(\ref{yq4}) coincide with the transmutation relations, with the discrete symmetries, and with the sign changes of $\alpha$, $\beta$, $\mu$ of the balanced hypergeometric equation, respectively; see Section~\ref{subs-1}.

\subsection{Standard hypergeometric operator}

Alternatively, we can slightly change the coordinates (\ref{coor}), replacing $u_1$, $u_2$ with
\begin{gather*}
\tilde{u}_1:= \frac{z_{-1}}{\sqrt{z_{-1}z_1+z_{-2}z_2}}=u_1\sqrt{w},\qquad
\tilde{u}_2:=\frac{z_{-2}}{\sqrt{z_{-1}z_1+z_{-2}z_2}}=u_2\sqrt{1-w}. 
\end{gather*}
As compared with the previous coordinates, we need to replace $\partial_w$ with
\begin{gather}
\partial_w+N_1^\sph\frac{1}{2w}+N_2^\sph\frac{1}{2(w-1)}.\label{change1}\end{gather}
Let us only quote the results for the Cartan operators
\begin{gather*}
N_{1}^\sph =\tilde u_1\p_{\tilde u_1},\qquad N_{2}^\sph =\tilde u_2\p_{\tilde u_2},\qquad N_{3}^\sph =u_3\ddu{3},
\end{gather*}
and the reduced Laplacian:
\begin{gather*}
\begin{split}
& \Delta_{\cc^6}^\sph =w(1-w)\p_w^2+\big(\big(1+N_1^\sph\big)(1-w)-\big(1+N_2^\sph \big)w\big)\p_w\\
& \hphantom{\Delta_{\cc^6}^\sph =}{} +\frac14\big(N_3^\sph\big)^2-\frac14\big(N_1^\sph+N_2^\sph +1\big)^2.
\end{split}
\end{gather*}

If we now make the ansatz
 \begin{gather}
f(w,\tilde u_1,\tilde u_2,u_3)=\tilde u_1^\alpha \tilde u_2^\beta u_3^\mu \tilde F(w),\label{las}\end{gather}
then clearly,
\begin{gather*}
N_1^\sph f=\alpha f,\qquad N_2^\sph f=\beta f,\qquad N_3^\sph f=\mu f.
\end{gather*}
It is easy to see that on functions of the form (\ref{las}), $\Delta_{\cc^6}^\sph$ coincides with the standard hypergeometric operator~(\ref{hy1-tra}). When~(\ref{change1}) is applied to root operators and Weyl symmetries, we also obtain the symmetries of the standard hypergeometric operator described in~\cite{De}.

\subsection{Factorizations}

In the Lie algebra $\so(\cc^6)$ represented in (\ref{sq0}) we have~3 distinguished Lie subalgebras isomorphic to $\so(\cc^4)$: in an obvious notation,
 \begin{gather*}\so_{12}\big(\cc^4\big),\quad \so_{13}\big(\cc^4\big),\quad \so_{23}\big(\cc^4\big).\end{gather*}
 By (\ref{casimir2}), the corresponding Casimir operators are
\begin{gather*}
\cC_{12} =4B_{1,2}B_{-1,-2}-(N_1+N_2)^2-2N_1-2N_2
 =4B_{-1,-2}B_{1,2}-(N_1+N_2)^2+2N_1+2N_2\\
 \hphantom{\cC_{12}}{} =4B_{1,-2}B_{-1,2}-(N_1-N_2)^2-2N_1+2N_2
 =4B_{-1,2}B_{1,-2}-(N_1-N_2)^2+2N_1-2N_2,\\
\cC_{13} =4B_{1,3}B_{-1,-3}-(N_1+N_3)^2-2N_1-2N_3
 =4B_{-1,-3}B_{1,3}-(N_1+N_3)^2+2N_1+2N_3\\
 \hphantom{\cC_{13}}{} =4B_{1,-3}B_{-1,3}-(N_1-N_2)^2-2N_1+2N_3
 =4B_{-1,3}B_{1,-3}-(N_1-N_3)^2+2N_1-2N_3,\\
\cC_{23} =4B_{2,3}B_{-2,-3}-(N_2+N_3)^2-2N_2-2N_3 =4B_{-2,-3}B_{2,3}-(N_2+N_3)^2+2N_2+2N_3\\
\hphantom{\cC_{23}}{}
 =4B_{2,-3}B_{-2,3}-(N_2-N_2)^2-2N_2+2N_3 =4B_{-2,3}B_{2,-3}-(N_2-N_3)^2+2N_2-2N_3.
\end{gather*}
After the reduction described in (\ref{deq1a}), we obtain the identities
\begin{subequations}
\begin{gather}\label{facto1}
(2z_{-1}z_1+2z_{-2}z_2)\Delta_{\cc^6}^\diamond =-1+\cC_{12}^{\diamond,-1}+\big(N_3^{\diamond,-1}\big)^2,\\
\label{facto2}(2z_{-1}z_1+2z_{-3}z_3)\Delta_{\cc^6}^\diamond =-1+\cC_{13}^{\diamond,-1}+\big(N_2^{\diamond,-1}\big)^2,\\\label{facto3}
(2z_{-2}z_2+2z_{-3}z_3)\Delta_{\cc^6}^\diamond =-1+\cC_{23}^{\diamond,-1}+\big(N_1^{\diamond,-1}\big)^2.
\end{gather}
\end{subequations}

If we use the spherical section, (\ref{facto1}), (\ref{facto2}), (\ref{facto3}) become
\begin{subequations}
\begin{gather}\label{facto1b}
4\Delta_{\cc^6}^\sph =-1+\cC_{12}^{\sph,-1}+\big(N_3^{\sph,-1}\big)^2,\\
\label{facto2b} -4(1-w)\Delta_{\cc^6}^\sph =-1+\cC_{13}^{\sph,-1}+\big(N_2^{\sph,-1}\big)^2,\\
\label{facto3b} -4w\Delta_{\cc^6}^\sph=-1+\cC_{23}^{\sph,-1}+\big(N_1^{\sph,-1}\big)^2.
\end{gather}
\end{subequations}
They yield the factorizations of the balanced hypergeometric operator described in Section~\ref{subs-1}. Applying \eqref{change1}, we also obtain the factorizations of the standard hypergeometric operator described in~\cite{De}.

\subsection[Conformal symmetries of $\Delta_{\cc^4}$]{Conformal symmetries of $\boldsymbol{\Delta_{\cc^4}}$}\label{sec-c4}

In this subsection we describe the reduction of the Laplacian on $\cc^6$ to $\cc^4$, which is accomplished by aplying the f\/lat section. This will lead us to an alternative derivation of the hypergeometric equation. Besides, the material of this subsection will be needed when we will discuss the conf\/luent equation.

To a large extent, this subsection is a specif\/ication of Section~\ref{subsec-conf} to $n=4$. Recall that the f\/lat section is given by
\begin{gather*} z_{-3}=-z_{-1}z_1-z_{-2}z_2,\qquad z_3=1. \end{gather*}
It is parametrized with $y\in\cc^4$. More precisely, we introduce the coordinates
\begin{gather}
y_{-1}=z_{-1},\qquad y_1=z_1,\qquad y_{-2}=z_{-2},\qquad y_2=z_2.\label{sq4}
\end{gather}
Thus this section can be identif\/ied with $\cc^4$ with the scalar product given by the square
\begin{gather} \langle y|y\rangle=2y_{-1}y_1+2y_{-2}y_2.\label{sq5}\end{gather}

{\bf Lie algebra $\so(\cc^6)$.} Cartan algebra:
\begin{gather*}
N_i^{\fl} =y_{-i}\partial_{y_{-i}}-y_{i}\partial_{y_{i}},\qquad i=1,2,\\
N_3^{\fl,\eta } =y_{-1}\partial_{y_{-1}}+y_{1}\partial_{y_{1}}+y_{-2}\partial_{y_{-2}} +y_{2}\partial_{y_{2}}-\eta.
\end{gather*}
Root operators:
\begin{gather*}
B_{1,2}^\fl =y_{-1}\p_{y_2}-y_{-2}\p_{y_1},\qquad \  \, B_{-1,-2}^\fl =y_{1}\p_{y_{-2}}-y_{2}\p_{y_{-1}},\\
B_{1,-2}^\fl =y_{-1}\p_{y_{-2}}-y_{2}\p_{y_1},\qquad
B_{-1,2}^\fl =y_{1}\p_{y_2}-y_{-2}\p_{y_{-1}},\\
B_{3,1}^{\fl,\eta} =y_{-1}(\partial_{y_{-1}}-\eta)-y_{-2}y_2\partial_{y_1} +y_{-1}y_{-2}\partial_{y_{-2}}+y_{-1}y_{2}\partial_{y_{2}},\\
B_{-3,-1}^{\fl} =\p_{y_{-1}},\qquad B_{3,-1}^{\fl,\eta} =y_{1}(\partial_{y_{1}}-\eta)-y_{-2}y_2\partial_{y_{-1}}
+y_{1}y_{-2}\partial_{y_{-2}}+ y_{1}y_{2}\partial_{y_{2}},\\
B_{-3,1}^{\fl} =\p_{y_1},\qquad\quad \,
B_{3,2}^{\fl,\eta} =y_{-2}(\partial_{y_{-2}}-\eta)-y_{-1}y_1\partial_{y_2}+y_{-2}y_{-1}\partial_{y_{-1}}+y_{-2}y_{1}\partial_{y_{1}},\\
B_{-3,-2}^{\fl} =\p_{y_{-2}},\qquad B_{3,-2}^{\fl,\eta} =y_{2}(\partial_{y_{2}}-\eta)-y_{-1}y_1\partial_{y_{-2}}
+y_{2}y_{-1}\partial_{y_{-1}}+y_{2}y_{1}\partial_{y_{1}},\\
B_{-3,2}^{\fl} =\p_{y_2}.
\end{gather*}

 {\bf Weyl symmetries.} Transpositions:
\begin{gather*}
\sigma_{(12)}^\fl f(y_{-1},y_1,y_{-2},y_2) = f(y_{-2},y_2,y_{-1},y_1),\\
\sigma_{(13)}^{\fl,\eta }f(y_{-1},y_1,y_{-2},y_2)= y_1^{\eta }f\left(\frac{-y_{-1}y_1-y_{-2}y_2}{y_1},
\frac{1}{y_1},\frac{y_{-2}}{y_1},\frac{y_2}{y_1}\right),\\
\sigma_{(23)}^{\fl,\eta }f(y_{-1},y_1,y_{-2},y_2) = y_2^{\eta }f\left(\frac{y_{-1}}{y_2},\frac{y_{1}}{y_2},
\frac{-y_{-1}y_1-y_{-2}y_2}{y_2},\frac{1}{y_2}\right).
\end{gather*}
Cycles:
\begin{gather*}
\sigma_{(123)}^{\fl,\eta } f(y_{-1},y_1,y_{-2},y_2)=
y_2^{\eta }f\left(\frac{-y_{-1}y_1-y_{-2}y_2}{y_2},
\frac{1}{y_2},\frac{y_{-1}}{y_2},\frac{y_1}{y_2}\right),\\
\sigma_{(132)}^{\fl,\eta }f(y_{-1},y_1,y_{-2},y_2) =
y_1^{\eta }f\left(\frac{y_{-2}}{y_1},\frac{y_2}{y_1},
\frac{-y_{-1}y_1-y_{-2}y_2}{y_1},\frac{1}{y_1}\right).
\end{gather*}
Flips:
\begin{gather*}
\tau_1^{\fl}f(y_{-1},y_1,y_{-2},y_2)=f(y_{1},y_{-1},y_{-2},y_2),\\
\tau_2^{\fl}f(y_{-1},y_1,y_{-2},y_2)=f(y_{-1},y_{1},y_{2},y_{-2}),\\
\tau_3^{\fl,\eta }f(y_{-1},y_1,y_{-2},y_2)=(-2y_{-1}y_1-2y_{-2}y_2)^{\eta }f\left(\frac{y_{-1},y_{1},y_{-2},y_2}{-y_{-1}y_1-y_{-2}y_2}\right).
\end{gather*}

{\bf Reduced Laplacian} coincides with the 4-dimensional Laplacian:
\begin{gather*}\Delta_{\cc^6}^\fl=\Delta_{\cc^4}=2\partial_{y_{-1}}\partial_{y_{1}}+ 2\partial_{y_{-2}}\partial_{y_{2}}.\end{gather*}

{\bf Generalized symmetries:}
\begin{alignat*}{3}
& N_i^{\fl,-3}\Delta_{\cc^4} =\Delta_{\cc^4} N_i^{\fl,-1},\qquad &&1 \leq i\leq 3, & \\ 
& B_{i,j}^{\fl,-3}\Delta_{\cc^4} =\Delta_{\cc^4} B_{i,j}^{\fl,-1},\qquad && 1 \leq |i|<|j|\leq3, & \\ 
& \sigma^{\fl,-3}\Delta_{\cc^4} =\Delta_{\cc^4}\sigma^{\fl,-1},\qquad &&\sigma \in S_3, & \\ 
& \tau_j^{\fl,-3}\Delta_{\cc^4} =\Delta_{\cc^4}\tau_j^{\fl,-1},\qquad&&1 \leq j\leq3. & 
\end{alignat*}

\subsection[Deriving balanced hypergeometric operator from $\Delta_{\cc^4}$]{Deriving balanced hypergeometric operator from $\boldsymbol{\Delta_{\cc^4}}$}

Introduce the following coordinates in $\cc^4$:
\begin{gather}
w =\frac{y_{-1}y_1}{y_{-1}y_1+y_{-2}y_{2}},\qquad\!
r =\sqrt{2\left(y_{-1}y_1+y_{-2}y_2\right)},\qquad\!
u_1 =\sqrt\frac{y_{-1}}{y_{1}},\qquad\!
u_2 =\sqrt\frac{y_{-2}}{y_{2}}.\!\!\! \label{coor1}
\end{gather}
We check that
\begin{gather*}
N_1^\fl =u_1\p_{u_1},\qquad N_2^\fl =u_2\p_{u_2},\qquad N_3^{\fl,\eta} =r\p_r-\eta,\\[1mm]
\Delta_{\cc^4} =\frac{1}{r^2}\left((r\ddr+1)^2-1 +4\ddw w(1-w)\ddw-\frac{(u_1\ddu{1})^2}{w}-\frac{(u_2\ddu{2})^2}{(1-w)}\right).
\end{gather*}
Thus the ansatz
\begin{gather*} f(w,u_1,u_2,r)=u_1^\alpha u_2^\beta r^{\mu-1}F(w) \end{gather*}
leads to the balanced hypergeometric operator.

\subsection[Deriving standard hypergeometric operator from $\Delta_{\cc^4}$]{Deriving standard hypergeometric operator from $\boldsymbol{\Delta_{\cc^4}}$}

Alternatively, we can slightly change the coordinates (\ref{coor1}), replacing $u_1$, $u_2$ with
\begin{gather*} \tilde u_1:=
\frac{y_{-1}}{\sqrt{y_{-1}y_1+y_{-2}y_2}} = u_1\sqrt{w},\qquad \tilde u_2:=\frac{y_{-2}}{\sqrt{y_{-1}y_1+y_{-2}y_2}} = u_2\sqrt{1-w}.
\end{gather*}
We check that
\begin{gather*}
N_1^\fl =\tilde u_1\p_{\tilde u_1},\qquad N_2^\fl =\tilde u_2\p_{\tilde u_2},\qquad N_3^{\fl,\eta} =r\p_r-\eta, \\
\Delta_{\cc^4} =\frac{1}{r^2}\Big((r\ddr+1)^2-1+4w(1-w)\p_w^2 \\
\hphantom{\Delta_{\cc^4} =}{} +4\big((1+u_1\p_{u_1})(1-w)-(1+u_2\p_{u_2} )w\big)\p_w -(u_1\p_{u_1}+u_2\p_{u_2} +1)^2\Big). 
\end{gather*}
Thus the ansatz
\begin{gather*}f(w,\tilde u_1,\tilde u_2,r)=\tilde u_1^\alpha \tilde u_2^\beta r^{\mu-1}F(w)
\end{gather*} leads to the standard hypergeometric equation.

\section[$\so(\cc^5)$ and the Gegenbauer equation]{$\boldsymbol{\so(\cc^5)}$ and the Gegenbauer equation}\label{s7}

In this section we derive the Gegenbauer operator and its $\so(\cc^5)$ symmetries. The whole section is very similar to Section~\ref{s6}, where we derived the hypergeometric operator with its $\so(\cc^6)$ symmetries. The main dif\/ference is lower dimension.

We will consider the following levels:
\begin{enumerate}\itemsep=0pt
\item[(1)] extended space $\cc^5$ and the Laplacian $\Delta_{\cc^5}$,
\item[(2)] reduction to the so-called spherical section and the corresponding Laplace--Beltrami ope\-ra\-tor,
\item[(3)] depending on the choice of coordinates, separation of variables leads to the balanced or standard Gegenbauer operator.
\end{enumerate}

There exists an alternative derivation:
\begin{enumerate}\itemsep=0pt
\item[$(2)'$] reduction to $\cc^3$ and $\Delta_{\cc^3}$ with help of the f\/lat section,
\item[$(3)'$] with appropriate coordinates, separation of variables leads to the balanced or standard Gegenbauer operator.
\end{enumerate}

Some of the aspects of the Gegenbauer equation are actually more complicated than the corresponding aspects of the hypergeometric equation. This is seen, in particular, when we consider factorizations of the Gegenbauer operator, which come in two separate varieties, unlike for the hypergeometric operator, which has a~single variety of factorizations. This corresponds to the fact that $\so(6)$ is simply-laced, whereas $\so(5)$ is not, i.e., its root operators are not of equal length.

\subsection[Extended space $\cc^5$]{Extended space $\boldsymbol{\cc^5}$}\label{sub-geg1}

We consider $\cc^5$ with the coordinates
\begin{gather} z_0,z_{-2},z_2,z_{-3},z_3\label{geg1}\end{gather}
and the scalar product given by
\begin{gather} \langle z|z\rangle=z_0^2+2z_{-2}z_2+2z_{-3}z_3.\label{geg5}\end{gather}
Note that we omit the indices $-1$, $1$; this makes it easier to compare $\cc^5$ with $\cc^6$.

{\bf Lie algebra $\so(\cc^5)$.} Cartan algebra:
\begin{gather*}N_i=z_{-i}\partial_{z_{-i}}-z_{i}\partial_{z_{i}},\qquad i=2,3.\end{gather*}
Root operators:
\begin{gather*}
B_{i,j} =z_{-i}\p_{{j}}-z_{-j}\p_{{i}},\quad |i| =2,\quad |j|=3,\qquad
B_{0,j} =z_0\p_j-z_{-j}\p_0,\quad |j| =2,3.
\end{gather*}

{\bf Weyl symmetries.} Transposition:
\begin{gather*}
\sigma_{(23)}K(z_0,z_{-2},z_2,z_{-3},z_3)=K(z_0,z_{-3},z_3,z_{-2},z_2).
\end{gather*}
Ref\/lection and f\/lips:
\begin{gather*}
\tau_0 K(z_0,z_{-2},z_2,z_{-3},z_3)=K(-z_0,z_{2},z_{-2},z_{-3},z_3),\\
\tau_2 K(z_0,z_{-2},z_2,z_{-3},z_3)=K(z_0,z_{2},z_{-2},z_{-3},z_3),\\
\tau_3 K(z_0,z_{-2},z_2,z_{-3},z_3)=K(z_0,z_{-2},z_{2},z_{3},z_{-3}).
\end{gather*}

{\bf Laplacian:}
\begin{gather*}\Delta_{\cc^5}=\partial_{z_0}^2+2\partial_{z_{-2}}\partial_{z_{2}}+
2\partial_{z_{-3}}\partial_{z_{3}}.\end{gather*}

{\bf Generalized symmetries:}
\begin{alignat*}{3}
& N_i\Delta_{\cc^5} =\Delta_{\cc^5} N_i,\qquad&& i=2,3, &\\
&B_{i,j}\Delta_{\cc^5}=\Delta_{\cc^5} B_{i,j},\qquad &&|i|=2,\quad |j|=3, &\\
&B_{0,j}\Delta_{\cc^5}=\Delta_{\cc^5} B_{0,j},\qquad &&|j|=2,3, &\\
&\sigma_{(23)}\Delta_{\cc^5}=\Delta_{\cc^5}\sigma_{(23)}, \qquad &&&\\
&\tau_j\Delta_{\cc^5}=\Delta_{\cc^5}\tau_j,\qquad&& j=0,2,3.&
\end{alignat*}

\subsection{Spherical section}

We consider the section of the quadric
\begin{gather*}\cV^4:=\big\{z\in\cc^5\colon z_0^2+2z_{-2}z_2+2z_{-3}z_3=0\big\}\end{gather*}given by equations
\begin{gather*}1 =z_0^2+2z_{-2}z_2=-2z_3z_{-3}.\end{gather*}
We will call it the {\em spherical section}, because it is $\cS^2(1)\times\cS^1(-1)$. The superscript used for this section will be ``$\sph$'' for spherical.

Introduce the following coordinates in $\cc^5$:
\begin{gather*}
r =\sqrt{z_0^2+2z_{-2}z_2},\qquad p =\sqrt{2z_3z_{-3}},\\
w =\sqrt{\frac{z_0^2}{2z_{-2}z_2+z_0^2}},\qquad u_2 =\sqrt{\frac{z_{-2}}{z_{2}}},\qquad
u_3 =\sqrt{\frac{z_{-3}}{z_{3}}}\;.\phantom{\sqrt{\frac{1}{2}}}
\end{gather*}
Similarly as in the previous section, the null quadric in these coordinates is given by $r^2+p^2=0$. The generator of dilations is
\begin{gather*}D_{\cc^5}=r\ddr+p\ddp.\end{gather*}
The spherical section is given by the condition $r^2=1$.

Below we describe various objects in the spherical section.

{\bf Lie algebra $\so(\cc^5)$.} Cartan operators:
\begin{gather*}
N_{2}^\sph =u_2\ddu{2},\qquad N_{3}^\sph =u_3\ddu{3}.
\end{gather*}

Root operators:
\begin{gather*}
B_{3,2}^{\sph,{\eta}}=\ii u_2u_3\frac{\sqrt{1-w^2}}{2} \left(\eta-w\ddw-\frac{N_2^\sph}{1-w^2}-N_3^\sph\right),\\[1mm]
B_{-3,-2}^{\sph,{\eta}}=\ii\frac{1}{u_2u_3} \frac{\sqrt{1-w^2}}{2} \left(\eta-w\ddw+\frac{N_2^\sph}{1-w^2}+N_3^\sph\right),\\[1mm]
B_{3,-2}^{\sph,{\eta}}=\ii \frac{u_3}{u_2}\frac{\sqrt{1-w^2}}{2} \left(\eta-w\ddw+\frac{N_2^\sph}{1-w^2}-N_3^\sph\right),\\[1mm]
B_{-3,2}^{\sph,{\eta}}=\ii \frac{u_2}{u_3}\frac{\sqrt{1-w^2}}{2} \left(\eta-w\ddw-\frac{N_2^\sph}{1-w^2}+N_3^\sph\right),\\[1mm]
B_{3,0}^{\sph,{\eta}} = \ii u_3\frac{w}{\sqrt{2}} \left(\eta+\frac{1-w^2}{w}\ddw-N_3^\sph\right),\\[1mm]
B_{-3,0}^{\sph,{\eta}} = \ii \frac{1}{u_3}\frac{w}{\sqrt{2}} \left(\eta+\frac{1-w^2}{w}\ddw+N_3^\sph\right),\\[1mm]
B_{2,0} = u_2\sqrt{\frac{1-w^2}{2}} \left(\ddw+\frac{w}{1-w^2} N_2^\sph \right),\\[1mm]
B_{-2,0} = \frac{1}{u_2}\sqrt{\frac{1-w^2}{2}} \left(\ddw-\frac{w}{1-w^2} N_2^\sph \right).
\end{gather*}

{\bf Weyl symmetries.} Transpositions:
\begin{gather*}
\sigma_{23}^{\sph,{\eta}} f(w,u_2,u_3) = \left(\ii\sqrt{1-w^2}\right)^\eta f\left(\frac{w}{\sqrt{w^2-1}},u_3,u_2\right).
\end{gather*}
Ref\/lection and f\/lips:
\begin{gather*}
\tau_{0}^{\sph} f(w,u_2,u_3) = f(-w,u_2,u_3),\\[1mm]
\tau_{2}^{\sph} f(w,u_2,u_3)= f(w,\frac{1}{u_2},u_3),\qquad \tau_{3}^{\sph} f(w,u_2,u_3)= f(w,u_2,\frac{1}{u_3}).
\end{gather*}

The Laplacian in coordinates is
\begin{gather*}
\Delta_{\cc^5}=\frac{1}{r^2}\left((r\ddr)^2+(r\ddr)+\frac{r^2}{p^2}(p\ddp)^2+\ddw (1-w^2)\ddw-\frac{(u_2\ddu{1})^2}{1-w^2}-\frac{r^2}{p^2}(u_3\ddu{3})^2\right).
\end{gather*}
Using $r^2+p^2=1$ and
\begin{gather*}
(r\ddr)^2+r\ddr+\frac{r^2}{p^2}(p\ddp)^2 =\left(r\ddr+\frac{r^2}{p^2}(p\ddp)+\frac{1}{2}\right)\left(r\ddr+p\ddp+\frac{1}{2}\right)-\frac{1}{4}\\
\hphantom{(r\ddr)^2+r\ddr+\frac{r^2}{p^2}(p\ddp)^2 =}{}
+\left(\frac{r^2}{p^2}+1\right)\left(r\ddr+\frac{1}{2}\right) (p\ddp ).
\end{gather*}
we obtain
\begin{gather*}
\Delta_{\cc^5}^\diamond = \frac{1}{r^2}\left(\ddw \big(1-w^2\big)\ddw-\frac{(u_2\ddu{1})^2}{1-w^2}+(u_3\ddu{3})^2-\frac14\right).
\end{gather*}

To convert $\Delta_{\cc^5}^\diamond$ into the {\it reduced Laplacian} $\Delta_{\cc^5}^\sph$ we simply remove $\frac{1}{r^2}$, obtaining the Laplace--Beltrami operator on $\cS^2(1)\times\cS(1)$:
\begin{gather}
\Delta_{\cc^5}^\sph= \ddw \big(1-w^2\big)\ddw-\frac{\big(N_2^\sph\big)^2}{1-w^2}+\big(N_3^\sph\big)^2-\frac{1}{4}.\label{sph1}
\end{gather}

We have
\begin{alignat}{3}
& N_i^{\sph}\Delta_{\cc^3} =\Delta_{\cc^3}N_i^{\sph},\qquad&& i=2, 3,&\label{yt1}\\
&B_{i,j}^{\sph,-\frac52}\Delta_{\cc^3}=\Delta_{\cc^3}B_{i,j}^{\sph,-\frac12},\qquad&&|i|=2,\quad |j|=3,&\label{yt2}\\
&B_{0,j}^{\sph,-\frac52}\Delta_{\cc^3}=\Delta_{\cc^3} B_{0,j}^{\sph,-\frac12},\qquad&&|j|=2,3,&\label{yt3}\\
&\sigma_{(23)}^{\sph,-\frac52}\Delta_{\cc^3}=\Delta_{\cc^3}\sigma_{(23)}^{\sph,-\frac12},\qquad &&&\label{yt4}\\
&\tau_j^{\sph}\Delta_{\cc^3}=\Delta_{\cc^3}\tau_j^{\sph},\qquad&&j=0,2,3.& \label{yt5}
\end{alignat}

\subsection{Balanced Gegenbauer operator}

Using the spherical section we make an ansatz
\begin{gather} f(w,u_2,u_3)=u_2^\alpha u_3^\lambda F(w).\label{form1}\end{gather}
Clearly,
\begin{gather*}
N_2^\sph f=\alpha f,\qquad N_3^\sph f=\lambda f.
\end{gather*}
Therefore, on functions of the form (\ref{form1}), $\Delta_{\cc^5}^\sph$ (\ref{sph1}) coincides with the balanced Gegenbauer ope\-ra\-tor~(\ref{bal2}). The generalized symmetries for the roots (\ref{yt2}) and (\ref{yt3}), for the permuta\-tion~(\ref{yt4}), and for the f\/lips~(\ref{yt5}) coincide with the transmutation relations, the discrete symmetries, and the sign changes of $\alpha$, $\lambda$ of the balanced Gegenbauer operator, respectively; see Section~\ref{subs-2}.

\subsection{Standard Gegenbauer operator}

Alternatively, we can replace the coordinate $u_2$ with
\begin{gather*} \tilde u_2:=\frac{z_{-2}}{\sqrt{z_0^2{+}2z_{-2}z_2}}= u_2\sqrt{\frac{1-w^2}{2}}.
\end{gather*}
As compared with the previous coordinates, we need to replace $\partial_w$ with
\begin{gather*} \partial_w-\frac{w}{\sqrt{1-w^2}}N_2^\sph.\end{gather*}
In these coordinates
\begin{gather*}
N_{2}^\sph =\tilde u_2\p_{\tilde u_2} ,\qquad N_{3}^\sph =u_3\ddu{3},\\
\Delta_{\cc^5}^\sph =\big(1-w^2\big)\p_w^2-2\big(1+N_2^\sph \big)w\p_w+\big(N_3^\sph\big)^2-\left(N_2^\sph +\frac{1}{2}\right)^2.
\end{gather*}

We make the ansatz
\begin{gather} f(w,u_2,u_3)=\tilde u_2^\alpha u_3^\lambda F(w).\label{form2}\end{gather}
Clearly,
\begin{gather*}
N_2^\sph f=\alpha f,\qquad N_3^\sph f=\lambda f.
\end{gather*}
Therefore, on functions of the form (\ref{form2}), $\Delta_{\cc^5}^\sph$ coincides with the standard Gegenbauer ope\-rator.

\subsection{Factorizations}

In the Lie algebra $\so(\cc^5)$ with the coordinates $z_0$, $z_{-2}$, $z_2$, $z_{-3}$, $z_3$ we have 3 distinguished Lie subalgebras: one isomorphic to $\so(\cc^4)$ and two isomorphic to $\so(\cc^3)$. In an obvious notation,
 \begin{gather*}\so_{23}\big(\cc^4\big),\quad \so_{02}\big(\cc^3\big),\quad \so_{03}\big(\cc^3\big).\end{gather*}
By (\ref{casimir2}) and (\ref{casimir1}), the corresponding Casimir operators are{\samepage
\begin{gather*}
\cC_{23}=4B_{2,3}B_{-2,-3}-(N_2+N_3)^2-2N_2-2N_3
=4B_{-2,-3}B_{2,3}-(N_2+N_3)^2+2N_2+2N_3\\
\hphantom{\cC_{23}}{}=4B_{2,-3}B_{-2,3}-(N_2-N_2)^2-2N_2+2N_3=4B_{-2,3}B_{2,-3}-(N_2-N_3)^2+2N_2-2N_3,\\
\cC_{02}=2B_{0,2}B_{0,-2}-N_2^2-N_2=2B_{0,2}B_{0,-2}-N_2^2+N_2,\\
\cC_{03} =2B_{0,3}B_{0,-3}-N_3^2-N_3=2B_{0,3}B_{0,-3}-N_3^2+N_3.
\end{gather*}}

After the reduction described in (\ref{deq1a}) and (\ref{deq3a}), we obtain the identities{\samepage
\begin{subequations}
\begin{gather}
\label{factor1}(2z_{-2}z_2+2z_{-3}z_3)\Delta_{\cc^5}^\diamond =-\frac34+\cC_{23}^{\diamond,-\frac12},\\ \label{factor2}
\big(z_0^2+2z_{-2}z_2\big)\Delta_{\cc^5}^\diamond =-\frac14+\cC_{02}^{\diamond,-\frac12}+\big(N_3^{\diamond,-\frac12}\big)^2,\\ \label{factor3}
\big(z_0^2+2z_{-3}z_3\big)\Delta_{\cc^5}^\diamond =-\frac14+\cC_{03}^{\diamond,-\frac12}+\big(N_2^{\diamond,-\frac12}\big)^2.
\end{gather}
\end{subequations}}

If we use the spherical section, (\ref{factor1}), (\ref{factor2}), (\ref{factor3}) become
\begin{gather*}
-w^2\Delta_{\cc^5}^\sph =-\frac34+\cC_{23}^{\sph,-\frac12},\\ 
\Delta_{\cc^5}^\sph =-\frac14+\cC_{02}^{\sph,-\frac12}+\big(N_3^{\sph}\big)^2,\\ 
\big(w^2-1\big)\Delta_{\cc^5}^\sph =-\frac14+\cC_{03}^{\sph,-\frac12}+\big(N_2^{\sph}\big)^2.
\end{gather*}
They yield the factorizations of the balanced Gegenbauer operator described in Section~\ref{subs-2} and of the standard Gegenbauer operator described in~\cite{De}.

\subsection[Conformal symmetries of $\Delta_{\cc^3}$]{Conformal symmetries of $\boldsymbol{\Delta_{\cc^3}}$}\label{subsec-c3}

In this subsection we describe the reduction of the Laplacian on~$\cc^5$ to~$\cc^3$. To this end we apply the f\/lat section. This will lead us to an alternative derivation of the Gegenbauer equation. Besides, the material of this subsection will be needed when we will discuss the Hermite equation.

To a large extent, this subsection is a specif\/ication of Section~\ref{subsec-conf} to $n=3$. Recall that the f\/lat section is given by
\begin{gather*} z_{-3}=-\frac12z_0^2-z_{-2}z_2,\qquad z_3=1. \end{gather*}

We introduce the coordinates
\begin{gather}
y_{0}=z_{0},\qquad y_{-2}=z_{-2},\qquad y_2=z_2.\label{geg2}
\end{gather}
Thus we obtain $\cc^3$ with the scalar product given by
\begin{gather} \langle y|y\rangle=y_0^2+2y_{-2}y_2.\label{geg3}\end{gather}

{\bf Lie algebra $\so(\cc^5)$.} Cartan operators:
\begin{gather*}
N_2^{\fl} =y_{-2}\partial_{y_{-2}}-y_{2}\partial_{y_{2}},\\
N_3^{\fl,\eta } =y_0\p_{y_0}+y_{-2}\partial_{y_{-2}}+y_{2}\partial_{y_{2}}-\eta.
\end{gather*}
Root operators:
\begin{gather*}
B_{0,2}^\fl =y_{0}\p_{y_2}-y_{-2}\p_{y_0},\qquad
B_{0,-2}^\fl =y_{0}\p_{y_{-2}}-y_{2}\p_{y_0},\\
B_{0,3}^\fl =y_{0}\p_{y_3}-y_{-3}\p_{y_0},\qquad
B_{0,-3}^\fl=y_{0}\p_{y_{-3}}-y_{3}\p_{y_0},\\
B_{3,2}^{\fl,\eta}=y_{-2}(\partial_{y_{-2}}-\eta)-y_{-2}y_2\partial_{y_2}+y_{-2}y_{-2}\partial_{y_{-2}}+y_{-2}y_{2}\partial_{y_{2}},\qquad
B_{-3,-2}^{\fl} =\p_{y_{-2}},\\
B_{3,-2}^{\fl,\eta}=y_{2}(\partial_{y_{2}}-\eta)-y_{-2}y_2\partial_{y_{-2}}+y_{2}y_{-2}\partial_{y_{-2}}+y_{2}y_{2}\partial_{y_{2}},\qquad\quad \
B_{-3,2}^{\fl} =\p_{y_2}.
\end{gather*}

{\bf Weyl symmetries.} Transpositions:
\begin{gather*}
\sigma_{(23)}^{\fl,\eta }f(y_0,y_{-2},y_2)=y_2^{\eta }f\left(\frac{-\frac12y_0^2-y_{-2}y_2}{y_2},\frac{1}{y_2}\right).
\end{gather*}
Flips:
\begin{gather*}
\tau_0^{\fl} f(y_0,y_{-2},y_2)=f(y_0,y_{2},y_{-2}),\qquad \tau_2^{\fl} f(y_0,y_{-2},y_2)=f(y_0,y_{2},y_{-2}),\\
\tau_3^{\fl,\eta } f(y_0,y_{-2},y_2)=\big({-}y_0^2-2y_{-2}y_2\big)^{\eta } f\left(\frac{y_0,y_{-2},y_{2}}{-\frac12y_0^2-y_{-2}y_2}\right).
\end{gather*}

{\bf Reduced Laplacian} coincides with the 3-dimensional Laplacian:
\begin{gather*}\Delta_{\cc^3}^\fl=\Delta_{\cc^3}=\partial_{y_0}^2+2\partial_{y_{-2}}\partial_{y_{2}}.\end{gather*}

{\bf Generalized symmetries:}
\begin{alignat*}{3}
&N_i^{\fl,-\frac52}\Delta_{\cc^3}=\Delta_{\cc^3} N_i^{\fl,-\frac12},\qquad &&i=2,3, & \\ 
&B_{i,j}^{\fl,-\frac52}\Delta_{\cc^3}=\Delta_{\cc^3} B_{i,j}^{\fl,-\frac12},\qquad&& |i|=2,\quad |j|=3, & \\ 
&B_{0,j}^{\fl,-\frac52}\Delta_{\cc^3}=\Delta_{\cc^3} B_{0,j}^{\fl,-\frac12},\qquad&&|j|=2,3,& \\ 
&\sigma_{(23)}^{\fl,-\frac52}\Delta_{\cc^3}=\Delta_{\cc^3}\sigma_{(23)}^{\fl,-\frac12},\qquad && \\ 
&\tau_j^{\fl,-\frac52}\Delta_{\cc^3}=\Delta_{\cc^3}\tau_j^{\fl,-\frac12},\qquad&& j=0,2,3.& 
\end{alignat*}

\subsection[Deriving balanced Gegenbauer operator from $\Delta_{\cc^3}$]{Deriving balanced Gegenbauer operator from $\boldsymbol{\Delta_{\cc^3}}$}

Introduce the following coordinates in $\cc^3$:
\begin{gather*}
u:=\sqrt{\frac{y_{-2}}{y_2}},\qquad r:=\sqrt{y_0^2+2y_{-2}y_2},\qquad w:=\frac{y_0}{\sqrt{y_0^2+2y_{-2}y_2}}.
\end{gather*}
Clearly,
\begin{gather*}
N_2^\fl =u\p_{u},\qquad N_3^{\fl,\eta} =r\p_{r}-\eta,\\[1mm]
\Delta_{\cc^3}= \frac{1}{r^2} \left( \ddw \big(1-w^2\big)\ddw-\frac{(u\p_u)^2}{1-w^2}+\left(r\p_r+\frac12\right)^2-\frac{1}{4}\right).
\end{gather*}

Thus the ansatz
\begin{gather*}f(w,u,r)=u^\alpha r^{\lambda-\frac12} F(w)\end{gather*}
leads to the balanced Gegenbauer operator (\ref{bal2}).

\subsection[Deriving standard Gegenbauer operator from $\Delta_{\cc^3}$]{Deriving standard Gegenbauer operator from $\boldsymbol{\Delta_{\cc^3}}$}

Instead of the coordinate $u$ choose
\begin{gather*}\tilde u:=\frac{y_{-2}}{\sqrt{y_0^2{+}2y_{-2}y_2}}=u\sqrt{\frac{1-w^2}{2}}.\end{gather*}
Clearly,
\begin{gather*}
N_2^\fl =\tilde u\p_{\tilde u},\qquad N_3^{\fl,\eta} = r\p_{ r}-\eta,\\[1mm]
\Delta_{\cc^3}=\frac{1}{r^2}\left( \big(1-w^2\big)\p_w^2-2\big(1+N_2^\fl \big)w\p_w-\left(\tilde u\p_{\tilde u} +\frac{1}{2}\right)^2 +\left(r\p_r+\frac12\right)^2\right).
\end{gather*}
 Thus the ansatz
\begin{gather*}f(w,\tilde u,r)=\tilde u^\alpha r^{\lambda-\frac12} \tilde F(w)\end{gather*}
leads to the standard Gegenbauer operator (\ref{stan2}).

\section{Symmetries of the heat equation -- the Schr\"odinger algebra}\label{sec-sch}

The main subject of this section are generalized (inf\/initesimal) symmetries of the heat equation
\begin{gather}\big(\Delta_{\cc^{n-2}}+\p_t\big)f=0.\label{heat}\end{gather}
We will see in particular that the Lie group of generalized symmetries of~(\ref{heat}) is $\sch(\cc^{n-2})$, the so-called {\em Schr\"odinger Lie algebra}.

We will reduce the heat equation (\ref{heat}) to the Laplace equation on $\cc^n$ (\ref{lapl}), whose Lie algebra of generalized inf\/initesimal symmetries is, as we saw, $\so(\cc^{n+2})$. $\sch(\cc^{n-2})$ can be viewed as a~subalgebra of~$\so(\cc^{n+2})$.

Note that the choice of the dimension $n-2$ in (\ref{heat}) makes our presentation of the heat equation consistent with that of the Laplace equation of Section~\ref{Conformal invariance}. It will be convenient to start again from the extended space $\cc^{n+2}$, where all symmetries greatly simplify.

\subsection[The Schr\"odinger Lie algebra and group on $\cc^{n+2}$]{The Schr\"odinger Lie algebra and group on $\boldsymbol{\cc^{n+2}}$}

We consider again the space $\cc^{n+2}$ with the scalar product given by
\begin{gather*}\langle z|z\rangle=\sum_{i\in I_{n+2}}z_{-i}z_i,\qquad z\in\cc^{n+2},\end{gather*}
and the Laplacian
\begin{gather*}\Delta_{\cc^{n+2}}=\sum_{i\in I_{n+2}}\p_{z_{-i}}\p_{z_i}.\end{gather*}

Recall that the Lie algebra $\so(\cc^{n+2})$ and the group $\mathrm{O}(\cc^{n+2})$ have natural representations on~$\cc^{n+2}$~(\ref{popo1}) and~(\ref{popo2}) commuting with $\Delta_{\cc^{n+2}}$, see (\ref{syme1a}), (\ref{syme2a}). A special role will be played by the operator
\begin{gather*}
B_{-m-1,m}= z_{m+1}\p_{z_{m}}-z_{-m}\p_{z_{-m-1}}\in\so\big(\cc^{n+2}\big).
\end{gather*}
We def\/ine the {\em Schr\"odinger Lie algebra}
\begin{gather*}
\sch\big(\cc^{n-2}\big):=\big\{B\in \so\big(\cc^{n+2}\big)\colon [B,B_{-m-1,m}]=0\big\}.
\end{gather*}
We also have the {\em full} and {\em special Schr\"odinger group}
\begin{gather*}
\Sch\big(\cc^{n-2}\big):=\big\{{\alpha}\in \mathrm{O}\big(\cc^{n+2}\big)\colon {\alpha}B_{-m-1,m}=B_{-m-1,m}\alpha\big\},\\
\SSch\big(\cc^{n-2}\big):= \Sch\big(\cc^{n-2}\big)\cap \SO\big(\cc^{n+2}\big).
\end{gather*}

\subsection[Structure of $\sch(\cc^{n-2})$]{Structure of $\boldsymbol{\sch(\cc^{n-2})}$}

Let us describe the structure of~$\sch(\cc^{n-2})$.

We will use our usual notation for elements of $\so(\cc^{n+2})$ and $\mathrm{O}(\cc^{n+2})$. In particular,
\begin{gather*}N_{m}=z_{-m}\p_{z_{-m}}-z_{m}\p_{z_{m}},\qquad N_{m+1}=z_{-m-1}\p_{z_{-m-1}}-z_{m+1}\p_{z_{m+1}}.\end{gather*}

Def\/ine
\begin{gather*}N_{m,m+1}:=N_{m}+N_{m+1}.\end{gather*}
Note that $N_{m,m+1}$ belongs to $\sch(\cc^{n-2})$ and commutes with $\so(\cc^{n-2})$, which is naturally embedded in~$\sch(\cc^{n-2})$.

$\sch(\cc^{n-2})$ is spanned by the following operators:
\begin{enumerate}\itemsep=0pt
\item[(1)] $B_{-m-1,m}$, which spans the center of~$\sch(\cc^{n-2})$,
\item[(2)] $B_{m,j}$, $B_{-m-1,j}$, $j=1,\dots,m-1$, which have the following nonzero commutator:
\begin{gather}[B_{m,j},B_{-m-1,-j}]=B_{-m-1,m},\label{cent}\end{gather}
\item[(3)] $B_{-m-1,-m}$, $B_{m+1,m}$, $N_{m,m+1}$, which have the usual commutation relations of $\mathrm{sl}(\cc^2)\simeq \so(\cc^3)$:
\begin{gather*}[B_{m+1,m},B_{-m-1,-m}]=N_{m.m+1},\qquad [B_{\pm(m+1),\pm m},N_{m,m+1}]=\pm B_{\pm(m+1),\pm m},\end{gather*}
\item[(4)] $B_{i,j}$, $|i|<|j|\leq m-1$, $N_i$, $i=1,\dots,m-1,$ with the usual commutation relations of~$\so(\cc^{n-2})$.
\end{enumerate}

The span of (2) can be identif\/ied with $\cc^{n-2}\oplus \cc^{n-2}\simeq \cc^2\otimes \cc^{n-2}$, which has a natural structure of a symplectic space. The span of (1) and (2) is the central extension of the abelian algebra $\cc^2\otimes \cc^{n-2}$ by (\ref{cent}). Such a Lie algebra is usually called the {\em Heisenberg Lie algebra over} $\cc^2\otimes \cc^{n-2}$ and can be denoted by
\begin{gather*}\heis\big(\cc^2\otimes \cc^{n-2}\big)=\cc\rtimes\big(\cc^2\otimes \cc^{n-2}\big).\end{gather*}

$\mathrm{sl}(\cc^2)$ acts in the obvious way on $\cc^2$ and $\so(\cc^{n-2})$ acts on~$\cc^{n-2}$. Thus $\mathrm{sl}(\cc^2)\oplus \so(\cc^{n-2})$ acts on $\cc^2\otimes \cc^{n-2}$. Thus
\begin{gather*}
\sch\big(\cc^{n-2}\big)\simeq\cc\rrtimes\big( \cc^2\otimes\cc^{n-2}\big)\rtimes \big(\mathrm{sl}\big(\cc^2\big)\oplus \so\big(\cc^{n-2}\big)\big).
\end{gather*}

Note, in particular, that neither $\sch(\cc^{n-2})$ nor $\SSch(\cc^{n-2})$ are semisimple.

The subalgebra spanned by the usual Cartan algebra of $\so(\cc^{n-2})$, $N_{m,m+1}$ and $B_{-m-1,m}$ is a maximal commutative subalgebra of $\sch(\cc^{n-2})$. It will be called the {\em ``Cartan algebra'' of $\sch(\cc^{n-2})$.}

Let us introduce $\kappa\in \SO(\cc^{n-2}\oplus\cc^2\oplus\cc^2)$:
\begin{gather*}
\kappa(\dots,z_{-m},z_{m},z_{-m-1},z_{m+1}) :=(\dots,z_{m+1},z_{-m-1},-z_{m},-z_{-m}).
\end{gather*}
Note that $\kappa^4=\id$ and $\kappa\in\SSch(\cc^{n-2})$. On the level of functions
\begin{gather*}
\kappa K(\dots,z_{-m},z_{m},z_{-m-1},z_{m+1}):= K(\dots,-z_{m+1},-z_{-m-1},z_{m},z_{-m}).
\end{gather*}

The subgroup of $\Sch(\cc^{n-2})$ generated by $W(\cc^{n-2})\subset \mathrm{O}(\cc^{n-2})$ and $\kappa$ will be called the group of {\em Weyl symmetries of
$\sch(\cc^{n-2})$.}

\subsection[The Schr\"odinger Lie algebra and group on $\cc^{n}$]{The Schr\"odinger Lie algebra and group on $\boldsymbol{\cc^{n}}$}

Recall that in Section~\ref{subsec-conf} we used the decomposition $\cc^{n+2}=\cc^n\oplus\cc^2$. Elements of $\cc^n$ were generically denoted by~$y$. The space~$\cc^n$ will be also useful in this section. Further on, it will be decomposed as $\cc^n=\cc^{n-2}\oplus\cc^2$. Thus the square of an element of $\cc^n$ is equal to
\begin{gather*}\langle y|y\rangle_{\cc^n} =\langle y|y\rangle_{\cc^{n-2}}+2y_{-m}y_m,\qquad y\in\cc^n,\end{gather*}
and the Laplacian
\begin{gather*}\Delta_{\cc^n}=\Delta_{\cc^{n-2}}+2\p_{y_{-m}}\p_{y_m}.\end{gather*}

Recall that we have the representations
\begin{gather*}
\so(\cc^{n+2})\ni B \mapsto B^{\fl,\eta} \in \cA\rtimes\hol\big(\cc^{n}\big),\\
 \mathrm{O}\big(\cc^{n+2}\big)\ni {\alpha} \mapsto
{\alpha}^{\fl,\eta} \underset{\loc}\in \cA^\times\rtimes\Hol\big(\cc^n\big),
\end{gather*}
and the generalized symmetry
\begin{alignat}{3}
& B^{\fl,\frac{-2-n}{2}} \Delta_{\cc^n}=\Delta_{\cc^n}B^{\fl,\frac{2-n}{2}},\qquad && B\in \so\big(\cc^{n+2}\big),& \label{syme1}\\
& {\alpha}^{\fl,\frac{-2-n}{2}} \Delta_{\cc^n}=\Delta_{\cc^n}{\alpha}^{\fl,\frac{2-n}{2}},\qquad && {\alpha} \in \mathrm{O}\big(\cc^{n+2}\big).& \label{syme2}
\end{alignat}

\subsection[The Schr\"odinger Lie algebra and group on $\cc^{n-2}\oplus\cc$]{The Schr\"odinger Lie algebra and group on $\boldsymbol{\cc^{n-2}\oplus\cc}$}\label{subsec-sch}

We consider now the space $\cc^{n-2}\oplus\cc$ with the generic variables $(y,t)=(\dots,y_{m-1},t)$. Note that~$t$ should be understood as a new name for $y_{-m}$, and we keep the old names for the f\/irst $n-2$ coordinates.

We def\/ine the map $\theta\colon \cA(\cc^{n-2}\oplus\cc)\to \cA(\cc^n)$ by setting for~$h$
\begin{gather*}
(\theta h)(\dots,y_{m-1},y_{-m},y_m):=h(\dots,y_{m-1},y_{-m})\e^{y_m}.
\end{gather*}
We also def\/ine $\zeta\colon \cA(\cc^n)\to \cA(\cc^{n-2}\oplus\cc)$, which to $f$ associates
\begin{gather*}(\zeta f)(\dots,y_{m-1},t):=f(\dots,y_{m-1},t,0).
\end{gather*}
Clearly, $\zeta$ is a left inverse of $\theta$:
\begin{gather*}\zeta\circ\theta=\id.\end{gather*}
Therefore, $\theta\circ\zeta=\id$ is true on the range of $\theta$.

The heat operator in $n-2$ spatial dimensions can be obtained from the Laplacian in $n$ dimension:
\begin{gather}
\cL_{\cc^{n-2}}:=\Delta_{\cc^{n{-}2}}+2\p_{t}=\zeta \Delta_{\cc^n}\theta.\label{conta}\end{gather}

For $B\in \sch(\cc^{n-2})\subset \so(\cc^{n+2})$ we def\/ine
\begin{gather*}B^{\sch,\eta}:=\zeta B^{{\fl},\eta}\theta,\end{gather*}
and for ${\alpha}\in \Sch(\cc^{n-2})\subset \mathrm{O}(\cc^{n+2})$,
\begin{gather*}{\alpha}^{\sch,\eta}:=\zeta {\alpha}^{\fl,\eta}\theta. \end{gather*}

\begin{Lemma} $\sch(\cc^{n-2})$, $\Sch(\cc^{n-2})$ and $\Delta_{\cc^n}$ preserve the range of~$\theta$.
\end{Lemma}

\begin{proof} Note that
\begin{gather*}B_{-m-1,m}^{\fl ,\eta}=\partial_{y_m}.\end{gather*}

Let $B\in \sch(\cc^{n-2})$. Then $[B^{\fl ,\eta},\partial_{y_m}]=0$. Therefore,
\begin{gather*}B^{\fl ,\eta}=C+D\partial_{y_m},\end{gather*}
where $C\in\hol(\cc^{n-2}\oplus\cc)$ and $D\in\cA(\cc^{n-2}\oplus\cc)$ (they do not involve the variable~$y_m$). Therefore, $B$ preserves the range of~$\theta$.

Likewise, if ${\alpha}\in \Sch(\cc^{n-2})$, then we have
\begin{gather*}{\alpha}^{\fl ,\eta}f(\dots,y_{-m},y_m)=\beta f (\dots,y_{-m},y_m+d(\dots,y_{-m})),\end{gather*}
where $\beta\in\cA^\times\rtimes\Hol(\cc^n\oplus\cc)$, $d\in\cA(\cc^{n-2}\oplus\cc)$ (they do not involve the variable~$y_m$). Therefore, ${\alpha}^{\fl ,\eta}$ preserves the range of $\theta$.

For $\Delta_{\cc^n}$ the statement is contained in the formula~(\ref{conta}).
\end{proof}

\begin{Theorem}\label{dod} \quad
\begin{enumerate}\itemsep=0pt
\item[$(1)$] For any $\eta$,
\begin{gather*}
\sch\big(\cc^{n-2}\big)\ni B \mapsto B^{\sch,\eta}\in {\cA}\rrtimes\hol\big(\cc^{n-2}\oplus\cc\big),\\
\Sch\big(\cc^{n-2}\big)\ni {\alpha} \mapsto {\alpha}^{\sch,\eta}\underset{\loc}{\in} \cA^\times\rrtimes\Hol\big(\cc^{n-2}\oplus\cc\big)
\end{gather*}
is a representation/local representation.
\item[$(2)$] We have a generalized symmetry
\begin{alignat*}{3}
& B^{\sch,\frac{-2-n}{2}}\cL_{\cc^{n-2}} =\cL_{\cc^{n-2}} B^{\sch,\frac{2-n}{2}} ,\qquad&& B\in \sch\big(\cc^{n-2}\big),&\\
& {\alpha}^{\sch,\frac{-2-n}{2}}\cL_{\cc^{n-2}} =\cL_{\cc^{n-2}} {\alpha}^{\sch,\frac{2-n}{2}} ,\qquad && {\alpha}\in \Sch\big(\cc^{n-2}\big).&
\end{alignat*}
\end{enumerate}
\end{Theorem}

\begin{proof} Let us f\/irst prove (1). Obviously, for $B_1,B_2\in \so(\cc^{n+2})$
\begin{gather*}\zeta \big[B_1^{\fl ,\eta}B_2^{\fl ,\eta}\big]\theta=\zeta[B_1,B_2]^{\fl,\eta}\theta.
\end{gather*}
If $B_1,B_2\in \sch(\cc^{n-2})$, we can insert $\theta\circ\zeta$ in the middle of the commutator. Hence
\begin{gather*}\big[B_1^{\sch,\eta},B_2^{\sch,\eta}\big]=[B_1,B_2]^{\sch,\eta}.\end{gather*}

Now, for any $\alpha_1,\alpha_2\in\mathrm{O}(\cc^{n+2})$
\begin{gather*}\zeta{\alpha}_1^{\fl,\eta}{\alpha}_2^{\fl,\eta}\theta =\zeta({\alpha}_1{\alpha}_2)^{\fl,\eta}\theta.\end{gather*}
For ${\alpha}_1,{\alpha}_2\in \Sch(\cc^{n-2})$ we can insert $\theta\circ\zeta$ in the middle of the composition, obtaining
\begin{gather*}{\alpha}_1^{\sch,\eta}{\alpha}_2^{\sch,\eta}=({\alpha}_1{\alpha}_2)^{\sch,\eta}.\end{gather*}

To prove (2) we multiply (\ref{syme1}) and (\ref{syme2}) by $\zeta$ from the left and $\theta$ from the right:
\begin{alignat*}{3}
& \zeta B^{\fl ,\frac{-2-n}{2}} \Delta_{\cc^n}\theta =\zeta\Delta_{\cc^n}B^{\fl ,\frac{2-n}{2}}\theta,\qquad &&  B\in \so\big(\cc^{n+2}\big),& \\ 
& \zeta {\alpha}^{\fl ,\frac{-2-n}{2}} \Delta_{\cc^n}\theta= \zeta\Delta_{\cc^n}{\alpha}^{\fl ,\frac{2-n}{2}}\theta,\qquad && {\alpha}\in \mathrm{O} \big(\cc^{n+2}\big).& 
\end{alignat*}
For $B\in\sch(\cc^{n-2})$ and $\alpha\in\Sch(\cc^{n-2})$ we can insert $\theta\circ\zeta$, which yields~(2).
 \end{proof}

\subsection{Hermite operator}

Consider again the space $\cc^{n-2}\oplus\cc$. This time its generic coordinates will be denoted $(w,s)$. We assume that the space $\cc^{n-2}$ is equipped with a~scalar product. The following operator can be called the {\em $(n-2)$-dimensional Hermite operator}:
\begin{gather*}\cH_{\cc^{n-2}}:=\Delta_{\cc^{n-2}}-D_{\cc^{n-2}}+s\p_s.\end{gather*}

The heat operator is closely related to the Hermite operator. Indeed, let us change the coordinates from $(y,t)\in\cc^{n-2}\oplus\cc^{\times}$ to $(w,s)\in \cc^{n-2}\oplus\cc^{\times}$ by
\begin{gather*}
w=t^{-\frac{1}{2}}y,\qquad s=t^{\frac{1}{2}},\end{gather*}
with the inverse transformation
\begin{gather*}y=ws,\qquad t=s^{2}.\end{gather*}
Under this transformation the heat operator $\cL_{\cc^{n-2}}$ becomes $\frac{1}{s^2}\cH_{\cc^{n-2}}$.

In Section~\ref{sec-her} we will use this change of coordinates to obtain the (1-dimensional) Hermite operator. The construction is, however, interesting in higher dimensions as well, therefore we mention it here.

Strictly speaking, the above coordinate change does not work globally: in particular, we need to assume $s\neq0$, $t\neq0$, besides $s$ doubly covers $t$. We usually are not absolutely precise about specifying the domains of coordinate changes~-- if needed, the reader can easily f\/ill in such details.

\subsection{Schr\"odinger symmetries in coordinates}

In this subsection we sum up information about Schr\"odinger symmetries on 4 levels described in the previous subsections. Note that the last two levels dif\/fer only by a change of coordinates. Therefore, the operators on these two levels are denoted by the same symbols, with the same superscript~${}^\sch$.

We start with generic names of the variables and the corresponding squares:
\begin{alignat*}{3}
&z\in\cc^{n+2},\qquad && \langle z|z\rangle_{\cc^{n+2}} =\sum_{j\in I_{n+2}}z_{-j}z_j,&\\
&y\in\cc^n,\qquad&&\langle y|y\rangle_{\cc^n}=\sum_{j\in I_{n}}y_{-j}y_j,&\\
&(y,t)\in\cc^{n-2}\oplus\cc,\qquad&&\langle y|y\rangle_{\cc^{n-2}}=\sum_{j\in I_{n-2}}y_{-j}y_j,&\\
&(w,s)\in\cc^{n-2}\oplus\cc,\qquad&& \langle w|w\rangle_{\cc^{n-2}}=\sum_{j\in I_{n-2}}w_{-j}w_j.&
\end{alignat*}

{\bf Cartan algebra of $\sch(\cc^{n-2})$.} Central element:
\begin{gather*}
B_{-m-1,m}=z_{m+1}\p_{z_{m}}-z_{-m}\p_{z_{-m-1}},\\
B_{-m-1,m}^{{\fl}}=\p_{y_{m}},\qquad B_{-m-1,m}^{\sch}=1,\qquad B_{-m-1,m}^{\sch}=1.
\end{gather*}
Cartan algebra of $\so(\cc^{n-2})$, $j=1,\dots,m-1$:
\begin{alignat*}{3}
& N_j=z_{-j}\p_{z_{-j}}-z_j\p_{z_j},\qquad && N_j^\fl=y_{-j}\p_{y_{-j}}-y_j\p_{y_j},& \\
& N_j^{\sch}=y_{-j}\p_{y_{-j}}-y_j\p_{y_j},\qquad && N_j^{\sch}=w_{-j}\p_{w_{-j}}-w_j\p_{w_j}.&
\end{alignat*}
Generator of scaling:
\begin{gather*}
N_{m,m+1}=z_{-m}\p_{z_{-m}}-z_{m}\p_{z_{m}}+ z_{-m-1}\p_{z_{-m-1}}-z_{m+1}\p_{z_{m+1}},\\
N_{m,m+1}^{\fl,\eta}=\sum\limits_{j\in I_{n-2}}y_j\p_{y_j}+2y_{-m}\p_{y_{-m}}-\eta,\\
N_{m,m+1}^{\sch,\eta}=\sum\limits_{j\in I_{n-2}}y_j\p_{y_j}+2t\p_t-\eta,\qquad
N_{m,m+1}^{\sch,\eta}=s\p_s-\eta.
\end{gather*}

{\bf Root operators of $\sch(\cc^{n-2})$.} Roots of $\so(\cc^{n-2})$, $|i|<|j|$, $i,j\in I_{n-2}$:
\begin{alignat*}{3}
& B_{i,j}=z_{-i}\p_{z_j}-z_{-j}\p_{z_i},\qquad && B_{i,j}^{\fl}=y_{-i}\p_{y_j}-y_{-j}\p_{y_i},&\\
& B_{i,j}^{\sch}=y_{-i}\p_{y_j}-y_{-j}\p_{y_i},\qquad B_{i,j}^{\sch}=w_{-i}\p_{w_j}-w_{-j}\p_{w_i}.&
\end{alignat*}
Space translations, $ j\in I_{n-2}$:
\begin{gather*}
B_{-m-1,j}=z_{m+1}\p_{z_j}-z_{-j}\p_{z_{-m-1}},\\
B_{-m-1,j}^{\fl}=\p_{y_j},\qquad B_{-m-1,j}^{\sch}=\p_{y_j},\qquad B_{-m-1,j}^{\sch}=\frac{1}{s}\p_{w_j}.
\end{gather*}
Time translation:
\begin{gather*}
B_{-m-1,-m}=z_{m+1}\p_{z_{-m}}-z_{m}\p_{z_{-m-1}},\qquad B_{-m-1,-m}^{\fl}=\p_{y_{-m}},\qquad B_{-m-1,-m}^{\sch}=\p_t,\\
B_{-m-1,-m}^{\sch}=\frac{1}{2s^2}\left(-\sum\limits_{j\in I_{n-2}}w_j\p_{w_j}+s\p_s\right).
\end{gather*}
Additional roots, $j\in I_{n-2}$:
\begin{gather*}
B_{m,j} =z_{-m}\p_{z_j}-z_{-j}\p_{z_{m}},\qquad
B_{m,j}^{\fl} =y_{-m}\p_{y_j}-y_{-j}\p_{y_m},\qquad B_{m,j}^{\sch} =t\p_{y_j}-y_{-j},\\
B_{m,j}^{\sch} =s(\p_{w_j}-w_{-j}),\qquad B_{m+1,m}=z_{-m-1}\p_{z_{m}}-z_{-m}\p_{z_{m+1}},\\
B_{m+1,m}^{\fl,\eta}=y_{-m}\left(\sum\limits_{j\in I_{n-2}}y_j\p_{y_j}+y_{-m}\p_{y_{-m}}-\eta \right)-\frac{1}{2} \sum_{j\in I_{n-2}}y_{-j}y_j\p_{y_m},\\
B_{m+1,m}^{\sch,\eta}=t\left(\sum\limits_{j\in I_{n-2}}y_j\p_{y_j}+y_{-m}\p_{y_{-m}}-\eta \right)-\frac{1}{2} \sum_{j\in I_{n-2}}y_{-j}y_j,\\
B_{m+1,m}^{\sch,\eta}=\frac{s^2}{2}\left(s\p_s-2\eta+\sum\limits_{j\in I_{n-2}}w_j\p_{w_j}-\sum\limits_{j\in I_{n-2}}w_j{w_{-j}}\right).
\end{gather*}

{\bf Weyl symmetries.} We will write $K$ for a function on $\cc^{n+2}$, $f$ for a function on $\cc^n$, $h$ for a function on $\cc^{n-2}\oplus \cc$ in both coordinates $\big(\dots,y_{m-1},t\big)$ and $(\dots,w_{m-1},s)$.

Ref\/lection:
\begin{gather*}
\tau_0 K(z_0,\dots, z_{-m},z_m,z_{-m-1},z_{m+1})= K(-z_0,\dots,z_{-m},z_m,z_{-m-1},z_{m+1}),\\
\tau_0^\fl f(y_0,\dots,y_{-m},y_m)= f(-y_0,\dots,y_{-m},y_m),\\
\tau_0^\sch h(y_0,\dots,t)=h(-y_0,\dots,t),\qquad \tau_0^\sch h(w_0,\dots,s)= h(-w_0,\dots,s).
\end{gather*}
Flips, $j=1,\dots,m-1$:
\begin{gather*}
\begin{split}
& \tau_j K(\dots,z_{-j},z_j,\dots,z_{-m},z_m,z_{-m-1},z_{m+1})= K(\dots,z_{j},z_{-j},\dots,z_{-m},z_m,z_{-m-1},z_{m+1}),\\
& \tau_j^\fl f(\dots,y_{-j},y_j,\dots,y_{-m},y_m)= f(\dots,y_{j},y_{-j},\dots,y_{-m},y_m),\\
& \tau_j^\sch h(\dots,y_{-j},y_j,\dots,t)=h(\dots,y_j,y_{-j},\dots,t),\\
& \tau_j^\sch h(\dots,w_{-j},w_j,\dots,s)= h(\dots,w_{j},w_{-j},\dots,s).
\end{split}
\end{gather*}
Permutations, $\sigma\in S_{m-1}$:
\begin{gather*}
\sigma K(\dots,z_{-m+1},z_{m-1},z_{-m},z_m,z_{-m-1},z_{m+1})\\
\qquad{} =K(\dots,z_{-\sigma_{m-1}},z_{\sigma_{m-1}},z_{-m},z_m,z_{-m-1},z_{m+1}),\\
\sigma^\fl f(\dots,y_{-m+1},y_{m-1},y_{-m},y_m) = f(\dots,y_{-\sigma_{m-1}},y_{\sigma_{m-1}},y_{-m},y_m),\\
\sigma^\sch h(\dots,y_{-m+1},y_{m-1},t)=h(\dots,y_{-\sigma_{m-1}},y_{\sigma_{m-1}},t),\\
\sigma^\sch h(\dots,w_{-m+1},w_{m-1},s)= h(\dots,w_{-\sigma_{m{-}1}},w_{\sigma_{m{-}1}},s).
\end{gather*}
Special transformation $\kappa$:
\begin{gather*}
\kappa K(\dots,z_{m-1},z_{-m},z_{m},z_{-m-1},z_{m+1}) = K(\dots,z_{m-1},-z_{m+1},-z_{-m-1},z_{m},z_{-m}),\\
\kappa^{\fl,\eta}f(\dots,y_{m-1},y_{-m},y_m) = y_{-m}^\eta f\left(\dots, \frac{y_{m-1}}{y_{-m}},-\frac{1}{y_{-m}}, \frac{1}{2y_{-m}}\sum_{j\in I_{n}}y_{-j}y_j\right),\\
\kappa^{\sch,\eta}h(\dots,y_{m-1},t) = t^\eta\exp\left(\frac{1}{2t}\sum_{j\in I_{n-2}}y_{-j}y_j\right) h\left(\dots,\frac{y_{m-1}}{t},-\frac{1}{t}\right), \\
 \kappa^{\sch,\eta} h(\dots,w_{m-1},s) =s^{2\eta}\exp\left(\frac{1}{2}\sum_{j\in I_{n-2}}w_{-j}w_j\right) h\left(\dots,-\ii w_{m-1},\frac{\ii}{s}\right).
\end{gather*}
Square of $\kappa$:
\begin{gather*}
\kappa^2 K(\dots,z_{m-1},z_{-m},z_{m},z_{-m-1},z_{m+1}) =K(\dots,z_{m-1},-z_{-m},-z_{m},-z_{-m-1},-z_{m+1}),\\
\big(\kappa^{\fl,\eta}\big)^2 f(\dots,y_{m-1},y_{-m},y_{m})=f(\dots,-y_{m-1},y_{-m},y_{m}),\\
\big(\kappa^{\sch,\eta}\big)^2h(\dots,y_{m-1},t)=(-1)^\eta h(\dots,-y_{m-1},t),\\
\big(\kappa^{\sch,\eta}\big)^2 h(\dots,w_{m-1},s)=(-1)^\eta h(\dots,-w_{m-1},s).
\end{gather*}

{\bf Laplacian / Laplacian / heat operator / Hermite operator:}
\begin{gather*}
\Delta_{\cc^{n+2}} =\sum\limits_{j\in I_{n+2}}\p_{z_{-j}}\p_{z_j},\qquad
\Delta_{\cc^{n}}=\sum\limits_{j\in I_{n}}\p_{y_{-j}}\p_{y_j},\qquad
\cL_{\cc^{n-2}}=\sum\limits_{j\in I_{n-2}}\p_{y_{-j}}\p_{y_j}+2\p_t,\\
\frac{1}{s^2}\cH_{\cc^{n-2}}^{} =\frac{1}{s^2}\left( \sum\limits_{j\in I_{n-2}}\p_{w_{-j}}\p_{w_j}-\sum\limits_{j\in I_{n-2}}w_j\p_{w_j}+s\p_s\right).
\end{gather*}

{\bf Computations.} Let us sketch how we computed the Schr\"odinger Lie algebra and group in coordinates. We set
 \begin{gather*}\Phi^{\sch,\eta}:=\Phi^{\fl ,\eta}\circ\theta, \qquad \Psi^{\sch,\eta}:=\zeta\circ\Psi^{\fl,\eta}.\end{gather*}
Then $\Phi^{\sch,\eta}$, maps $h\in\cA(\cc^{n-2}\oplus\cc)$ onto
\begin{gather*}
\big(\Phi^{\sch,\eta}h\big)(\dots,z_{-m+1},z_{m-1}, z_{-m},z_{m},z_{-m-1},z_{m+1})\\
\qquad :=z_{m+1}^\eta h\left(\dots,\frac{z_{-m+1}}{z_{m+1}}, \frac{z_{m-1}}{z_{m+1}},\frac{z_{-m}}{z_{m+1}}\right) \exp\left(\frac{z_{m}}{z_{m+1}}\right).
\end{gather*}
For $ K\in\cA(\cc^{n-2}\oplus\cc^2\oplus\cc^2)$
\begin{gather*}
\big(\Psi^{\sch,\eta} K\big)(\dots,y_{-m+1},y_{m-1},t):= K\big(\dots,y_{-m+1},y_{m-1},t,0,-\tfrac12\langle y|y\rangle,1\big).
\end{gather*}
Note that
\begin{gather*}
\Psi^{\sch,\eta}\Phi^{\sch,\eta} =\id,\\
\Psi^{\sch,\eta}\Delta_{\cc^{n+2}} \Phi^{\sch,\eta}=\cL_{\cc^{n-2}},\\
\Psi^{\sch,\eta}B\Phi^{\sch,\eta}=B^{\sch,\eta},\qquad B\in\sch\big(\cc^{n-2}\big),\\
\Psi^{\sch,\eta}{\alpha}\Phi^{\sch,\eta} ={\alpha}^{\sch,\eta},\qquad {\alpha}\in \Sch\big(\cc^{n-2}\big).
\end{gather*}

\section[$\sch(\cc^2)$ and the conf\/luent equation]{$\boldsymbol{\sch(\cc^2)}$ and the conf\/luent equation}
\label{sec-con}

In this section we derive the conf\/luent operator and its $\sch(\cc^2)$ symmetries. We will consider the following levels:
\begin{enumerate}\itemsep=0pt
\item[(1)] extended space $\cc^6$ and the Laplacian,
\item[(2)] reduction to $\cc^4$ and the Laplacian,
\item[(3)] reduction to $\cc^2\oplus\cc$ and the heat operator,
\item[(4)] special coordinates,
\item[(5)] sandwiching with a weight,
\item[(6)] depending on the choice of coordinates, separation of variables leads to the balanced or standard conf\/luent operator.
\end{enumerate}

A separate subsection will be devoted to factorizations of the conf\/luent operator.

\subsection[$\cc^6$]{$\boldsymbol{\cc^6}$}

We again consider $\cc^6$ with the coordinates (\ref{sq0}) and the product given by (\ref{sq1}). We describe various object related to the Lie algebra $\sch(\cc^2)$. Remember that $\sch(\cc^2)$ is a subalgebra of~$\so(\cc^6)$ and we keep the notation from~$\so(\cc^6)$.

{\bf Lie algebra $\sch(\cc^2)$.} Cartan algebra is spanned by
\begin{gather*}
N_1 =z_{-1}\p_{z_{-1}}-z_1\p_{z_1},\\
N_{2,3} =z_{-2}\p_{z_{-2}}-z_{2}\p_{z_{2}}+z_{-3}\p_{z_{-3}}-z_{3}\p_{z_{3}},\\
B_{-3,2} =z_3\p_{z_2}-z_{-2}\p_{z_{-3}}.
\end{gather*}
Root operators:
\begin{alignat*}{5}
 & B_{2,-1}=z_{-2}\p_{z_{-1}}-z_1\p_{z_2},\qquad &&
 B_{2,1}=z_{-2}\p_{z_1}-z_{-1}\p_{z_{2}},\qquad&&
 B_{-3,-1}=z_{3}\p_{z_{-1}}-z_1\p_{z_{-3}},&\\
 & B_{-3,1}=z_{3}\p_{z_{1}}-z_{-1}\p_{z_{-3}},\qquad&&
 B_{-3,-2}=z_{3}\p_{z_{-2}}-z_2\p_{z_{-3}},\qquad&&
 B_{3,2}=z_{-3}\p_{z_{2}}-z_{-2}\p_{z_{3}}.&
\end{alignat*}

{\bf Weyl symmetries.} Special symmetry of order $4$:
\begin{gather*}\kappa K(z_{-1},z_1,z_{-2},z_2,z_{-3},z_3)= K(z_{-1},z_1,-z_{3},-z_{-3},z_{2},z_{-2}).\end{gather*}
Flip:
\begin{gather*}\tau_1 K(z_{-1},z_1,z_{-2},z_2,z_{-3},z_3)=K(z_{1},z_{-1},z_{-2},z_2,z_{-3},z_3).\end{gather*}

We also have the Laplacian (\ref{sq2}) satisfying (\ref{sq31})--(\ref{sq34}).

\subsection[$\cc^4$]{$\boldsymbol{\cc^4}$}

We descend on the level of $\cc^4$, with the coordinates (\ref{sq4}) and the scalar product given by~(\ref{sq5}).

{\bf Lie algebra $\sch(\cc^2)$.} Cartan algebra:
\begin{gather*}
N_{2,3}^{\fl ,\eta}=y_{-1}\p_{y_-1}+y_{1}\p_{y_{1}}+2y_{-2}\p_{y_{-2}}- \eta,\qquad
N_1^\fl =y_{-1}\p_{y_-1}-y_{1}\p_{y_{1}},\qquad
B_{-3,2}^\fl=\p_{y_2}.
\end{gather*}
Root operators:
\begin{gather*}
 B_{2,-1}^{\fl}=y_{-2}\p_{y_-1}-y_1\p_{y_2},\qquad
 B_{2,1}^{\fl}=y_{-2}\p_{y_1}-y_1\p_{y_{2}},\qquad
 B_{-3,-1}^{\fl}=\p_{y_{-1}},\qquad
 B_{-3,1}^{\fl}=\p_{y_{1}},\\
 B_{-3,-2}^{\fl}=\p_{y_{-2}},\qquad
 B_{3,2}^{\fl,\eta}=-y_{-1}y_1\p_{y_2}+y_{-2}(y_{-1}\p_{y_{-1}}+y_{1}\p_{y_{1}}+y_{-2}\p_{y_{-2}}-\eta).
\end{gather*}

{\bf Weyl symmetries.} Special symmetry of order $4$:
\begin{gather*}
\kappa^{\fl,\eta}f(y_{-1},y_{1},y_{-2},y_2) =y_{-2}^\eta f\left(\frac{y_1}{y_{-2}},\frac{y_{-1}}{y_{-2}},-\frac1{y_{-2}},
\frac{2y_{-1}y_1+2y_{-2}y_2}{2y_{-2}}\right).\end{gather*}
Flip:
\begin{gather*} {\tau_1^\fl}f(y_{-1},y_{1},y_{-2},y_2)=f(y_1,y_{-1},y_{-2},y_2).\end{gather*}

\subsection[$\cc^2\oplus\cc$]{$\boldsymbol{\cc^2\oplus\cc}$}

We apply the ansatz involving the exponential $\e^{y_2}$. We rename $y_{-2}$ to $t$. The operator $B_{-3,2}^\sch$ becomes equal to $1$, therefore it can be ignored further on.

{\bf Lie algebra $\sch(\cc^2)$.} Cartan algebra:
\begin{gather*}
N_{2,3}^{\sch,\eta}=y_{-1}\p_{y_{-1}}+y_{1}\p_{y_{1}}+2t\p_{t}-\eta,\qquad
N_1^{\sch} =y_{-1}\p_{y_{-1}}-y_{1}\p_{y_{1}}.
\end{gather*}
Root operators:
\begin{gather*}
 B_{2,-1}^{\sch}=t\p_{y_{-1}}-y_1,\qquad B_{2,1}^{\sch}=t\p_{y_1}-y_{-1},\qquad
 B_{-3,-1}^{\sch}=\p_{y_{-1}},\qquad B_{-3,1}^{\sch}=\p_{y_{1}},\\
 B_{-3,-2}^{\sch}=\p_{t},\qquad B_{3,2}^{\sch,\eta}=-y_{-1}y_1+ t(y_{-1}\p_{y_{-1}}+y_{1}\p_{y_{1}}+t\p_{t}-\eta).
\end{gather*}

{\bf Weyl symmetries.} Special symmetry of order $4$:
\begin{gather*}
\kappa^{\sch,\eta}h(y_{-1},y_1,t)=t^\eta\exp\left(\frac{y_{-1}y_1}{t}\right)
h\left(\frac{y_{-1}}{t},\frac{y_1}{t},-\frac1{t}\right).\end{gather*}
Flip:
\begin{gather*}\tau_1^\sch h(y_{-1},y_{1},t)=h(y_1,y_{-1},t).\end{gather*}

{\bf Heat operator:}
\begin{gather*}\cL_{\cc^2}=2\p_{y_{-1}}\p_{y_1}+2\p_{t}.\end{gather*}

{\bf Generalized symmetries:}
\begin{gather}
N_1^{\sch}\cL_{\cc^2} =\cL_{\cc^2} N_1^{\sch},\label{porr1}\\
N_{2,3}^{\sch,-3}\cL_{\cc^2} =\cL_{\cc^2} N_{23}^{\sch,-1},\label{porr2}\\
B_{i,j}^{\sch,-3}\cL_{\cc^2} =\cL_{\cc^2} B_{i,j}^{\sch,-1},\qquad (i, j)=(2,\pm 1),\ (-3,\pm 1),\ \pm(3,2);\label{porr3}\\
\kappa^{\sch,-3}\cL_{\cc^2}=\cL_{\cc^2}\kappa^{\sch,-1},\label{porr4}\\
\tau_1^{\sch}\cL_{\cc^2}=\cL_{\cc^2}\tau_1^{\sch}.\label{porr5}
\end{gather}

\subsection[Coordinates $u$, $w$, $s$]{Coordinates $\boldsymbol{u}$, $\boldsymbol{w}$, $\boldsymbol{s}$}

Let us def\/ine new complex variables as
\begin{gather*}
w =\frac{y_{-1}y_1}{t},\qquad u =\sqrt{\frac{y_{-1}}{y_{1}}} ,\qquad s =\sqrt{t} .
\end{gather*}
Here are the reverse transformations:
\begin{gather*}
y_{-1} =u s \sqrt{w},\qquad y_{1} =\frac{1}{u}s\sqrt{w},\qquad t =s^2.
\end{gather*}

{\bf Lie algebra $\sch(\cc^2)$.} Cartan algebra:
\begin{gather*}
 N_{2,3}^{\sch,{\eta}} =s\dds-\eta,\qquad N_1^\sch =u\ddu.
\end{gather*}
 Root operators:
\begin{alignat*}{3}
&B_{-3,-1}^\sch =\frac{1}{us}\frac{1}{\sqrt{w}}\left(w\ddw+\frac{N_1^\sch}{2}\right),\qquad&&
B_{-3,1}^\sch =\frac{u}{s}\frac{1}{\sqrt{w}}\left(w\ddw-\frac{N_1^\sch}{2}\right),&\\
&B_{2,-1}^\sch =\frac{s}{u}\frac{1}{\sqrt{w}}\left(w\ddw+\frac{N_1^\sch}{2}-w\right),\qquad&&
B_{2,1}^\sch =su\frac{1}{\sqrt{w}}\left(w\ddw-\frac{N_1^\sch}{2}-w\right),&\\
&B_{-3,-2}^{\sch\eta} =\frac{1}{s^2}\left(-w\ddw+\frac{N_{23}^{\sch,{\eta}}}{2}+\frac{\eta}{2}\right),\qquad&&
B_{3,2}^{\sch,{\eta}} =s^2\left(w\ddw+\frac{N_{23}^{\sch,{\eta}}}{2}-\frac{\eta}{2}-w\right).&
\end{alignat*}

{\bf Weyl symmetries.} Special symmetry of order $4$:
\begin{gather*}
\kappa^{\sch,\eta}h(w,u,s) =s^{2\eta}\e^w h\left(-w,u,\frac 1s\right).\end{gather*}
Flip:
\begin{gather*}\tau_1^\sch h (w,u,s )=h\left(w,\frac1{u},s\right).\end{gather*}

{\bf Heat operator:}
\begin{gather*}
\cL_{\cc^2} =\frac{2}{s^2}\left(\ddw w\ddw-w\ddw-\frac{(u\ddu{})^2}{4 w}+\frac{1}{2}s\dds\right).
\end{gather*}

\subsection{Sandwiching with an exponential}

For any operator $C$ we def\/ine
\begin{gather*}\hat C:=\e^{-\frac{w}{2}} C\e^{\frac{w}{2}}.\end{gather*}

The \qt{hat} isomorphism will not change the Cartan operators:
\begin{gather*}
\hat N_{2,3}^{\sch,{\eta}} =s\dds-\eta,\qquad \hat N_1^\sch =u\ddu{}.
\end{gather*}
Root operators:{\samepage
\begin{alignat*}{3}
&\hat{B}_{-3,-1}^\sch
=\frac{1}{us}\frac{1}{\sqrt{w}}\left(w\ddw+\frac{\hat N_1^\sch}{2}+\frac{w}{2}\right),\qquad&&
\hat{B}_{-3,1}^\sch =\frac{u}{s}\frac{1}{\sqrt{w}}\left(w\ddw-\frac{\hat N_1^\sch}{2}+\frac{w}{2}\right),&\\
&\hat{B}_{2,-1}^\sch =\frac{s}{u}\frac{1}{\sqrt{w}}\left(w\ddw+\frac{\hat N_1^\sch}{2}-\frac{w}{2}\right),\qquad&&
\hat{B}_{2,1}^\sch =su\frac{1}{\sqrt{w}}\left(w\ddw-\frac{\hat N_1^\sch}{2}-\frac{w}{2}\right),&\\
&\hat{B}_{-3,-2}^{\sch,\eta} =\frac{1}{s^2}\left(-w\ddw+\frac{\hat N_{2,3}^{\sch,{\eta}}}{2}+\frac{\eta}{2}-\frac{w}{2}\right),\qquad&&
\hat{B}_{3,2}^{\sch,{\eta}} =s^2\left(w\ddw+\frac{\hat N_{2,3}^{\sch,{\eta}}}{2}-\frac{\eta}{2}-\frac{w}{2}\right).&
\end{alignat*}}

{\bf Weyl symmetries.} Special symmetry of order $4$:
\begin{gather*}
\hat\kappa^{\sch,\eta}h(w,u,s) =s^{2\eta} h\left(-w,u,\frac1s\right).\end{gather*}
Flip:
\begin{gather*}\hat\tau_1^\sch h (w,u,s )=h\left(w,\frac1{u},s\right).\end{gather*}

{\bf Heat operator.}
\begin{gather*}
\hat\cL_{\cc^2} =\ee^{-\tfrac{w}{2}} \cL_{\cc^2} \ee^{\tfrac{w}{2}}=\frac{2}{s^2}\left(\ddw w\ddw-\frac{w}{4}-\frac{\big(\hat
 N_1^\sch\big)^2}{4 w}+\frac{1}{2}\hat N_{2,3}^{\sch,-1}\right).
\end{gather*}

\subsection{Balanced conf\/luent operator}

We make an ansatz
\begin{gather} h(w,u,s)=u^\alpha s^{-\theta-1} F(w).\label{ans3}\end{gather}
Clearly,
\begin{gather*}
\hat N_1^\sch h =\alpha h,\qquad \hat N_{2,3}^{\sch,-1} h =-\theta h.
\end{gather*}
Therefore, on functions of this form, $\frac{s^2}{2}\cL_{\cc^2}$ coincides with the balanced conf\/luent opera\-tor~(\ref{bal3}). The generalized symmetries for the roots~(\ref{porr3}), for the special Weyl symmetry~(\ref{porr4}) and for the f\/lip~(\ref{porr5}) coincide with the transmutation relations, the discrete symmetry and the sign changes of $\alpha$, $\theta$ of the balanced conf\/luent operator, respectively; see Section~\ref{subs-3}.

\subsection{Standard conf\/luent operator}

Let us change slightly coordinates by replacing $u$ with
\begin{gather*}
\tilde{u}:=\frac{y_{-1}}{\sqrt{t}}=u\sqrt{w}.
\end{gather*}
The derivative $\partial_w$ is then replaced by
\begin{gather*}\partial_w + \frac{1}{ 2w}N_1.\end{gather*}

Let us make an ansatz
\begin{gather} h(w,\tilde u,s)=\tilde u^\alpha s^{-\theta-1} \tilde F(w).\label{ansatz}\end{gather}
Clearly,
\begin{gather*}
N_1^{\sch} h =\alpha h,\qquad N_{2,3}^{\sch,-1} h =-\theta h.
\end{gather*}
Then, on functions of the form (\ref{ansatz}), $\frac{s^2}{2}\cL_{\cc^2}$ coincides with the standard conf\/luent opera\-tor~(\ref{stan3}).

\subsection{Factorizations}

Let us note the commutation relation
\begin{gather*}
[B_{-3,2},B_{3,2}] = N_2+N_3=N_{2,3}.\end{gather*}
It shows that the triple $B_{-3,2}$, $B_{3,2}$ and $N_{2,3}$ def\/ines a~subalgebra isomorphic to $\so(\cc^3)$, which we will denote $\so_{23}(\cc^3)$. The Casimir operator for~$\so_{23}(\cc^3)$ is
\begin{gather*}
\cC_{23} =4 B_{3,2} B_{-3,-2}-N_{2,3}^2+2N_{2,3} =4 B_{-3,-2} B_{3,2}-N_{2,3}^2-2N_{2,3}.
\end{gather*}
By the same arguments as for (\ref{facto1b}) we obtain
\begin{gather}
-2y_{-1}y_1\cL_{\cc^2}=-1+\cC_{23}^{\sch,-1}+\big(N_1^{\sch,-1}\big)^2.\label{fac-1}
\end{gather}

Moreover, we have
\begin{gather}
 [B_{2,-1},B_{-3,1}] = [B_{2,1} ,B_{-3,-1}] = B_{-3,2}.\label{heis:1}\end{gather}
The commutation relations \eqref{heis:1} def\/ine two Heisenberg subalgebras
\begin{gather*}
\heis_+\big(\cc^2\big) \qquad\text{spanned by}\quad B_{2,-1}, B_{-3,1},\ B_{-3,2},\\
\heis_-\big(\cc^2\big) \qquad\text{spanned by}\quad B_{2,1},\ B_{-3,-1},\ B_{-3,2}.
\end{gather*}
Let us remark that $\heis_+(\cc^2)$ is the $\tau_1$-image of $\heis_-(\cc^2)$.

Let us def\/ine
\begin{gather*}
\cC_{+} =2 B_{2,1} B_{-3,-1}+N_{2,3}+N_1-B_{-3,2} =2 B_{-3,-1} B_{2,1}+N_{2,3}+N_1+B_{-3,2},\\
\cC_{-} =2 B_{2,-1} B_{-3,1}+N_{2,3}-N_1-B_{-3,2} =2 B_{-3,1} B_{2,-1}+N_{2,3}-N_1+B_{-3,2}.
\end{gather*}

$\cC_+$ and $\cC_-$ can be viewed as the Casimir operators for $\heis_+(\cc^2)$ and $\heis_-(\cc^2)$ respectively. Indeed, $\cC_+$, resp.~$\cC_-$
commute with all operators in $\heis_+(\cc^2)$, resp.~$\heis_-(\cc^2)$.

Let us now consider the operators on the level of $\cc^2\oplus\cc$. Direct calculation yields
\begin{gather*}
\cC_+^{\sch,\eta} =2t(\p_{y_{-1}}\p_{y_1}+\p_{t})-\eta-1,\qquad
\cC_-^{\sch,\eta} =2t(\p_{y_{-1}}\p_{y_1}+\p_{t})-\eta-1.
\end{gather*}
Therefore,
\begin{gather}
t\cL_{\cc^2}=\cC_{+}^{\sch,{-1}}=\cC_{-}^{\sch,{-1}}.
\label{fac-2}
\end{gather}

In the variables $w$, $u$, $s$ and after sandwiching with the exponential weight, (\ref{fac-1}) and (\ref{fac-2}) become
\begin{gather*}
-2ws^2\hat\cL_{\cc^2} =-1+\hat\cC_{23}^{\sch,-1}+\big(\hat N_1^{\sch,-1}\big)^2, \qquad 
s^2\hat\cL_{\cc^2} =\hat\cC_{+}^{\sch,{-1}}=\hat\cC_{-}^{\sch,{-1}}.
\end{gather*}
We apply the ansatz (\ref{ans3}) and obtain all the factorizations of the balanced conf\/luent operator of Section~\ref{subs-3}.

\section[$\sch(\cc^1)$ and the Hermite equation]{$\boldsymbol{\sch(\cc^1)}$ and the Hermite equation}
\label{sec-her}

In this section we derive the Hermite operator and its $\sch(\cc^1)$ symmetries. We will consider the following levels:
\begin{enumerate}\itemsep=0pt
\item[(1)] extended space $\cc^5$ and the Laplacian,
\item[(2)] reduction to $\cc^3$ and the Laplacian,
\item[(3)] reduction to $\cc\oplus\cc$ and the heat operator,
\item[(4)] special coordinates,
\item[(5)] sandwiching with a weight,
\item[(6)] separation of variables in (5) leads to the balanced Hermite operator,
\item[(7)] separation of variables in (4) leads to the standard Hermite operator.
\end{enumerate}

\subsection[$\cc^5$]{$\boldsymbol{\cc^5}$}

We again consider $\cc^5$ with the coordinates~(\ref{geg1}) and the product given by the square~(\ref{geg5}). Remember that $\sch(\cc^1)$ is a subalgebra of $\so(\cc^5)$ and we keep the notation from~$\so(\cc^5)$.

{\bf Lie algebra $\sch(\cc^1)$.} The Cartan algebra is spanned by
\begin{gather*}
N_{2,3} =z_{-2}\p_{z_{-2}}-z_{2}\p_{z_{2}}+z_{-3}\p_{z_{-3}}-z_{3}\p_{z_{3}},\qquad
B_{-3,2} =z_3\p_{z_2}-z_{-2}\p_{z_{-3}}.
\end{gather*}
Root operators:
\begin{gather*}
 B_{2,0} =z_{-2}\p_{z_0}-z_0\p_{z_2},\qquad
 B_{-3,0} =z_{3}\p_{z_0}-z_0\p_{z_{-3}},\qquad
 B_{-3,-2} =z_{3}\p_{z_{-2}}-z_2\p_{z_{-3}},\\
 B_{3,2} =z_{-3}\p_{z_2}-z_2\p_{z_{-3}}.
\end{gather*}

{\bf Weyl symmetry:}
 \begin{gather*} \kappa K(z_0,z_{-2},z_2,z_{-3},z_3)=K(z_0,-z_{3},-z_{-3},z_{2},z_{-2}). \end{gather*}
It generates a group isomorphic to $\zz_4$.

\subsection[$\cc^3$]{$\boldsymbol{\cc^3}$}

We descend on the level of $\cc^3$, as described in Section~\ref{subsec-c3}. In particular, we use the coordina\-tes~(\ref{geg2}) with the scalar product given by (\ref{geg3}).

{\bf Lie algebra $\sch(\cc^1)$.} Cartan algebra:
\begin{gather*}
N_{2,3}^{\fl,\eta} =y_{0}\p_{y_0}+2y_{-2}\p_{y_{-2}}-\eta, \qquad B_{-3,2}^\fl=\p_{y_2}.
\end{gather*}
Root operators:
\begin{gather*}
 B_{2,0}^{\fl} =y_{-2}\p_{y_0}-y_0\p_{y_2},\qquad
 B_{-3,0}^{\fl} =\p_{y_0},\qquad
 B_{-3,-2}^{\fl} =\p_{y_{-2}},\\
 B_{3,2}^{\fl,\eta} =y_{-2}y_0\p_{y_0}+y_{-2}^2\p_{y_{-2}}-\frac12y_0^2\p_{y_2}-\eta y_2.
\end{gather*}

{\bf Weyl symmetry:}
\begin{gather*}
\kappa^{\fl,\eta}f(y_0,y_{-2},y_2)=y_2^\eta
f\left(\frac{y_0}{y_2},\frac{y_0^2+2y_{-2}y_2}{2y_2},-\frac1{y_2}\right).
\end{gather*}

\subsection[$\cc\oplus\cc$]{$\boldsymbol{\cc\oplus\cc}$}

We descend onto the level of $\cc\oplus\cc$, as described in Section~\ref{subsec-sch}. $B_{-3,2}$ becomes equal to~$1$, therefore it will be ignored further on. We rename $y_{-2}$ to $t$ and $y_0$ to $y$.

 {\bf Lie algebra $\sch(\cc^1)$.} Cartan algebra:
\begin{gather*}
 N_{2,3}^{\sch,\eta} =y\p_y+2t\p_t-\eta.
 \end{gather*}
Root operators:
\begin{gather*}
 B_{2,0}^{\sch} =t\p_y-y,\qquad B_{-3,0}^{\sch} =\p_y,\qquad B_{-3,-2}^{\sch} =\p_t,\qquad
 B_{3,2}^{\sch,\eta} =t(y\p_y+t\p_t-\eta)-y^2.
\end{gather*}

{\bf Weyl symmetry:}
\begin{gather*}
 \kappa^{\sch,\eta}h(y,t)= t^{\eta}\exp\left(\frac{y^2}{2t}\right) h\left(\frac{y}{t},-\frac{1}{t}\right).\end{gather*}

 {\bf Heat operator:}
\begin{gather*}\cL_{\cc}=\partial_y^2+2\partial_t.\end{gather*}

 {\bf Generalized symmetries:}
\begin{gather}
N_{2,3}^{\sch,-\frac52}\cL_{\cc} =\cL_{\cc} N_{2,3}^{\sch,-\frac12},\label{torr1}\\
B_{i,j}^{\sch,-\frac52}\cL_{\cc} =\cL_{\cc} B_{i,j}^{\sch,-\frac12},\qquad (i, j)=(2, 0),\; ({-3},0),\; \pm(3,2),\label{torr2}\\
 \kappa^{\sch,-\frac52}\cL_{\cc} =\cL_{\cc}\kappa^{\sch,-\frac12}.\label{torr3}
\end{gather}

 \subsection[Coordinates $w$, $s$]{Coordinates $\boldsymbol{w}$, $\boldsymbol{s}$}

Let us def\/ine new complex variables as
\begin{gather*}
w =\frac{y}{\sqrt{2 t}} ,\qquad s =\sqrt{t} .
\end{gather*}
 Reverse transformations are
 \begin{gather*}
y =\sqrt{2} s w ,\qquad t =s^2 .
\end{gather*}

{\bf Lie algebra $\sch(\cc^1)$.} Cartan operators:
\begin{gather*}
N_{2,3}^{\sch,{\eta}} =s\dds-\eta.
\end{gather*}
Root operators:
\begin{gather*}
B_{-3,0}^\sch =\frac{1}{\sqrt{2} s} \ddw,\qquad
B_{2,0}^\sch =\frac{s}{\sqrt{2}} (\ddw-2w ),\\
B_{-3,-2}^\sch =\frac{1}{2 s^2} \big({-}w\ddw+N_{2,3}^{\sph,{\eta}}+\eta\big),\qquad
B_{3,2}^{\sch,{\eta}} =\frac{s^2}{2} \big(w\ddw+N_{2,3}^{\sph,{\eta}}-\eta-2w^2\big).
\end{gather*}
Above $B_{-3,-2}^\sch$ does not depend on $\eta$ even if at f\/irst glance it might seem so.

{\bf Weyl symmetry:}
\begin{gather*}\kappa^{\sch,{\eta}} h(w,s)=s^{2\eta}\e^{w^2}h\big(\ii w,-\tfrac{\ii}{s}\big).\end{gather*}

{\bf Heat operator:}
\begin{gather*}
\cL_{\cc}=\frac{1}{2s^2}\big(\ddw^2-2w\ddw+2s\dds\big).\end{gather*}

\subsection{Sandwiching with a Gaussian}

For any operator $C$ we will write
 \begin{gather}\hat C:=\ee^{-\tfrac{w^2}{2}} C\ee^{\tfrac{w^2}{2}}.\label{sandwich}
\end{gather}

{\bf Lie algebra $\sch(\cc^1)$.} Cartan algebra:
\begin{gather*}
\hat N_{2,3}^{\sch,{\eta}} =s\dds-\eta.
\end{gather*}
Root operators:
\begin{gather*}
\begin{split}
& \hat{B}_{-3,0}^\sch =\frac{1}{\sqrt{2} s} (\ddw+w),\qquad
\hat{B}_{2,0}^\sch =\frac{s}{\sqrt{2}} (\ddw-w ),\\
& \hat{B}_{-3,-2}^\sch =\frac{1}{2 s^2}\big({-}w\ddw+N_{23}^{\sph,{\eta}}+\eta-w^2\big),\qquad
\hat{B}_{3,2}^{\sch,{\eta}} =\frac{s^2}{2}\big(w\ddw+N_{23}^{\sph,{\eta}}-\eta-w^2\big).
\end{split}
\end{gather*}

{\bf Weyl symmetry:}
\begin{gather*}\kappa^{\sch,{\eta}} h(w,s)=s^{2\eta}h\big(\ii w,-\tfrac{\ii}{s}\big).\end{gather*}

{\bf Heat operator:}
\begin{gather*}
\hat\cL_{\cc} =\frac{1}{2s^2}\big(\ddw^2-w^2-2N_{2,3}^{\sch,-\frac12}\big).
\end{gather*}

\subsection{Balanced Hermite operator}

We make an ansatz
\begin{gather} h(w,s)= s^{\lambda-\frac12} F(w).\label{ansatz4}\end{gather}
Clearly,
\begin{gather*} \hat N_{2,3}^{\sch,-\frac12} h=\lambda h. \end{gather*}
Therefore, on functions of this form, $2s^2\hat\cL_{\cc^1}$ coincides with the balanced Hermite opera\-tor~(\ref{bal4}). The generalized symmetries for the roots (\ref{torr2}) and for the Weyl symmetry (\ref{torr3}) coincide with the transmutation relations and the discrete symmetry of the balanced Hermite operator, respectively; see Section~\ref{subs-4}.

\subsection{Standard Hermite operator}

Alternatively, we can use the ansatz
\begin{gather} h(w,s)=s^{\lambda-\frac12} F(w).\label{ansatz1}\end{gather}
without the sandwiching (\ref{sandwich}). Clearly,
\begin{gather*}
N_{2,3}^{\sch,-\frac12} h=\lambda h.
\end{gather*}
Then, on functions of the form (\ref{ansatz1}), $2s^2\cL_{\cc^2}$ coincides with the standard Hermite opera\-tor~(\ref{stan4}).

\subsection{Factorizations}

In $\sch(\cc^1)$ we have a distinguished subalgebra isomorphic to $\so(\cc^3)$
\begin{gather*}
\so_{23}\big(\cc^3\big)\qquad \text{spanned by}\quad B_{-3,-2},\ B_{3,2},\ N_{2,3},
\end{gather*}
and a distinguished Heisenberg algebra
\begin{gather*}
\heis_0\big(\cc^2\big)\qquad \text{spanned by}\quad B_{2,0},\ B_{-3,0},\ B_{-3,2}.
\end{gather*}

We set
\begin{gather*}
\begin{split}
& \cC_{23} =4 B_{3,2} B_{-3,-2}-N_{2,3}^2+2N_{2,3} =4 B_{-3,-2} B_{3,2}-N_{2,3}^2-2N_{2,3},\\
& \cC_0 =2 B_{2,0} B_{-3,0}+2N_{2,3}-B_{-3,2} =2 B_{-3,0} B_{2,0}+2N_{2,3}+B_{-3,2}.
\end{split}
\end{gather*}
$\cC_{23}$ is the Casimir operator of $\so_{23}(\cc^3)$. $\cC_0$ can be treated as the Casimir operator of $\heis_0(\cc^2)$. We have the identities
\begin{gather*}
-y_0^2\cL_{\cc} =\cC_{23}^{\sch,{-\tfrac{1}{2}}}-\frac{3}{4},\qquad
2t\cL_{\cc} =\cC_{0}^{\sch,{-\tfrac{1}{2}}}.
\end{gather*}
In the coordinates $w$, $s$ we can rewrite this as
\begin{gather*}
-w^22s^2\hat\cL_{\cc} =\hat\cC_{23}^{\sch,{-\tfrac{1}{2}}}-\frac{3}{4},
\qquad
2s^2\hat\cL_{\cc} =\hat\cC_{0}^{\sch,{-\tfrac{1}{2}}}.
\end{gather*}
We apply the ansatz (\ref{ansatz4}) and obtain all the factorizations of the balanced Hermite operator of Section~\ref{subs-4}.

\section[$\cc^2\rtimes\so(\cc^2)$ and the ${}_0F_1$ equation]{$\boldsymbol{\cc^2\rtimes\so(\cc^2)}$ and the $\boldsymbol{{}_0F_1}$ equation}\label{sec-bess}

In this section we derive the ${}_0F_1$ operator and its $\cc^2\rtimes\so(\cc^2)$ symmetries from the symmetries of the Helmholtz equation in 2 dimensions. One can argue that this is the simplest case among the f\/ive cases considered in this paper, because only {\em true} (that is, not generalized) symmetries are used here. This derivation is also extensively discussed in the literature. (Strictly speaking, in the literature usually the Bessel and modif\/ied Bessel equations are considered. They are, however, equivalent to the ${}_0F_1$ equation, as described, e.g., in~\cite{De}.) We included this section for the sake of completeness.

Perhaps, it would be suf\/f\/icient to discuss only two levels of derivation~-- the 2-dimensional Helmholtz equation and the ${}_0F_1$ equation. However, to make this section easier to compare with the previous ones, we will start from a higher level.

Thus, we will consider the following levels:
\begin{enumerate}\itemsep=0pt
\item[(1)] space $\cc^5$ and the Laplacian $\Delta_{\cc^5}$,
\item[(2)] space $\cc^3$ and the Laplacian $\Delta_{\cc^3}$,
\item[(3)] space $\cc^2$ and the Helmholtz operator $\Delta_{\cc^2}-1$,
\item[(4)] choosing appropriate coordinates we obtain the ${}_0F_1$ operator.
\end{enumerate}

\subsection[Space $\cc^5$]{Space $\boldsymbol{\cc^5}$}

As in Section~\ref{sub-geg1}, we consider $\cc^5$ with the coordinates
\begin{gather*} z_0,z_{-2},z_2,z_{-3},z_3
\end{gather*}
and the product given by
\begin{gather*} \langle z|z\rangle=z_0^2+2z_{-2}z_2+2z_{-3}z_3.
\end{gather*}

{\bf Lie algebra $\cc^2\rtimes\so(\cc^2)$ on $\cc^5$.} Cartan operator:
\begin{gather*}N_2=z_{-2}\partial_{z_{-2}}-z_{2}\partial_{z_{2}}.\end{gather*}
Root operators:
\begin{gather*}
B_{-3,-2} =z_{3}\p_{-2}-z_{2}\p_{-3},\qquad B_{-3,2} =z_{3}\p_{2}-z_{-2}\p_{-3}.\end{gather*}

{\bf Weyl symmetry.} Flip:
\begin{gather*}
\tau_2 K(z_0,z_{-2},z_2,z_{-3},z_3)=K(z_0,z_{2},z_{-2},z_{-3},z_3).
\end{gather*}

{\bf Laplacian:}
\begin{gather*}\Delta_{\cc^5}=\partial_{z_0}^2+2\partial_{z_{-2}}\partial_{z_{2}}+2\partial_{z_{-3}}\partial_{z_{3}}.
\end{gather*}

\subsection[Space $\cc^3$]{Space $\boldsymbol{\cc^3}$}\label{subsec-c3a}

As in Section~\ref{subsec-c3}, we consider $\cc^3$ with coordinates $(y_0,y_{-2},y_2)$ and the scalar product given by
\begin{gather*} \langle y|y\rangle=y_0^2+2y_{-2}y_2.
\end{gather*}

{\bf Lie algebra $\cc^2\rtimes\so(\cc^2)$.} Cartan operator:
\begin{gather*} N_2^{\fl}=y_{-2}\partial_{y_{-2}}-y_{2}\partial_{y_{2}}. \end{gather*}

{\bf Root operators:}
\begin{gather*}
B_{-3,-2}^{\fl} =\p_{y_{-2}},\qquad B_{-3,2}^{\fl} =\p_{y_2}.
\end{gather*}

{\bf Weyl symmetry.} Flip:
\begin{gather*}
\tau_2^{\fl} f(y_0,y_{-2},y_2)=f(y_0,y_{2},y_{-2}).
\end{gather*}

{\bf Reduced Laplacian:}
\begin{gather*}
\Delta_{\cc^5}^\fl=\Delta_{\cc^3}=\partial_{y_0}^2+2\p_{y_{-2}}\p_{y_2}.\end{gather*}

\subsection[Space $\cc^2$ and the Helmholtz equation]{Space $\boldsymbol{\cc^2}$ and the Helmholtz equation}

We make an asatz
\begin{gather*}f(y_0,y_-,y_+)=\e^{-y_0}h(y_-,y_+).\end{gather*}
In particular, the coordinates $y_{-2}$, $y_2$ are renamed to $y_-$, $y_+$. We also simplify the names of various operators in an obvious way.

{\bf Lie algebra $\cc^2\rtimes\so(\cc^2)$.} Cartan operator:
\begin{gather*} N=y_{-}\partial_{y_{-}}-y_{+}\partial_{y_+}. \end{gather*}
Root operators:
\begin{gather*}
B_- =\p_{y_-},\qquad
B_+ =\p_{y_+}.
\end{gather*}

{\bf Weyl symmetry.} Flip
\begin{gather*}
\tau f(y_-,y_+)=f(y_+,y_-).
\end{gather*}

{\bf Helmholtz operator:}
\begin{gather*}\cK_{\cc^2}:= -1+2\p_{y_-}\p_{y_+}.\end{gather*}

{\bf Symmetries:}
\begin{gather*}
 N\cK_{\cc^2} =\cK_{\cc^2}N, \qquad B_\pm\cK_{\cc^2} =\cK_{\cc^2}B_\pm, \qquad 
 \tau \cK_{\cc^2} =\cK_{\cc^2}\tau.
 \end{gather*}

\subsection[Balanced ${}_0F_1$ operator]{Balanced $\boldsymbol{{}_0F_1}$ operator}

We introduce the coordinates
\begin{gather}
w =\frac{y_-y_+}{2},\qquad u =\sqrt{\frac{y_-}{y_+}}.\label{coor3}
\end{gather}

{\bf Lie algebra $\cc^2\rtimes\so(\cc^2)$.} Cartan operator:
\begin{gather*}N=u\p_u.\end{gather*}
Root operators:
\begin{gather*}
B_+ =u\frac{1}{\sqrt{2w}}\left(w\ddw-\frac{N_1}{2}\right),\qquad B_- =\frac{1}{u} \frac{1}{\sqrt{2 w}}\left(w\ddw+\frac{N_1}{2}\right).
\end{gather*}

 {\bf Weyl symmetry.} Flip:
\begin{gather*}
\tau h(w,u)=h\left(w,\frac1u\right).
\end{gather*}

{\bf Helmholtz operator:}
\begin{gather*}
\cK_{\cc^2}=\ddw w \ddw-\frac{N^2}{4 w}-1.
\end{gather*}

Making an ansatz
\begin{gather*}h(w,u)=u^\alpha F(w),
\end{gather*}
we obtain the balanced ${}_0F_1$ operator. Symmetries for the root operators and the f\/lip coincide with the transmutation relation and the change of the sign of~$\alpha$ in the balanced ${}_0F_1$ operator, respectively; see Section~\ref{subs-5}.

\subsection[Standard ${}_0F_1$ operator]{Standard $\boldsymbol{{}_0F_1}$ operator}

Modify the coordinates (\ref{coor3}) by replacing $u$ with
\begin{gather*}\tilde u:=y_-=u\sqrt{2w}.\end{gather*}
We then have
\begin{gather*}N= \tilde u\p_{\tilde u},\qquad
\cK_{\cc^2}=w\p_w^2+(1+N)\p_w-1.\end{gather*}
Making an ansatz
\begin{gather*} h(w,u)=\tilde u^\alpha \tilde F(w),\end{gather*}
we obtain the standard ${}_0F_1$ operator.

\subsection{Factorizations}

The factorizations
\begin{gather*}
\cK_{\cc^2}=2B_-B_+-1 =2B_+B_--1
\end{gather*}
are completely obvious. They yield the factorizations of the ${}_0F_1$ operator.

\subsection*{Acknowledgements}

We thank Tom Koornwinder and anonymous referees for useful remarks. J.D.~gratefully acknowledges f\/inancial support of the National Science Center, Poland, under the grant UMO-2014/15/B/ST1/00126.

\LastPageEnding

\end{document}